\documentclass[a4paper,11pt,final]{article}

\usepackage[english]{babel}
\usepackage[utf8]{inputenc}
\usepackage{graphicx}
\usepackage{algorithm}
\usepackage{algorithmic}
\usepackage{mathtools}
\usepackage{amssymb}
\usepackage{mathptmx}
\usepackage{xcolor}
\usepackage{amsmath}
\usepackage{amsthm}
\usepackage{multirow}
\usepackage[margin=2cm]{geometry}

\newtheorem{theorem}{Theorem}
\newtheorem{lemma}[theorem]{Lemma}
\newtheorem{definition}[theorem]{Definition}

\usepackage{lastpage}
\usepackage{fancyhdr}
\usepackage{titling}

\usepackage{url}
\usepackage{doi}
\usepackage[autostyle]{csquotes}
\usepackage[style=authoryear,citestyle=authoryear-comp,maxcitenames=2,maxbibnames=99,backend=biber,backref=true,dashed=false]{biblatex}
\usepackage
	{todonotes}

\usepackage{lineno}
\usepackage{empheq}

\usepackage{tikz}
\usetikzlibrary{arrows}
\usetikzlibrary{positioning,decorations.pathreplacing}
\definecolor{myred}{RGB}{204,0,0}
\definecolor{mygreen}{RGB}{36,143,36}
\definecolor{myblue}{RGB}{0,0,255}
\tikzstyle{nn}=[circle,thick,draw=black!75,minimum size=6mm,fill=white]
\tikzstyle{pp}=[circle,inner sep=1pt,fill=black,minimum size=2mm]
\tikzstyle{ppp}=[circle,inner sep=1pt,fill=black,minimum size=1mm]
\tikzstyle{pppr}=[circle,inner sep=1pt,fill=myred,minimum size=1mm]
\tikzstyle{pppg}=[circle,inner sep=1pt,fill=mygreen,minimum size=1mm]
\tikzstyle{pppb}=[circle,inner sep=1pt,fill=myblue,minimum size=1mm]
\tikzstyle{redlink}=[line width=0.2cm,draw=red!50,line cap=rect]
\tikzstyle{redlinkb}=[line width=0.2cm,draw=red!50,line cap=butt]
\tikzstyle{redlinkd}=[line width=0.2cm,draw=red!50,dashed]
\tikzstyle{bluelink}=[line width=0.2cm,draw=blue!50,line cap=rect]
\tikzstyle{bluelinkb}=[line width=0.2cm,draw=blue!50,line cap=butt]
\tikzstyle{greylink}=[line width=0.2cm,draw=black!40,line cap=rect]

\title{Continuous Average Straightness in Spatial Graphs}
\author{\vspace{-0.2cm}Vincent Labatut\\
	\vspace{-0.2cm}\small Laboratoire Informatique d'Avignon -- LIA EA 4128, Université d'Avignon, France \\
    \small\href{mailto:vincent.labatut@univ-avignon.fr}{\texttt{vincent.labatut@univ-avignon.fr}}}
\date{}

\hypersetup{
    pdftitle={\thetitle},
    bookmarksnumbered=true,bookmarksopen=true,
	unicode=true,colorlinks=true,linktoc=all,
	linkcolor=blue,citecolor=blue,filecolor=blue,urlcolor=blue,
	pdfstartview=FitH
}

\usepackage{enumitem}
\setlist{nolistsep}

\usepackage{adjustbox}

\begin{document}
\maketitle

\begin{abstract}
\addcontentsline{toc}{section}{Abstract}
The Straightness is a measure designed to characterize a pair of vertices in a spatial graph. It is defined as the ratio of the Euclidean distance to the graph distance between these vertices. It is often used as an average, for instance to describe the accessibility of a single vertex relatively to all the other vertices in the graph, or even to summarize the graph as a whole. In some cases, one needs to process the Straightness between not only vertices, but also any other points constituting the graph of interest. Suppose for instance that our graph represents a road network and we do not want to limit ourselves to crossroad-to-crossroad itineraries, but allow any street number to be a starting point or destination. In this situation, the standard approach consists in: 1) discretizing the graph edges, 2) processing the vertex-to-vertex Straightness considering the additional vertices resulting from this discretization, and 3) performing the appropriate average on the obtained values. However, this discrete approximation can be computationally expensive on large graphs, and its precision has not been clearly assessed. In this article, we adopt a continuous approach to average the Straightness over the edges of spatial graphs. This allows us to derive $5$ distinct measures able to characterize precisely the accessibility of the whole graph, as well as individual vertices and edges. Our method is generic and could be applied to other measures designed for spatial graphs. We perform an experimental evaluation of our continuous average Straightness measures, and show how they behave differently from the traditional vertex-to-vertex ones. Moreover, we also study their discrete approximations, and show that our approach is globally less demanding in terms of both processing time and memory usage. Our R source code is publicly available under an open source license.

\vspace{0.3cm}

\vspace{0.3cm}
\noindent \textcolor{red}{\textbf{Cite as:} V. Labatut. \href{https://academic.oup.com/comnet/article-abstract/6/2/269/4090988}{Continuous Average Straightness in Spatial Graphs}, Journal of Complex Networks, 6(2):269-296, 2018. DOI: \href{https://doi.org/10.1093/comnet/cnx033}{10.1093/comnet/cnx033}}
\end{abstract}

\section{Introduction}
\label{sec:Intro}
In a spatial graph, vertices (and consequently edges) hold a position in a metric space \parencite{Barthelemy2011}. Such graphs allow modeling a variety of real-world systems in which spatial constraints affect the topology, and are particularly popular in quantitative geography. For instance, a spatial graph can be used to model the road transportation system of a city, the edges and vertices representing the streets and crossroads, respectively.

When characterizing a spatial graph, it can be interesting, or even necessary, to use this additional spatial information: a number of specific measures have been developed for this purpose \parencite{Barthelemy2011}. In this article, we focus on the \textit{Straightness} \parencite{Vragovic2005,Porta2006}, which characterizes a pair of vertices. It corresponds to the ratio of the \textit{Euclidean} to the \textit{graph} distances between the vertices. In other words, it compares the distance as the crow flies, to the distance obtained when following the shortest path over the graph edges. It quantifies how efficient the graph is in providing the most direct path from one vertex to the other. The Straightness gets close to zero if the path is very tortuous, and it can reach one if it is completely straight. Sometimes, its reciprocal is used instead, under the names \textit{Circuity} \parencite{OSullivan1996}, \textit{Directness} (though \textit{in}directness would be more relevant) \parencite{Hess1997}, \textit{Tortuosity} \parencite{Kansal2001}, \textit{Route Factor} \parencite{Gastner2006a} and \textit{Detour Index} \parencite{Louf2013}.

The Straightness and its reciprocal are used in various ways in the literature. Besides the description of a given pair of vertices, they are also averaged to characterize a single vertex \parencite{Porta2006}, a subgraph \parencite{Kansal2001} or the whole graph \parencite{Kansal2001, Vragovic2005, Gastner2006a, Louf2013}. For a single vertex, authors consider the mean value for all pairs containing the vertex of interest. For a (sub)graph, they average over all pairs of vertices in the (sub)graph. Alternatively, on large graphs, authors average only over a subset of the vertex pairs \parencite{Hess1997, Levinson2009}, primarily to limit the computational cost, but also to focus on certain parts of the graphs (e.g. daily home-office journeys in \parencite{Levinson2009}).

It is important to notice that the relevance of such a discrete average implicitly relies on the assumption that the only considered journeys go from one vertex to another. This is appropriate, for example, for graphs representing subway transportation systems: travelers can only go from one station (i.e. vertex) to another, they cannot get off or on the train anywhere else. However, there are also situations in which the journey could start or end anywhere on an edge, and not necessarily exactly on a vertex. This is for instance the case, on a road network, when the travelers use a personal car or simply walk: they can stop and start anywhere on a street (i.e. an edge), not necessarily at a crossroad (i.e. a vertex). In this case, averaging over pairs of vertices is likely to constitute a very rough approximation of the actual mean Straightness (or Tortuosity). Moreover, on a large network, this processing can be computationally expensive. Working on only a subset of the vertex pairs can solve this problem, but it will increase the imprecision of the already approximate computed value.

In this paper, to solve both precision and computational cost problems, we adopt a continuous approach, by considering all the points constituting the graph instead of just its vertices. We let a journey start and end not only on any vertex, but also on any point lying on an edge. Instead of processing the average Straightness through a sum over pairs of vertices, we integrate the measure over the concerned edges. In other words, we propose a different way of averaging the Straightness, while keeping the same definition of the Straightness itself. The continuous nature of our approach leads to a significantly lower computational cost, compared to a discrete approximation. Our method allows the classic uses of the average Straightness, as a vertex accessibility measure or to characterize a (sub)graph, and we additionally derive a new accessibility measure for individual edges.

The rest of this document is organized as follows. In Section~\ref{sec:DefinitionsNotations}, we introduce the required definitions and notations, and give the sketch of the method we use to derive the continuous average Straightness. In Sections~\ref{sec:StraightnessReformulation} and~\ref{sec:ContinuousAverageStraightness}, we derive $5$ variants of the continuous average Straightness. In Section~\ref{sec:ComplexityApproximation}, we compute their algorithmic complexity, and describe some discrete approximations. Section~\ref{sec:EmpiricalValidation} is dedicated to the empirical validation of our measures. This includes a performance study, comparing their processing time and memory usage to those of the discrete approximations, and an analysis of their behavior on artificial graphs and real-world road networks. Finally, we conclude with some possible extensions in Section~\ref{sec:Conclusion}.

\section{Definitions, Notations and Proof Sketch}
\label{sec:DefinitionsNotations}
We consider an \textit{undirected graph} $G=(V,E)$, where $V$ is the set of vertices and $E \subset V^2$ is the set of undirected edges. 
This graph is \textit{spatial}, so each vertex $v \in V$ is characterized by a \textit{position} $(x_v,y_v)$ in the Euclidean plane. Several distinct vertices can \textit{coincide}, i.e. hold the same position. An undirected \textit{edge} is a straight segment connecting two vertices $u,v \in V$, and is denoted as a pair of \textit{lexicographically ordered} vertices: $(u,v) \in E$.

We define the set of the \textit{graph points} $P$ as the set of all points constituting the graph. It includes the vertices ($V \subset P$) as well as all the points lying on the edges. Each \textit{point} $p \in P$ has a spatial position $(x_p,y_p)$. If the point is a vertex, it can lie on zero to several edges (as an end-vertex), depending on its degree. If it is not a vertex, it lies on \textit{exactly} one edge. Note that several \textit{distinct} non-vertex points can coincide (hold the same spatial position) in case of intersecting edges (vertices can also coincide, as mentioned earlier).

A path between two points is, in general, a sequence of adjacent segments starting from one point and leading to the other. When the path origin and destination are both vertices, as represented in solid red on the left graph of Figure~\ref{fig:BothDistances}, each segment corresponds to an edge. However, if the origin (resp. destination) is not a vertex, then the first (resp. last) segment is only a \textit{portion} of edge, connecting the point to one of the end-vertices of its edge, as shown on the center and right graphs of Figure~\ref{fig:BothDistances}.

We express the spatial distance between two points $p_1$ and $p_2$ through two distinct measures. The first is the classic \textit{Euclidean distance}, noted $d_E(p_1,p_2)$, and represented in dashed red in Figure~\ref{fig:BothDistances}. Based on this distance, we can define the \textit{length of an edge} $(u,v)$, which is the Euclidean distance $d_E(u,v)$ between its end-vertices. The \textit{length of a path} is the sum of the lengths of its constituting segments. The \textit{shortest path} between two points corresponds to the path of minimal length. The second measure is the \textit{weighted graph distance}, or \textit{graph distance} for short, noted $d_G(p_1,p_2)$. It is the length of the shortest path connecting $p_1$ and $p_2$ on the graph. Note that, since this is a Euclidean space, $d_E(p_1,p_2) \leq d_G(p_1,p_2)$. Figure~\ref{fig:BothDistances} shows three shortest paths, in solid red, whose length thus corresponds to the graph distance between the concerned points.

\begin{figure}[ht]
	\centering
	\begin{tikzpicture}
		\node[nn] (v1) at (0,0)    {$v_1$};
		\node[nn] (v2) at (-0.5,3) {$v_2$};
		\node[nn] (v3) at (1,4)    {$v_3$};
		\node[nn] (v4) at (3,2)    {$v_4$};
		\node[nn] (v5) at (2,0.5)  {$v_5$};
 		
		\draw[redlink]  (v1) edge (v2);
		\draw[redlink]  (v2) edge (v3);
		\draw[redlinkd] (v1) edge (v3);
     	
		\node[nn] (v1) at (0,0)    {$v_1$};
		\node[nn] (v2) at (-0.5,3) {$v_2$};
		\node[nn] (v3) at (1,4)    {$v_3$};
		\node[nn] (v4) at (3,2)    {$v_4$};
		\node[nn] (v5) at (2,0.5)  {$v_5$};
 		
		\draw (v1) edge (v2);
		\draw (v1) edge (v5);
		\draw (v2) edge (v3);
		\draw (v3) edge (v4);
		\draw (v4) edge (v5);
	\end{tikzpicture}
	\hfill
	\begin{tikzpicture}
		\node[nn] (v1) at (0,0)    {$v_1$};
		\node[nn] (v2) at (-0.5,3) {$v_2$};
		\node[nn] (v3) at (1,4)    {$v_3$};
		\node[nn] (v4) at (3,2)    {$v_4$};
		\node[nn] (v5) at (2,0.5)  {$v_5$};
		\node[pp] (p1) at (1.75,3.25) {};
 		
		\draw[redlink]  (v1) edge (v2);
		\draw[redlink]  (v2) edge (v3);
		\draw[redlink]  (v3) edge (p1);
		\draw[redlinkd] (v1) edge (p1); 
     	
		\node[nn] (v1) at (0,0)    {$v_1$};
		\node[nn] (v2) at (-0.5,3) {$v_2$};
		\node[nn] (v3) at (1,4)    {$v_3$};
		\node[nn] (v4) at (3,2)    {$v_4$};
		\node[nn] (v5) at (2,0.5)  {$v_5$};
		\node[pp] (p1) at (1.75,3.25) {};
		\node[below left=-0.3 and 0.01cm of p1] {$p_1$};
 		
		\draw (v1) edge (v2);
		\draw (v1) edge (v5);
		\draw (v2) edge (v3);
		\draw (v3) edge (v4);
		\draw (v4) edge (v5);
	\end{tikzpicture}
	\hfill
	\begin{tikzpicture}
		\node[nn] (v1) at (0,0)    {$v_1$};
		\node[nn] (v2) at (-0.5,3) {$v_2$};
		\node[nn] (v3) at (1,4)    {$v_3$};
		\node[nn] (v4) at (3,2)    {$v_4$};
		\node[nn] (v5) at (2,0.5)  {$v_5$};
		\node[pp] (p1) at (1.75,3.25) {};
		\node[pp] (p2) at (-0.2,1.2) {};
 		
		\draw[redlink]  (p2) edge (v2);
		\draw[redlink]  (v2) edge (v3);
		\draw[redlink]  (v3) edge (p1);
		\draw[redlinkd] (p1) edge (p2); 
     	
		\node[nn] (v1) at (0,0)    {$v_1$};
		\node[nn] (v2) at (-0.5,3) {$v_2$};
		\node[nn] (v3) at (1,4)    {$v_3$};
		\node[nn] (v4) at (3,2)    {$v_4$};
		\node[nn] (v5) at (2,0.5)  {$v_5$};
		\node[pp] (p1) at (1.75,3.25) {};
		\node[below left=-0.3 and 0.01cm of p1] {$p_1$};
		\node[pp] (p2) at (-0.2,1.2) {};
		\node[below left=-0.2 and 0.01cm of p2] {$p_2$};
 		
		\draw (v1) edge (v2);
		\draw (v1) edge (v5);
		\draw (v2) edge (v3);
		\draw (v3) edge (v4);
		\draw (v4) edge (v5);
	\end{tikzpicture}
	\caption{Representation of the Euclidean (dashed) and graph (solid) distances (in red) for three cases: between two vertices $v_1$ and $v_3$ (left), between a vertex $v_1$ and a non-vertex point $p_1$ (center) and between two non-vertex points $p_1$ and $p_2$ (right).}
    \label{fig:BothDistances}
\end{figure}
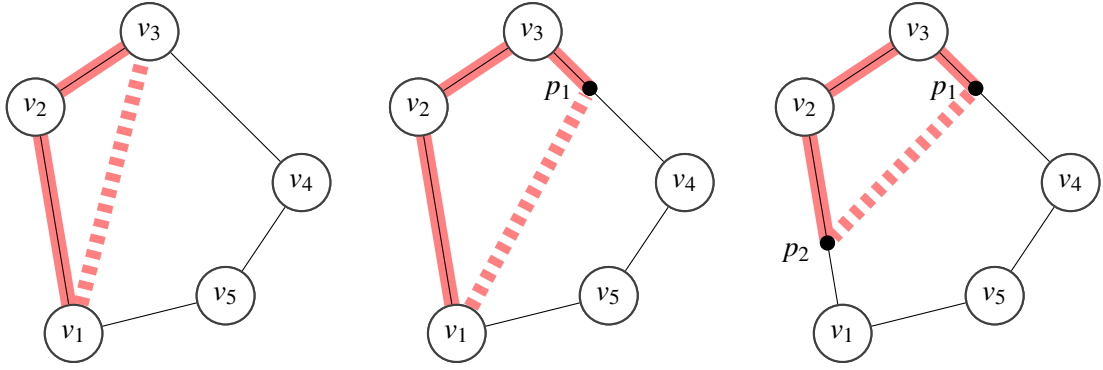

Based on both distances, we can define the Straightness measure. We present here a generalized version of this measure, in the sense that it is defined for a pair of points (and not the traditional pair of vertices).
\begin{definition}[Straightness]
	\label{def:Straightness}
	The \textit{Straightness} $S(p_1,p_2)$ between two points $p_1 $ and $p_2 \in P$ is the ratio of their Euclidean to their graph distances:
	\begin{equation}
	    S(p_1,p_2) = 
	    \begin{cases}
			0 & \text{if }d_G(p_1,p_2) = +\infty, \\
			1 & \text{if }p_1 = p_2, \\
    		\displaystyle \frac{d_E(p_1,p_2)}{d_G(p_1,p_2)} & \text{otherwise.}
		\end{cases}
	\end{equation}
\end{definition}

Since $d_E(p_1,p_2) \leq d_G(p_1,p_2)$, the maximal Straightness value is $1$, which occurs when both distances are equal, i.e. when the shortest path on the graph between $p_1$ and $p_2$ is a straight line. On the contrary, it tends towards $0$ when this path includes more and more detours. The ratio is not defined when $p_1=p_2$: by convention, we set it to $1$ since both distances are equal (to zero). When the points are not connected in the graph, the graph distance is conventionally $d_G(p_1,p_2)= +\infty$, and we set the Straightness to zero.

As mentioned in the introduction, the approach we present in this article consists in averaging a spatial measure through integration over edges instead of summation over vertices. We perform this integration by considering a point $p$ moving along the considered edge. For this purpose, we need to express the position of the point not in terms of its \textit{absolute coordinates} $(x_p,y_p)$, but rather with respect to its edge, using the notion of \textit{relative position}.

\begin{definition}[Relative position]
	\label{def:RelativePosition}
	The position of a point $p \in P$ relatively to its edge $(u,v) \in E$ is :
	\begin{equation}
		\ell_p = d_E(p,u), 
	\end{equation}
	i.e. its Euclidean distance to the first end-vertex of its edge.
\end{definition}

The \textit{first} end-vertex refers to the fact our edges are lexicographically ordered pairs of vertices. The notation $\ell_p$ can be ambiguous if $p$ is a vertex, since it could then lie on several edges. However, in this case the context will allow determining the concerned edge. The left-hand graph in Figure \ref{fig:RelativePosBreakEven} illustrates the notion of relative position.

\begin{figure}[ht]
	\centering
	\begin{tikzpicture}
		\node[nn] (v1) at (0,0)    {$v$};
		\node[nn] (v2) at (-0.5,3) {$u$};
		\node[pp] (p2) at (-0.2,1.2) {};
		\node[below left=-0.2 and 0.01cm of p2] {$p$};
 		
		\draw (v1) edge (v2);
		\draw[dashed] (v1) edge (1,1);
		\draw[dashed] (v1) edge (1,0.4);
		\draw[dashed] (v1) edge (-1.5,0.2);
		\draw[dashed] (v2) edge (-1.5,2.75);
		\draw[dashed] (v2) edge (1,2.5);
		\draw[dashed] (v2) edge (1,4);
		
		\draw[decoration={brace,raise=5pt},decorate] (v2) -- node[right=0.2cm] {$\ell_p$} (p2);
	\end{tikzpicture}
	\hspace{2cm}
	\begin{tikzpicture}
		\node[nn]  (v1) at (0,0)      {};
		\node[nn]  (v2) at (-0.5,3)   {};
		\node[nn]  (v3) at (1,4)      {$u$};
		\node[nn]  (v4) at (3,2)      {$v$};
		\node[nn]  (v5) at (2,0.5)    {};
		\node[ppp] (p1) at (2.6,2.4)  {};
		\node[pp]  (p2) at (-0.2,1.2) {};
 		
		\draw[redlink]   (p2) edge (v2);
		\draw[redlink]   (v2) edge (v3);
		\draw[redlinkb]  (v3) edge (2.6,2.4);
		\draw[bluelink]  (p2) edge (v1);
		\draw[bluelink]  (v1) edge (v5);
		\draw[bluelink]  (v5) edge (v4);
		\draw[bluelinkb] (v4) edge (2.6,2.4);
     	
		\node[nn]  (v1) at (0,0)      {};
		\node[nn]  (v2) at (-0.5,3)   {};
		\node[nn]  (v3) at (1,4)      {$u$};
		\node[nn]  (v4) at (3,2)      {$v$};
		\node[nn]  (v5) at (2,0.5)    {};
		\node[ppp] (p1) at (2.6,2.4)  {};
		\node[below left=-0.1 and 0.01cm of p1] {$q$};
		\node[pp]  (p2) at (-0.2,1.2) {};
		\node[below left=-0.2 and 0.01cm of p2] {$p$};
 		
		\draw (v1) edge (v2);
		\draw (v1) edge (v5);
		\draw (v2) edge (v3);
		\draw (v3) edge (v4);
		\draw (v4) edge (v5);

		\draw[decoration={brace,raise=5pt},decorate] (v3) -- node[above right=0.1cm and 0.2cm] {$\lambda_p$} (p1);
	\end{tikzpicture}
	\caption{Representation of the relative position $\ell_p$ of a point $p$ (left), and of the break-even point $q$ and break-even distance $\lambda_p$ of an edge$(u,v)$ for a point $p$ (right). The blue and red paths have the same length.}
    \label{fig:RelativePosBreakEven}
\end{figure}
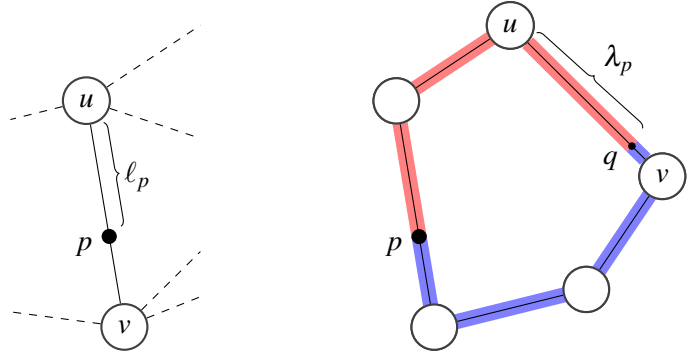

The general idea is then to reformulate the spatial measure as a function of the relative position of the concerned points, which allows integrating the measure over the edges containing these points. This generic approach could be applied to any measure, provided its reformulation is integrable. The Straightness is the ratio of the Euclidean to the graph distances between two points (cf. Definition~\ref{def:Straightness}), so in our case we need to reformulate both distances in terms of relative positions. For the graph distance, this in turn requires introducing the notions of \textit{break-even point} and \textit{break-even distance}.

\begin{definition}[Break-even point and break-even distance]
	\label{def:BreakEvenPoint}
	Consider a point $p \in P$ and an edge $(u,v) \in E$, such that $p$ is either not lying on $(u,v)$, or is one of its end-vertices $u$ or $v$. The break-even point of $(u,v)$ for $p$ is the point $q$ lying on $(u,v)$ at an Euclidean distance $\lambda_p$ from $u$, and such that:
    \begin{equation}
        d_G(p,u) + \lambda_p = d_G(p,v) + d_E(u,v) - \lambda_p.
    \end{equation}
    We call $\lambda_p$, which corresponds to the position of $q$ relatively to its edge (i.e. $\lambda_p=\ell_q$), the break-even distance. Note that it is possible for the break-even point to be $u$ or $v$ (i.e. one of the end-vertices of the edge).
\end{definition}

The right-hand graph in Figure~\ref{fig:RelativePosBreakEven} illustrates the notions of break-even point and break-even distance. By definition, the length of the shortest path between $p$ and the break-even point $q$ is the same whether we consider the paths going through $u$ (represented in red) or $v$ (in blue). 
The relative position $\ell_p$ should not be confused with the break-even distance $\lambda_p$: the former is expressed relatively to the edge containing the considered point $p$, whereas for the latter it is necessarily a different edge.

Note that many measures are based on the notion of distance, for instance the closeness centrality is the reciprocal of the average graph distance between a vertex of interest and the other vertices in the graph \parencite{Bavelas1950}. Their reformulation would therefore also require using the notions of break-even point and distances, which is why we presented them in this general section.

To conclude this part, we will now give the sketch of our method to derive the continuous average Straightness. In Section~\ref{sec:StraightnessReformulation}, we leverage the previous definitions to reformulate the Straightness in terms of relative positions of the concerned vertices. For this purpose, we first consider the Euclidean distance and the graph distance, then proceed with the Straightness between two points, resulting in Theorem \ref{lem:StraightnessRelative}, the main result of that section. If one would like to adapt our method in order to process the continuous average of some other measure, one should replace this step by the reformulation of the targeted measure. In the likely event of a distance-based measure, it is possible to take advantage of our reformulations of the Euclidean (Lemma~\ref{lem:EuclideanDistance}) and/or graph (Lemma \ref{lem:GraphDistancePoints}) distances.

In Section~\ref{sec:ContinuousAverageStraightness}, we explain how $5$ different continuous average variants can be derived based on the reformulation obtained at the previous step. We start by focusing on the average Straightness between a point and an edge (or, more precisely: all the points lying on this edge), which can be obtained by integrating the Straightness over the edge of interest (Definition~\ref{def:AvgStrPtLk}). From this result, it is straightforward to get the average Straightness between a point and the rest of the graph (i.e. all its edges), which can be considered as a vertex accessibility measure (Definition~\ref{def:AvgStrPtGr}). Then, we consider the average Straightness between two edges (or rather, between all their constituting points), which requires integrating twice: once over each one of the considered edges (Definition~\ref{def:AvgStrLkLk}). Based on this result, we can get the average Straightness between an edge and the rest of the graph (i.e. all its edges), which can be interpreted as an edge accessibility measure (Definition~\ref{def:AvgStrLkGr}). Finally, by considering all pairs of edges in the graph, we get the average Straightness between all pairs of points constituting a graph, which can be used to characterize the whole graph (Definition~\ref{def:AvgStrGr}). In this article, we focus on the Straightness, however the second step described in this Section~\ref{sec:ContinuousAverageStraightness} is quite generic and could be applied to any measure reformulated in terms of relative positions beforehand, resulting in the corresponding $5$ continuous average variants.

\section{Reformulation of the Straightness}
\label{sec:StraightnessReformulation}
In this section, we start by expressing the Euclidean and graph distances with respect to the relative positions of the considered points. We then take advantage of these results to reformulate the Straightness, also in relative terms.

We first express the Euclidean distance between two points $p_1$ and $p_2$ in terms of their relative positions $\ell_{p_1}$ and $\ell_{p_2}$.

\begin{lemma}[Euclidean distance between two points]
	\label{lem:EuclideanDistance}
    Let $p_1 \in P$ be a point lying on an edge $(u_1,v_1) \in E$ at a distance $\ell_{p_1}$ from $u_1$, and $p_2 \in P$ be a point lying on an edge $(u_2,v_2) \in E$ at a distance $\ell_{p_2}$ from $u_2$.
    
    The Euclidean distance between $p_1$ and $p_2$ can be written in terms of $\ell_{p_1}$ and $\ell_{p_2}$, as:
	\begin{align}
		\begin{split}
			d_E(p_1,p_2) 
            &= \bigg(\Big(x_{u_2} + \frac{\ell_{p_2}}{d_E(u_2,v_2)}(x_{v_2}-x_{u_2}) - x_{u_1} - \frac{\ell_{p_1}}{d_E(u_1,v_1)}(x_{v_1}-x_{u_1})\Big)^2 \\
            &\quad+ \Big(y_{u_2} + \frac{\ell_{p_2}}{d_E(u_2,v_2)}(y_{v_2}-y_{u_2}) - y_{u_1} - \frac{\ell_{p_1}}{d_E(u_1,v_1)}(y_{v_1}-y_{u_1})\Big)^2 \bigg)^{1/2}.
		\end{split}
		\label{eq:EuclideanRelative}
	\end{align}
\end{lemma}

\begin{proof}[Proof of Lemma~\ref{lem:EuclideanDistance}]
	By definition, $d_E(p_1,p_2) = ((x_{p_2}-x_{p_1})^2 + ((y_{p_2}-y_{p_1})^2)^{1/2}$. Now, let $p \in P$ be a point lying on an edge $(u,v) \in E$, at a distance $\ell_p$ from $u$. This point $p$ divides the edge into two parts such that their lengths are in the ratio $\ell_p:(d_E(u,v)-\ell_p)$. By the \textit{section formula}, the coordinates of $p$ are therefore $x_p = x_u + \ell_{p}/d_E(u,v)(x_v-x_u)$ and $y_p = y_u + \ell_{p}/d_E(u,v)(y_v-y_u)$. Replacing for $p_1$ and $p_2$ in the original expression of $d_E(p_1,p_2)$ yields Lemma~\ref{lem:EuclideanDistance}.
\end{proof}

Expressing the graph distance in terms of relative positions is not as straightforward. Unlike the Euclidean distance, it depends on the graph structure, since it involves the notion of shortest path. It consequently requires to consider $4$ possible cases: the shortest path between $p_1$ and $p_2$ can go through $u_1$ and $u_2$, $u_1$ and $v_2$, $v_1$ and $u_2$, or $v_1$ and $v_2$. In the following, we will first consider a simpler case involving a vertex and a non-vertex point, before dealing with a pair of non-vertex points.

\begin{lemma}[Break-even distance for a vertex]
	\label{lem:BreakEvenDistNode}
    Let $r \in V$ be a vertex and $(u,v) \in E$ an edge such that $r \neq v$ and $r \neq u$. 
    
    The break-even distance of $(u,v)$ for the vertex $r$, noted $\lambda_r$ is:
	\begin{equation}
		\lambda_r = \frac{d_G(r,v) - d_G(r,u) + d_E(u,v)}{2}.
        \label{eq:BreakEvenDistNode}
	\end{equation}
\end{lemma}

Note that since $r$, $u$ and $v$ are all vertices, all three distances appearing in the numerator are obtained directly and do not need further processing (all the vertex-to-vertex distances are supposedly known).

\begin{proof}[Proof of Lemma~\ref{lem:BreakEvenDistNode}]
	Equation (\ref{eq:BreakEvenDistNode}) comes directly from Definition~\ref{def:BreakEvenPoint}.
\end{proof}

\begin{lemma}[Graph distance between a vertex and a point]
	\label{lem:GraphDistanceNodePoint}
	 Let $r \in V$ be a vertex, $(u,v) \in E$ an edge, $p \in P$ a point lying on $(u,v)$ at a distance $\ell_p$ from $u$, and $\lambda_r$ the break-even distance of $(u,v)$ for $r$.
	 
	 The graph distance between the vertex $r$ and the point $p$ can be written as:
	\begin{equation}
		d_G(r,p) =
		\begin{cases}
    		d_G(r,u) + \ell_p,& \text{if } \ell_p \leq \lambda_r \\
		    d_G(r,v) + d_E(u,v) - \ell_p,& \text{otherwise.}
		\end{cases}
	\end{equation}
\end{lemma}

As a consequence of this Lemma: if $\ell_p \leq \lambda_r$, the shortest path between $r$ and $p$ goes through $u$, whereas if $\ell_p \geq \lambda_r$, it goes through $v$, as illustrated in Figure~\ref{fig:GraphDistNdPt}.

\begin{figure}[ht]
	\centering
	\begin{tikzpicture}
		\node[nn] (v1) at (0,0)    {$r$};
		\node[nn] (v2) at (-0.5,3) {};
		\node[nn] (v3) at (1,4)    {$u$};
		\node[nn] (v4) at (3,2)    {$v$};
		\node[nn] (v5) at (2,0.5)  {};
		\node[ppp] (p1) at (1.8,3.2) {};
		\node[pp] (p2) at (1.5,3.5) {};
 		
		\draw[redlink]  (v1) edge (v2);
		\draw[redlink]  (v2) edge (v3);
		\draw[redlink]  (v3) edge (p2);
     	
		\node[nn] (v1) at (0,0)    {$r$};
		\node[nn] (v2) at (-0.5,3) {};
		\node[nn] (v3) at (1,4)    {$u$};
		\node[nn] (v4) at (3,2)    {$v$};
		\node[nn] (v5) at (2,0.5)  {};
		\node[ppp] (p1) at (1.8,3.2) {};
		\node[below right=-0.3 and 0.01cm of p1] {$q$};
		\node[pp] (p2) at (1.5,3.5) {};
		\node[below=-0.05cm of p2] {$p$};
		
		\draw (v1) edge (v2);
		\draw (v1) edge (v5);
		\draw (v2) edge (v3);
		\draw (v3) edge (v4);
		\draw (v4) edge (v5);

		\draw[decoration={brace,raise=5pt},decorate] (v3) -- node[above right=0.1cm and 0.2cm] {$\lambda_r$} (p1);
		\draw[decoration={brace,raise=5pt,mirror},decorate] (v3) -- node[below left=-0.05cm and 0.1cm] {$\ell_p$} (p2);
	\end{tikzpicture}
	\hspace{2cm}
	\begin{tikzpicture}
		\node[nn] (v1) at (0,0)    {$r$};
		\node[nn] (v2) at (-0.5,3) {};
		\node[nn] (v3) at (1,4)    {$u$};
		\node[nn] (v4) at (3,2)    {$v$};
		\node[nn] (v5) at (2,0.5)  {};
		\node[ppp] (p1) at (1.8,3.2) {};
		\node[pp] (p2) at (2.5,2.5) {};
 		
		\draw[redlink]  (v1) edge (v5);
		\draw[redlink]  (v5) edge (v4);
		\draw[redlink]  (v4) edge (p2);
     	
		\node[nn] (v1) at (0,0)    {$r$};
		\node[nn] (v2) at (-0.5,3) {};
		\node[nn] (v3) at (1,4)    {$u$};
		\node[nn] (v4) at (3,2)    {$v$};
		\node[nn] (v5) at (2,0.5)  {};
		\node[ppp] (p1) at (1.8,3.2) {};
		\node[below right=-0.3 and 0.01cm of p1] {$q$};
		\node[pp] (p2) at (2.5,2.5) {};
		\node[below left=-0.1 and -0.05cm of p2] {$p$};
		
		\draw (v1) edge (v2);
		\draw (v1) edge (v5);
		\draw (v2) edge (v3);
		\draw (v3) edge (v4);
		\draw (v4) edge (v5);

		\draw[decoration={brace,raise=5pt},decorate] (v3) -- node[above right=0.1cm and 0.2cm] {$\lambda_r$} (p1);
		\draw[decoration={brace,raise=5pt,mirror},decorate] (v3) -- node[below left=0.1cm and 0.2cm] {$\ell_p$} (p2);
	\end{tikzpicture}
	\caption{The two possible cases for the graph distance between a vertex $r$ and a non-vertex point $p$: in the left graph, the shortest path (represented in red) goes through $u$ ($\ell_p \leq \lambda_r$), whereas in the right graph, it goes through $v$ ($\ell_p > \lambda_r$). Point $q$ is the break-even point of $(u,v)$ for $r$.}
    \label{fig:GraphDistNdPt}
\end{figure}
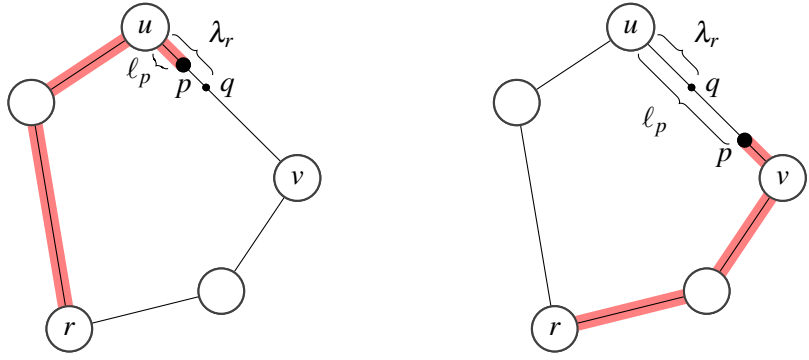

\begin{proof}[Proof of Lemma~\ref{lem:GraphDistanceNodePoint}]
	If $\ell_p \leq \lambda_r$, then 
	\begin{equation}
	    d_G(r,u) + \ell_p \leq d_G(r,u) + \lambda_r,
        \label{eq:ProofGeoDis1}
	\end{equation}
    and
	\begin{equation}
	    d_G(r,v) + d_E(u,v) - \ell_p \geq d_G(r,v) + d_E(u,v) - \lambda_r.
        \label{eq:ProofGeoDis2}
	\end{equation}
    According to Definition~\ref{def:BreakEvenPoint} (Break-even point and distance), the right-hand sides in Equations (\ref{eq:ProofGeoDis1}) and (\ref{eq:ProofGeoDis2}) are equal, so we get:
	\begin{equation}
	    d_G(r,u) + \ell_p \leq d_G(r,v) + d_E(u,v) - \ell_p.
        \label{eq:ProofGeoDis3}
	\end{equation}
    This means the shortest path between $r$ and $p$ going through $u$ is shorter than the one going through $v$. Its length should therefore be considered as the graph distance, as stated by Lemma~\ref{lem:GraphDistanceNodePoint}.
    
    If $\ell_p > \lambda_r$, we proceed symmetrically: the inequalities are reversed in both Equations (\ref{eq:ProofGeoDis1}) and (\ref{eq:ProofGeoDis2}), and therefore in Equation (\ref{eq:ProofGeoDis3}) too, meaning the shortest path goes through $v$ and not $u$. Its length should thus be considered as the graph distance, as stated by Lemma~\ref{lem:GraphDistanceNodePoint}.
\end{proof}

Using the previous results, we can now express the break-even distance for any point (including non-vertices) in terms of relative positions.

\begin{lemma}[Break-even distance for a point]
	\label{lem:BreakEvenDistancePoint} 
	Let $p_1 \in P$ be a point lying on an edge $(u_1,v_1) \in E$, at a distance $\ell_{p_1}$ from $u_1$, and $(u_2,v_2) \in E$ be a distinct edge (though they can have one common end-vertex). Let the break-even distances of $(u_1,v_1)$ for $u_2$ and $v_2$ be $\lambda_{u_2}$ and $\lambda_{v_2}$, respectively. 
	
	The break-even distance $\lambda_{p_1}$ of $(u_2,v_2)$ for $p_1$ is:
	\begin{equation}
		\lambda_{p_1} = 
		\begin{cases}
    		\displaystyle \lambda_{p_1}^{(1)} = \frac{d_G(u_1,v_2) - d_G(u_1,u_2) + d_E(u_2,v_2)}{2},& \text{if } \ell_{p_1} \leq \lambda_{u_2}, \lambda_{v_2}, \\
    		\displaystyle \lambda_{p_1}^{(2)} = \frac{d_G(v_1,v_2) - d_G(u_1,u_2) + d_E(u_1,v_1) + d_E(u_2,v_2) - 2 \ell_{p_1}}{2},& \text{if } \lambda_{v_2} < \ell_{p_1} \leq \lambda_{u_2}, \\
    		\displaystyle \lambda_{p_1}^{(3)} = \frac{d_G(u_1,v_2) - d_G(v_1,u_2) - d_E(u_1,v_1) + d_E(u_2,v_2) + 2 \ell_{p_1}}{2},& \text{if } \lambda_{u_2} < \ell_{p_1} \leq \lambda_{v_2}, \\
    		\displaystyle \lambda_{p_1}^{(4)} = \frac{d_G(v_1,v_2) - d_G(v_1,u_2) + d_E(u_2,v_2)}{2},& \text{if } \ell_{p_1} > \lambda_{u_2}, \lambda_{v_2}.
		\end{cases}
	\end{equation}
\end{lemma}

In the first case ($\ell_{p_1} \leq \lambda_{u_2}, \lambda_{v_2}$), both shortest paths between $p_1$ on one side and $u_2$ and $v_2$ on the other side, go through $u_1$. On the contrary, they both go through $v_1$ in the fourth case ($\ell_{p_1} > \lambda_{u_2}, \lambda_{v_2}$). In the second case ($\lambda_{v_2} < \ell_{p_1} \leq \lambda_{u_2}$), the shortest path goes through $u_1$ for  $u_2$, and through $v_1$ for $v_2$. For the third case ($\lambda_{u_2} < \ell_{p_1} \leq \lambda_{v_2}$), it is the opposite: it goes through $v_1$ for $u_2$ and through $u_1$ for $v_2$.

\begin{proof}[Proof of Lemma~\ref{lem:BreakEvenDistancePoint}]
	Since $u_2$ and $v_2$ are both vertices, the break-even distances $\lambda_{u_2}$ and $\lambda_{v_2}$ can be obtained using Lemma~\ref{lem:BreakEvenDistNode}. However, $p_1$ is not a vertex, which is why some additional processing is required for $\lambda_{p_1}$. 
	
	From Definition~\ref{def:BreakEvenPoint}, we get:
    \begin{equation}
    	\lambda_{p_1} = \frac{d_G(p_1,v_2) - d_G(p_1,u_2) + d_E(u_2,v_2)}{2}.
        \label{eq:Lambda2Def}
    \end{equation}
	We use Lemma~\ref{lem:GraphDistanceNodePoint} to develop $d_G(p_1,v_1)$ and $d_G(p_1,u_2)$:
    \begin{align}
		\label{eq:BEDpointProof1}
        d_G(p_1,u_2) &=
		\begin{cases}
    		d_G(u_1,u_2) + \ell_{p_1},& \text{if } \ell_{p_1} \leq \lambda_{u_2}, \\
		    d_G(v_1,u_2) + d_E(u_1,v_1) - \ell_{p_1},& \text{otherwise.}
		\end{cases}\\
		\label{eq:BEDpointProof2}
		d_G(p_1,v_2) &=
		\begin{cases}
    		d_G(u_1,v_2) + \ell_{p_1},& \text{if } \ell_{p_1} \leq \lambda_{v_2}, \\
		    d_G(v_1,v_2) + d_E(u_1,v_1) - \ell_{p_1},& \text{otherwise.}
		\end{cases}
    \end{align}
	Replacing in Equation (\ref{eq:Lambda2Def}) yields the expression given in Lemma~\ref{lem:BreakEvenDistancePoint}.
\end{proof}

We now use the previous results to identify the shortest path between two points in general, and compute their graph distance.

\begin{lemma}[Graph distance between two points]
	\label{lem:GraphDistancePoints}
	Let $p_1 \in P$ be a point lying on an edge $(u_1,v_1) \in E$, at a distance $\ell_{p_1}$ from $u_1$, and $p_2 \in P$ be a point lying on an edge $(u_2,v_2) \in E$, at a distance $\ell_{p_2}$ from $u_2$. Let the break-even distances of $(u_1,v_1)$ for $u_2$ and $v_2$ be $\lambda_{u_2}$ and $\lambda_{v_2}$, respectively. Let the break-even distance of $(u_2,v_2)$ for $p_1$ be $\lambda_{p_1}$, and $\lambda_{p_1}^{(i)}$ denote the value associated to the $i^{th}$ case of Lemma~\ref{lem:BreakEvenDistancePoint}.	
	
	The graph distance $d_G(p_1,p_2)$ between $p_1$ and $p_2$ is defined as follows:
	\begin{itemize}
		\item If $p_1$ and $p_2$ lie on the same edge, then their graph distance is simply equal to their Euclidean distance: $d_G(p_1,p_2) = d_E(p_1,p_2)$.
		\item Otherwise, if they lie on distinct edges, we must distinguish $4$ cases:
	\end{itemize}
	\begin{align}
		&d_G(p_1,p_2) = 
		\begin{cases}
	    	\ell_{p_1} + d_G(u_1,u_2) + \ell_{p_2} ,& \text{if } \ell_{p_1} \leq \lambda_{u_2}, \lambda_{v_2} \wedge \ell_{p_2} \leq \lambda_{p_1}^{(1)} \\
	        & \vee~ \lambda_{v_2} < \ell_{p_1} \leq \lambda_{u_2} \wedge \ell_{p_2} \leq \lambda_{p_1}^{(2)}, \\
	    	\ell_{p_1} + d_G(u_1,v_2) + d_E(u_2,v_2) - \ell_{p_2} ,& \text{if } \ell_{p_1} \leq \lambda_{u_2}, \lambda_{v_2} \wedge \ell_{p_2} > \lambda_{p_1}^{(1)} \\
	        & \vee~ \lambda_{u_2} < \ell_{p_1} \leq \lambda_{v_2} \wedge \ell_{p_2} > \lambda_{p_1}^{(3)}, \\
		    d_E(u_1,v_1) - \ell_{p_1} + d_G(v_1,u_2) + \ell_{p_2} ,& \text{if } \lambda_{u_2} < \ell_{p_1} \leq \lambda_{v_2} \wedge \ell_{p_2} \leq \lambda_{p_1}^{(3)} \\
	        & \vee~ \ell_{p_1} > \lambda_{u_2}, \lambda_{v_2} \wedge \ell_{p_2} \leq \lambda_{p_1}^{(4)}, \\
		    d_E(u_1,v_1) - \ell_{p_1} + d_G(v_1,v_2) + d_E(u_2,v_2) - \ell_{p_2} ,& \text{if } \lambda_{v_2} < \ell_{p_1} \leq \lambda_{u_2} \wedge \ell_{p_2} > \lambda_{p_1}^{(2)} \\
	        & \vee~ \ell_{p_1} > \lambda_{u_2}, \lambda_{v_2} \wedge \ell_{p_2} > \lambda_{p_1}^{(4)}. \\
		\end{cases}
	\end{align}
\end{lemma}

The particular case (both points on the same edge) occurs if $(u_1,v_1) = (u_2,v_2)$ (i.e. the edges are actually the same), or if the edges have a common end-vertex and this end-vertex is either $p_1$ or $p_2$. 

Figure~\ref{fig:GraphDistPtPt} illustrates the different possible situations described by the general case (i.e. two distinct edges), which includes $4$ possible expressions. In the first expression, the shortest path between $p_1$ and $p_2$ goes through $u_1$ and $u_2$ (graphs $a$ and $d$ in the Figure) ; in the second, it goes through $u_1$ and $v_2$ (graphs $e$ and $f$) ; in the third, through $v_1$ and $u_2$ (graphs $b$ and $c$) ; and in the fourth through $v_1$ and $v_2$ (graphs $g$ and $h$).

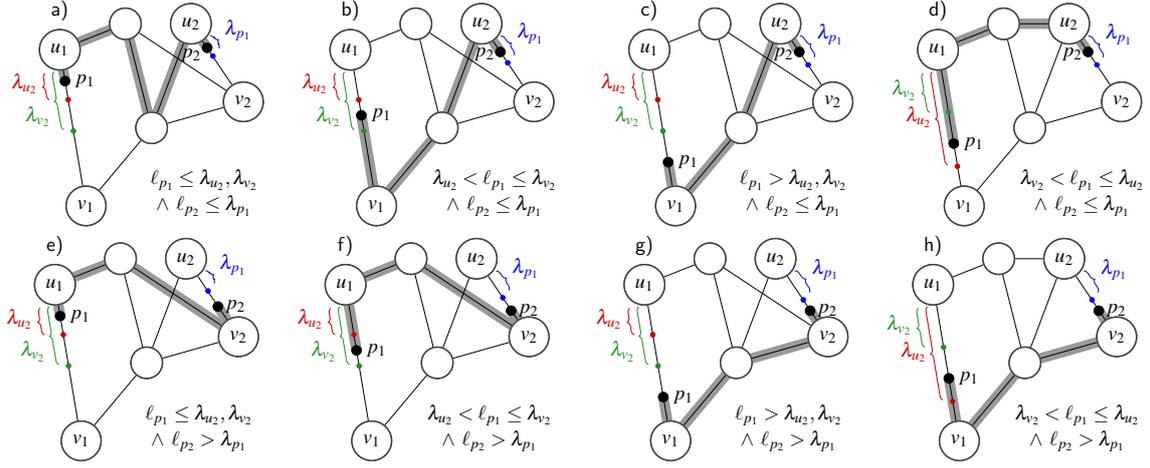
\begin{figure}[ht]
	\centering
	\scalebox{0.69}{\begin{tikzpicture}
		\node[nn] (v1) at (0,0)    {$v_1$};
		\node[nn] (v2) at (-0.5,3) {$u_1$};
		\node[nn] (v3) at (0.75,3.5)  {};
		\node[nn] (v4) at (2,3.5)    {$u_2$};
		\node[nn] (v5) at (3,2)    {$v_2$};
		\node[nn] (v6) at (1.25,1.5)  {};
 		\node[pppr] (qu) at (-0.34,2.04) {};
 		\node[pppg] (qv) at (-0.24,1.44) {};
 		\node[pppb] (q2) at (2.42,2.88)  {};
 		\node[pp] (p1) at (-0.4,2.4) {};
 		\node[pp] (p2) at (2.31,3.04) {};
 		
		\draw[greylink]  (p1) edge (v2);
		\draw[greylink]  (v2) edge (v3);
		\draw[greylink]  (v3) edge (v6);
		\draw[greylink]  (v6) edge (v4);
		\draw[greylink]  (v4) edge (p2);
		
		\draw (v1) edge (v2);
		\draw (v1) edge (v6);
		\draw (v2) edge (v3);
		\draw (v3) edge (v5);
		\draw (v3) edge (v6);
		\draw (v4) edge (v5);
		\draw (v4) edge (v6);
		\draw (v5) edge (v6);

		\node[nn] (v1) at (0,0)    {$v_1$};
		\node[nn] (v2) at (-0.5,3) {$u_1$};
		\node[nn] (v3) at (0.75,3.5)  {};
		\node[nn] (v4) at (2,3.5)    {$u_2$};
		\node[nn] (v5) at (3,2)    {$v_2$};
		\node[nn] (v6) at (1.25,1.5)  {};
 		\node[pppr] (qu) at (-0.34,2.04) {};
 		\node[pppg] (qv) at (-0.24,1.44) {};
 		\node[pppb] (q2) at (2.42,2.88)  {};
 		\node[pp] (p1) at (-0.4,2.4) {};
		\node[right=-0.05cm of p1] {$p_1$};
 		\node[pp] (p2) at (2.31,3.04) {};
		\node[below left=-0.2cm and -0.1cm of p2] {$p_2$};
		
		\draw[decoration={brace,raise=10pt,mirror},decorate,draw=myred] (v2) -- node[left=0.4cm] {\textcolor{myred}{$\lambda_{u_2}$}} (qu);
		\draw[decoration={brace,raise=5pt,mirror},decorate,draw=mygreen] (v2) -- node[below left=0.0cm and 0.2cm] {\textcolor{mygreen}{$\lambda_{v_2}$}} (qv);
		\draw[decoration={brace,raise=5pt},decorate,draw=myblue] (v4) -- node[above right=0.0cm and 0.2cm] {\textcolor{myblue}{$\lambda_{p_1}$}} (q2);
		
		\node (txt1) at (2.25,0.5) {$\ell_{p_1} \leq \lambda_{u_2},\lambda_{v_2}$};
		\node (txt2) at (2.25,0) {$\wedge~ \ell_{p_2} \leq \lambda_{p_1}$};
		\node (txt3) at (-0.5,3.75) {a)};
	\end{tikzpicture}}
	\hfill
	\scalebox{0.69}{\begin{tikzpicture}
		\node[nn] (v1) at (0,0)    {$v_1$};
		\node[nn] (v2) at (-0.5,3) {$u_1$};
		\node[nn] (v3) at (0.75,3.5)  {};
		\node[nn] (v4) at (2,3.5)    {$u_2$};
		\node[nn] (v5) at (3,2)    {$v_2$};
		\node[nn] (v6) at (1.25,1.5)  {};
 		\node[pppr] (qu) at (-0.34,2.04) {};
 		\node[pppg] (qv) at (-0.24,1.44) {};
 		\node[pppb] (q2) at (2.50,2.75)  {};
 		\node[pp] (p1) at (-0.29,1.74) {};
 		\node[pp] (p2) at (2.36,2.96) {};
		
		\draw[greylink]  (p1) edge (v1);
		\draw[greylink]  (v1) edge (v6);
		\draw[greylink]  (v6) edge (v4);
		\draw[greylink]  (v4) edge (p2);
		
		\draw (v1) edge (v2);
		\draw (v1) edge (v6);
		\draw (v2) edge (v3);
		\draw (v3) edge (v5);
		\draw (v3) edge (v6);
		\draw (v4) edge (v5);
		\draw (v4) edge (v6);
		\draw (v5) edge (v6);

		\node[nn] (v1) at (0,0)    {$v_1$};
		\node[nn] (v2) at (-0.5,3) {$u_1$};
		\node[nn] (v3) at (0.75,3.5)  {};
		\node[nn] (v4) at (2,3.5)    {$u_2$};
		\node[nn] (v5) at (3,2)    {$v_2$};
		\node[nn] (v6) at (1.25,1.5)  {};
 		\node[pppr] (qu) at (-0.34,2.04) {};
 		\node[pppg] (qv) at (-0.24,1.44) {};
 		\node[pppb] (q2) at (2.50,2.75)  {};
 		\node[pp] (p1) at (-0.29,1.74) {};
		\node[right=-0.05cm of p1] {$p_1$};
 		\node[pp] (p2) at (2.36,2.96) {};
		\node[below left=-0.3cm and -0.1cm of p2] {$p_2$};
		
		\draw[decoration={brace,raise=10pt,mirror},decorate,draw=myred] (v2) -- node[left=0.4cm] {\textcolor{myred}{$\lambda_{u_2}$}} (qu);
		\draw[decoration={brace,raise=5pt,mirror},decorate,draw=mygreen] (v2) -- node[below left=0.0cm and 0.2cm] {\textcolor{mygreen}{$\lambda_{v_2}$}} (qv);
		\draw[decoration={brace,raise=5pt},decorate,draw=myblue] (v4) -- node[above right=0.0cm and 0.2cm] {\textcolor{myblue}{$\lambda_{p_1}$}} (q2);
		
		\node (txt1) at (2.25,0.5) {$\lambda_{u_2} < \ell_{p_1} \leq \lambda_{v_2}$};
		\node (txt2) at (2.25,0) {$\wedge~ \ell_{p_2} \leq \lambda_{p_1}$};
		\node (txt3) at (-0.5,3.75) {b)};
	\end{tikzpicture}}
	\hfill
	\scalebox{0.69}{\begin{tikzpicture}
		\node[nn] (v1) at (0,0)    {$v_1$};
		\node[nn] (v2) at (-0.5,3) {$u_1$};
		\node[nn] (v3) at (0.75,3.5)  {};
		\node[nn] (v4) at (2,3.5)    {$u_2$};
		\node[nn] (v5) at (3,2)    {$v_2$};
		\node[nn] (v6) at (1.25,1.5)  {};
 		\node[pppr] (qu) at (-0.34,2.04) {};
 		\node[pppg] (qv) at (-0.24,1.44) {};
 		\node[pppb] (q2) at (2.53,2.71)  {};
 		\node[pp] (p1) at (-0.14,0.84) {};
 		\node[pp] (p2) at (2.36,2.96) {};
		
		\draw[greylink]  (p1) edge (v1);
		\draw[greylink]  (v1) edge (v6);
		\draw[greylink]  (v6) edge (v4);
		\draw[greylink]  (v4) edge (p2);
		
		\draw (v1) edge (v2);
		\draw (v1) edge (v6);
		\draw (v2) edge (v3);
		\draw (v3) edge (v5);
		\draw (v3) edge (v6);
		\draw (v4) edge (v5);
		\draw (v4) edge (v6);
		\draw (v5) edge (v6);
		
		\node[nn] (v1) at (0,0)    {$v_1$};
		\node[nn] (v2) at (-0.5,3) {$u_1$};
		\node[nn] (v3) at (0.75,3.5)  {};
		\node[nn] (v4) at (2,3.5)    {$u_2$};
		\node[nn] (v5) at (3,2)    {$v_2$};
		\node[nn] (v6) at (1.25,1.5)  {};
 		\node[pppr] (qu) at (-0.34,2.04) {};
 		\node[pppg] (qv) at (-0.24,1.44) {};
 		\node[pppb] (q2) at (2.53,2.71)  {};
 		\node[pp] (p1) at (-0.14,0.84) {};
		\node[right=-0.05cm of p1] {$p_1$};
 		\node[pp] (p2) at (2.36,2.96) {};
		\node[below left=-0.3cm and -0.1cm of p2] {$p_2$};
		
		\draw[decoration={brace,raise=10pt,mirror},decorate,draw=myred] (v2) -- node[left=0.4cm] {\textcolor{myred}{$\lambda_{u_2}$}} (qu);
		\draw[decoration={brace,raise=5pt,mirror},decorate,draw=mygreen] (v2) -- node[below left=0.0cm and 0.2cm] {\textcolor{mygreen}{$\lambda_{v_2}$}} (qv);
		\draw[decoration={brace,raise=5pt},decorate,draw=myblue] (v4) -- node[above right=0.0cm and 0.2cm] {\textcolor{myblue}{$\lambda_{p_1}$}} (q2);
		
		\node (txt1) at (2.25,0.5) {$\ell_{p_1} > \lambda_{u_2},\lambda_{v_2}$};
		\node (txt2) at (2.25,0) {$\wedge~ \ell_{p_2} \leq \lambda_{p_1}$};
		\node (txt3) at (-0.5,3.75) {c)};
	\end{tikzpicture}}
	\hfill
	\scalebox{0.69}{\begin{tikzpicture}
		\node[nn] (v1) at (0,0)    {$v_1$};
		\node[nn] (v2) at (-0.5,3) {$u_1$};
		\node[nn] (v3) at (0.75,3.5)  {};
		\node[nn] (v4) at (2,3.5)    {$u_2$};
		\node[nn] (v5) at (3,2)    {$v_2$};
		\node[nn] (v6) at (1.25,1.5)  {};
 		\node[pppr] (qu) at (-0.13,0.76) {};
 		\node[pppg] (qv) at (-0.30,1.80) {};
 		\node[pppb] (q2) at (2.53,2.71)  {};
 		\node[pp] (p1) at (-0.20,1.20) {};
 		\node[pp] (p2) at (2.36,2.96) {};
		
		\draw[greylink]  (p1) edge (v2);
		\draw[greylink]  (v2) edge (v3);
		\draw[greylink]  (v3) edge (v4);
		\draw[greylink]  (v4) edge (p2);
		
		\draw (v1) edge (v2);
		\draw (v1) edge (v6);
		\draw (v2) edge (v3);
		\draw (v3) edge (v4);
		\draw (v3) edge (v6);
		\draw (v4) edge (v5);
		\draw (v4) edge (v6);
		\draw (v5) edge (v6);
		
		\node[nn] (v1) at (0,0)    {$v_1$};
		\node[nn] (v2) at (-0.5,3) {$u_1$};
		\node[nn] (v3) at (0.75,3.5)  {};
		\node[nn] (v4) at (2,3.5)    {$u_2$};
		\node[nn] (v5) at (3,2)    {$v_2$};
		\node[nn] (v6) at (1.25,1.5)  {};
 		\node[pppr] (qu) at (-0.13,0.76) {};
 		\node[pppg] (qv) at (-0.30,1.80) {};
 		\node[pppb] (q2) at (2.53,2.71)  {};
 		\node[pp] (p1) at (-0.20,1.20) {};
		\node[right=-0.05cm of p1] {$p_1$};
 		\node[pp] (p2) at (2.36,2.96) {};
		\node[below left=-0.3cm and -0.1cm of p2] {$p_2$};
		
		\draw[decoration={brace,raise=5pt,mirror},decorate,draw=myred] (v2) -- node[left=0.2cm] {\textcolor{myred}{$\lambda_{u_2}$}} (qu);
		\draw[decoration={brace,raise=10pt,mirror},decorate,draw=mygreen] (v2) -- node[left=0.4cm] {\textcolor{mygreen}{$\lambda_{v_2}$}} (qv);
		\draw[decoration={brace,raise=5pt},decorate,draw=myblue] (v4) -- node[above right=0.0cm and 0.2cm] {\textcolor{myblue}{$\lambda_{p_1}$}} (q2);
		
		\node (txt1) at (2.25,0.5) {$\lambda_{v_2} < \ell_{p_1} \leq \lambda_{u_2}$};
		\node (txt2) at (2.25,0) {$\wedge~ \ell_{p_2} \leq \lambda_{p_1}$};
		\node (txt3) at (-0.5,3.75) {d)};
	\end{tikzpicture}}\\

	\scalebox{0.69}{\begin{tikzpicture}
		\node[nn] (v1) at (0,0)    {$v_1$};
		\node[nn] (v2) at (-0.5,3) {$u_1$};
		\node[nn] (v3) at (0.75,3.5)  {};
		\node[nn] (v4) at (2,3.5)    {$u_2$};
		\node[nn] (v5) at (3,2)    {$v_2$};
		\node[nn] (v6) at (1.25,1.5)  {};
 		\node[pppr] (qu) at (-0.34,2.04) {};
 		\node[pppg] (qv) at (-0.24,1.44) {};
 		\node[pppb] (q2) at (2.42,2.88)  {};
 		\node[pp] (p1) at (-0.4,2.4) {};
 		\node[pp] (p2) at (2.61,2.58) {};
		
		\draw[greylink]  (p1) edge (v2);
		\draw[greylink]  (v2) edge (v3);
		\draw[greylink]  (v3) edge (v5);
		\draw[greylink]  (v5) edge (p2);
		
		\draw (v1) edge (v2);
		\draw (v1) edge (v6);
		\draw (v2) edge (v3);
		\draw (v3) edge (v5);
		\draw (v3) edge (v6);
		\draw (v4) edge (v5);
		\draw (v4) edge (v6);
		\draw (v5) edge (v6);
		
		\node[nn] (v1) at (0,0)    {$v_1$};
		\node[nn] (v2) at (-0.5,3) {$u_1$};
		\node[nn] (v3) at (0.75,3.5)  {};
		\node[nn] (v4) at (2,3.5)    {$u_2$};
		\node[nn] (v5) at (3,2)    {$v_2$};
		\node[nn] (v6) at (1.25,1.5)  {};
 		\node[pppr] (qu) at (-0.34,2.04) {};
 		\node[pppg] (qv) at (-0.24,1.44) {};
 		\node[pppb] (q2) at (2.42,2.88)  {};
 		\node[pp] (p1) at (-0.4,2.4) {};
		\node[right=-0.05cm of p1] {$p_1$};
 		\node[pp] (p2) at (2.61,2.58) {};
		\node[right=-0.1cm of p2] {$p_2$};
		
		\draw[decoration={brace,raise=10pt,mirror},decorate,draw=myred] (v2) -- node[left=0.4cm] {\textcolor{myred}{$\lambda_{u_2}$}} (qu);
		\draw[decoration={brace,raise=5pt,mirror},decorate,draw=mygreen] (v2) -- node[below left=0.0cm and 0.2cm] {\textcolor{mygreen}{$\lambda_{v_2}$}} (qv);
		\draw[decoration={brace,raise=5pt},decorate,draw=myblue] (v4) -- node[above right=0.0cm and 0.2cm] {\textcolor{myblue}{$\lambda_{p_1}$}} (q2);
		
		\node (txt1) at (2.25,0.5) {$\ell_{p_1} \leq \lambda_{u_2},\lambda_{v_2}$};
		\node (txt2) at (2.25,0) {$\wedge~ \ell_{p_2} > \lambda_{p_1}$};
		\node (txt3) at (-0.5,3.75) {e)};
	\end{tikzpicture}}
	\hfill
	\scalebox{0.69}{\begin{tikzpicture}
		\node[nn] (v1) at (0,0)    {$v_1$};
		\node[nn] (v2) at (-0.5,3) {$u_1$};
		\node[nn] (v3) at (0.75,3.5)  {};
		\node[nn] (v4) at (2,3.5)    {$u_2$};
		\node[nn] (v5) at (3,2)    {$v_2$};
		\node[nn] (v6) at (1.25,1.5)  {};
 		\node[pppr] (qu) at (-0.34,2.04) {};
 		\node[pppg] (qv) at (-0.24,1.44) {};
 		\node[pppb] (q2) at (2.50,2.75)  {};
 		\node[pp] (p1) at (-0.29,1.74) {};
 		\node[pp] (p2) at (2.66,2.50) {};
		
		\draw[greylink]  (p1) edge (v2);
		\draw[greylink]  (v2) edge (v3);
		\draw[greylink]  (v3) edge (v5);
		\draw[greylink]  (v5) edge (p2);
		
		\draw (v1) edge (v2);
		\draw (v1) edge (v6);
		\draw (v2) edge (v3);
		\draw (v3) edge (v5);
		\draw (v3) edge (v6);
		\draw (v4) edge (v5);
		\draw (v4) edge (v6);
		\draw (v5) edge (v6);

		\node[nn] (v1) at (0,0)    {$v_1$};
		\node[nn] (v2) at (-0.5,3) {$u_1$};
		\node[nn] (v3) at (0.75,3.5)  {};
		\node[nn] (v4) at (2,3.5)    {$u_2$};
		\node[nn] (v5) at (3,2)    {$v_2$};
		\node[nn] (v6) at (1.25,1.5)  {};
 		\node[pppr] (qu) at (-0.34,2.04) {};
 		\node[pppg] (qv) at (-0.24,1.44) {};
 		\node[pppb] (q2) at (2.50,2.75)  {};
 		\node[pp] (p1) at (-0.29,1.74) {};
		\node[right=-0.05cm of p1] {$p_1$};
 		\node[pp] (p2) at (2.66,2.50) {};
		\node[right=-0.1cm of p2] {$p_2$};
		
		\draw[decoration={brace,raise=10pt,mirror},decorate,draw=myred] (v2) -- node[left=0.4cm] {\textcolor{myred}{$\lambda_{u_2}$}} (qu);
		\draw[decoration={brace,raise=5pt,mirror},decorate,draw=mygreen] (v2) -- node[below left=0.0cm and 0.2cm] {\textcolor{mygreen}{$\lambda_{v_2}$}} (qv);
		\draw[decoration={brace,raise=5pt},decorate,draw=myblue] (v4) -- node[above right=0.0cm and 0.2cm] {\textcolor{myblue}{$\lambda_{p_1}$}} (q2);
		
		\node (txt1) at (2.25,0.5) {$\lambda_{u_2} < \ell_{p_1} \leq \lambda_{v_2}$};
		\node (txt2) at (2.25,0) {$\wedge~ \ell_{p_2} > \lambda_{p_1}$};
		\node (txt3) at (-0.5,3.75) {f)};
	\end{tikzpicture}}
	\hfill
	\scalebox{0.69}{\begin{tikzpicture}
		\node[nn] (v1) at (0,0)    {$v_1$};
		\node[nn] (v2) at (-0.5,3) {$u_1$};
		\node[nn] (v3) at (0.75,3.5)  {};
		\node[nn] (v4) at (2,3.5)    {$u_2$};
		\node[nn] (v5) at (3,2)    {$v_2$};
		\node[nn] (v6) at (1.25,1.5)  {};
 		\node[pppr] (qu) at (-0.34,2.04) {};
 		\node[pppg] (qv) at (-0.24,1.44) {};
 		\node[pppb] (q2) at (2.53,2.71)  {};
 		\node[pp] (p1) at (-0.14,0.84) {};
 		\node[pp] (p2) at (2.66,2.50) {};
		
		\draw[greylink]  (p1) edge (v1);
		\draw[greylink]  (v1) edge (v6);
		\draw[greylink]  (v6) edge (v5);
		\draw[greylink]  (v5) edge (p2);
		
		\draw (v1) edge (v2);
		\draw (v1) edge (v6);
		\draw (v2) edge (v3);
		\draw (v3) edge (v5);
		\draw (v3) edge (v6);
		\draw (v4) edge (v5);
		\draw (v4) edge (v6);
		\draw (v5) edge (v6);
		
		\node[nn] (v1) at (0,0)    {$v_1$};
		\node[nn] (v2) at (-0.5,3) {$u_1$};
		\node[nn] (v3) at (0.75,3.5)  {};
		\node[nn] (v4) at (2,3.5)    {$u_2$};
		\node[nn] (v5) at (3,2)    {$v_2$};
		\node[nn] (v6) at (1.25,1.5)  {};
 		\node[pppr] (qu) at (-0.34,2.04) {};
 		\node[pppg] (qv) at (-0.24,1.44) {};
 		\node[pppb] (q2) at (2.53,2.71)  {};
 		\node[pp] (p1) at (-0.14,0.84) {};
		\node[right=-0.05cm of p1] {$p_1$};
 		\node[pp] (p2) at (2.66,2.50) {};
		\node[right=-0.1cm of p2] {$p_2$};
		
		\draw[decoration={brace,raise=10pt,mirror},decorate,draw=myred] (v2) -- node[left=0.4cm] {\textcolor{myred}{$\lambda_{u_2}$}} (qu);
		\draw[decoration={brace,raise=5pt,mirror},decorate,draw=mygreen] (v2) -- node[below left=0.0cm and 0.2cm] {\textcolor{mygreen}{$\lambda_{v_2}$}} (qv);
		\draw[decoration={brace,raise=5pt},decorate,draw=myblue] (v4) -- node[above right=0.0cm and 0.2cm] {\textcolor{myblue}{$\lambda_{p_1}$}} (q2);
		
		\node (txt1) at (2.25,0.5) {$\ell_{p_1} > \lambda_{u_2},\lambda_{v_2}$};
		\node (txt2) at (2.25,0) {$\wedge~ \ell_{p_2} > \lambda_{p_1}$};
		\node (txt3) at (-0.5,3.75) {g)};
	\end{tikzpicture}}
	\hfill
	\scalebox{0.69}{\begin{tikzpicture}
		\node[nn] (v1) at (0,0)    {$v_1$};
		\node[nn] (v2) at (-0.5,3) {$u_1$};
		\node[nn] (v3) at (0.75,3.5)  {};
		\node[nn] (v4) at (2,3.5)    {$u_2$};
		\node[nn] (v5) at (3,2)    {$v_2$};
		\node[nn] (v6) at (1.25,1.5)  {};
 		\node[pppr] (qu) at (-0.13,0.76) {};
 		\node[pppg] (qv) at (-0.30,1.80) {};
 		\node[pppb] (q2) at (2.53,2.71)  {};
 		\node[pp] (p1) at (-0.20,1.20) {};
 		\node[pp] (p2) at (2.66,2.50) {};
		
		\draw[greylink]  (p1) edge (v1);
		\draw[greylink]  (v1) edge (v6);
		\draw[greylink]  (v6) edge (v5);
		\draw[greylink]  (v5) edge (p2);
		
		\draw (v1) edge (v2);
		\draw (v1) edge (v6);
		\draw (v2) edge (v3);
		\draw (v3) edge (v4);
		\draw (v3) edge (v6);
		\draw (v4) edge (v5);
		\draw (v4) edge (v6);
		\draw (v5) edge (v6);
		
		\node[nn] (v1) at (0,0)    {$v_1$};
		\node[nn] (v2) at (-0.5,3) {$u_1$};
		\node[nn] (v3) at (0.75,3.5)  {};
		\node[nn] (v4) at (2,3.5)    {$u_2$};
		\node[nn] (v5) at (3,2)    {$v_2$};
		\node[nn] (v6) at (1.25,1.5)  {};
 		\node[pppr] (qu) at (-0.13,0.76) {};
 		\node[pppg] (qv) at (-0.30,1.80) {};
 		\node[pppb] (q2) at (2.53,2.71)  {};
 		\node[pp] (p1) at (-0.20,1.20) {};
		\node[right=-0.05cm of p1] {$p_1$};
 		\node[pp] (p2) at (2.66,2.50) {};
		\node[right=-0.1cm of p2] {$p_2$};
		
		\draw[decoration={brace,raise=5pt,mirror},decorate,draw=myred] (v2) -- node[left=0.2cm] {\textcolor{myred}{$\lambda_{u_2}$}} (qu);
		\draw[decoration={brace,raise=10pt,mirror},decorate,draw=mygreen] (v2) -- node[left=0.4cm] {\textcolor{mygreen}{$\lambda_{v_2}$}} (qv);
		\draw[decoration={brace,raise=5pt},decorate,draw=myblue] (v4) -- node[above right=0.0cm and 0.2cm] {\textcolor{myblue}{$\lambda_{p_1}$}} (q2);
		
		\node (txt1) at (2.25,0.5) {$\lambda_{v_2} < \ell_{p_1} \leq \lambda_{u_2}$};
		\node (txt2) at (2.25,0) {$\wedge~ \ell_{p_2} > \lambda_{p_1}$};
		\node (txt3) at (-0.5,3.75) {h)};
	\end{tikzpicture}}
	\caption{The different possible situations for the general case of Lemma~\ref{lem:GraphDistancePoints}, describing the graph distance between two non-vertex points $p_1$ and $p_2$. The shortest path is represented in gray.}
    \label{fig:GraphDistPtPt}
\end{figure}

In the general case, and for a fixed $p_1$, $d_G(p_1,p_2)$ is a piecewise linear function of $\ell_{p_2}$. Note that fixing $p_1$ amounts to fixing $\lambda_{p_1}$, which makes the function $2$-pieced. Moreover, by definition of the break-even point (cf. Definition~\ref{def:BreakEvenPoint}), this function is continuous since both pieces are equal when $\ell_{p_2} = \lambda_{p_1}$.

\begin{proof}[Proof of Lemma~\ref{lem:GraphDistancePoints}]
	We focus on the general situation, in which the edges are distinct, since the other one is trivial. This proof is two-fold: we must first identify which vertex of $(u_1,v_1)$ is on the shortest path between $p_1$ and $p_2$, and then which vertex of $(u_2,v_2)$ is. 

	We start with the first part, i.e. edge $(u_1,v_1)$. From Lemma~\ref{lem:BreakEvenDistancePoint}, we know that if $\ell_{p_1} \leq \lambda_{u_2}$ (resp. $\lambda_{v_2}$), the shortest path between $p_1$ and $u_2$ (resp. $v_2$) passes through $u_1$, otherwise it passes through $v_1$.
	
	Let us now switch to the second part, i.e. edge $(u_2,v_2)$. The principle of this proof is similar to that of Lemma~\ref{lem:GraphDistanceNodePoint}, except we must additionally take into account the value of $\ell_{p_1}$ relatively to $\lambda_{u_2}$ and $\lambda_{v_2}$.
    
    Let us first suppose that $\ell_{p_2} \leq \lambda_{p_1}$. We can state:
	\begin{empheq}[left=\empheqlbrace]{align}
   		\label{eq:ProofLemmaGraphDistPoints1}
    	&\quad \ell_{p_1} + d_G(u_1,u_2) + \ell_{p_2} \leq \ell_{p_1} + d_G(u_1,u_2) + \lambda_{p_1}, \\
    	\label{eq:ProofLemmaGraphDistPoints2}
	    &\quad d_E(u_1,v_1) - \ell_{p_1} + d_G(v_1,u_2) + \ell_{p_2} \leq d_E(u_1,v_1) - \ell_{p_1} + d_G(v_1,u_2) + \lambda_{p_1}, \\
    	\label{eq:ProofLemmaGraphDistPoints3}
    	&\quad \ell_{p_1} + d_G(u_1,v_2) + d_E(u_2,v_2) - \ell_{p_2} \geq \ell_{p_1} + d_G(u_1,v_2) + d_E(u_2,v_2) - \lambda_{p_1}, \\
    	\label{eq:ProofLemmaGraphDistPoints4}
	    &\quad d_E(u_1,v_1) - \ell_{p_1} + d_G(v_1,v_2) + d_E(u_2,v_2) - \ell_{p_2} \geq d_E(u_1,v_1) - \ell_{p_1} + d_G(v_1,v_2) + d_E(u_2,v_2) - \lambda_{p_1}.
	\end{empheq}    
According to Definition~\ref{def:BreakEvenPoint}, we have:
	\begin{equation}
    	\label{eq:ProofLemmaGraphDistPoints5}
		d_G(p_1,u_2) + \lambda_{p_1} = d_G(p_1,v_2) + d_E(u_2,v_2) - \lambda_{p_1}
	\end{equation}
By replacing Equations (\ref{eq:BEDpointProof1}) and (\ref{eq:BEDpointProof2}) in Equations (\ref{eq:ProofLemmaGraphDistPoints5}), we get:
	{\small\begin{equation}
    	\label{eq:ProofLemmaGraphDistPoints6}
		\begin{cases}
    		\ell_{p_1} + d_G(u_1,u_2) + \lambda_{p_1} = \ell_{p_1} + d_G(u_1,v_2) + d_E(u_2,v_2) - \lambda_{p_1} ,& \text{if } \ell_{p_1} \leq \lambda_{u_2}, \lambda_{v_2} \\
    		\ell_{p_1} + d_G(u_1,u_2) + \lambda_{p_1} = d_E(u_1,v_1) - \ell_{p_1} + d_G(v_1,v_2) + d_E(u_2,v_2) - \lambda_{p_1} ,& \text{if } \lambda_{v_2} < \ell_{p_1} \leq \lambda_{u_2} \\
    		d_E(u_1,v_1) - \ell_{p_1} + d_G(v_1,u_2) + \lambda_{p_1} = \ell_{p_1} + d_G(u_1,v_2) + d_E(u_2,v_2) - \lambda_{p_1} ,& \text{if } \lambda_{u_2} < \ell_{p_1} \leq \lambda_{v_2} \\
    		d_E(u_1,v_1) - \ell_{p_1} + d_G(v_1,u_2) + \lambda_{p_1} = d_E(u_1,v_1) - \ell_{p_1} + d_G(v_1,v_2) + d_E(u_2,v_2) - \lambda_{p_1} ,& \text{if } \ell_{p_1} > \lambda_{u_2}, \lambda_{v_2}. \\
		\end{cases}
	\end{equation}}
	Notice that when $\ell_{p_1} \leq \lambda_{u_2}, \lambda_{v_2}$, the right-hand sides of Equations (\ref{eq:ProofLemmaGraphDistPoints1}) and (\ref{eq:ProofLemmaGraphDistPoints3}) are equal. When $\lambda_{v_2} < \ell_{p_1} \leq \lambda_{u_2}$, the same remark can be made for Equations (\ref{eq:ProofLemmaGraphDistPoints1}) and (\ref{eq:ProofLemmaGraphDistPoints4}) ; when $\lambda_{u_2} < \ell_{p_1} \leq \lambda_{v_2}$ for Equations (\ref{eq:ProofLemmaGraphDistPoints2}) and (\ref{eq:ProofLemmaGraphDistPoints3}), and when $\ell_{p_1} > \lambda_{u_2}, \lambda_{v_2}$ for Equations (\ref{eq:ProofLemmaGraphDistPoints2}) and (\ref{eq:ProofLemmaGraphDistPoints4}). Consequently, we get the following system:
	{\small\begin{equation}
    	\label{eq:ProofLemmaGraphDistPoints7}
		\begin{cases}
			\ell_{p_1} + d_G(u_1,u_2) + \ell_{p_2} \leq \ell_{p_1} + d_G(u_1,v_2) + d_E(u_2,v_2) - \ell_{p_2} ,& \text{if } \ell_{p_1} \leq \lambda_{u_2}, \lambda_{v_2} \\
    		\ell_{p_1} + d_G(u_1,u_2) + \ell_{p_2} \leq d_E(u_1,v_1) - \ell_{p_1} + d_G(v_1,v_2) + d_E(u_2,v_2) - \ell_{p_2} ,& \text{if } \lambda_{v_2} < \ell_{p_1} \leq \lambda_{u_2} \\
    		d_E(u_1,v_1) - \ell_{p_1} + d_G(v_1,u_2) + \ell_{p_2} \leq \ell_{p_1} + d_G(u_1,v_2) + d_E(u_2,v_2) - \ell_{p_2} ,& \text{if } \lambda_{u_2} < \ell_{p_1} \leq \lambda_{v_2} \\
    		d_E(u_1,v_1) - \ell_{p_1} + d_G(v_1,u_2) + \ell_{p_2} \leq d_E(u_1,v_1) - \ell_{p_1} + d_G(v_1,v_2) + d_E(u_2,v_2) - \ell_{p_2} ,& \text{if } \ell_{p_1} > \lambda_{u_2}, \lambda_{v_2}. \\
		\end{cases}
	\end{equation}}
	This means the shortest path between $p_1$ and $p_2$ always go through $u_2$ when $\ell_{p_2} \leq \lambda_{p_1}$. So, from the first part of the proof, we can deduce only $\lambda_{u_2}$ is relevant regarding the shortest path: if $\ell_{p_1} \leq \lambda_{u_2}$ it goes through $u_1$, otherwise it goes through $v_1$. In these cases, the lengths of these paths should thus be considered as the graph distances, as stated in Lemma~\ref{lem:GraphDistancePoints} (first and third cases, respectively).

    If $\ell_{p_2} < \lambda_{p_1}$, we can proceed symmetrically, like we did for Lemma~\ref{lem:GraphDistanceNodePoint}, by reversing the inequalities in Equations (\ref{eq:ProofLemmaGraphDistPoints1})--(\ref{eq:ProofLemmaGraphDistPoints4}), and consequently in Equation~\ref{eq:ProofLemmaGraphDistPoints7}. This means that in this case, the shortest path always go through $v_2$. Therefore, only $\lambda_{v_2}$ is relevant regarding the shortest path: if $\ell_{p_1} \leq \lambda_{v_2}$ it goes through $u_1$, otherwise it goes through $v_1$. In these cases, the lengths of these paths should thus be considered as the graph distances, as stated in Lemma~\ref{lem:GraphDistancePoints} (second and fourth cases).
\end{proof}

We now have the necessary tools to express the Straightness in terms of relative positions of $p_1$ and $p_2$. However, before doing so, we define a group of auxiliary functions aiming at making later equations more compact and easier to read.

\begin{definition}[Auxiliary functions]
	\label{def:AuxiliaryFunction}
	Let us define the generic auxiliary function $f(\ell_{p_1},\ell_{p_2};\alpha,\beta,\gamma)$ as:
	\begin{equation}
		f(\ell_{p_1},\ell_{p_2};\alpha,\beta,\gamma) = \frac{d_E(p_1,p_2)}{\alpha + \beta \ell_{p_1} + \gamma \ell_{p_2}}
	\end{equation}
	where $\alpha$, $\beta$ and $\gamma$ are parameters of the function.
	
	We additionally define the following $4$ auxiliary functions, which are all parameterized instances of $f$:
	\begin{align}
	    f_{u_1u_2}(\ell_{p_1},\ell_{p_2}) &= f(\ell_{p_1},\ell_{p_2};d_G(u_1,u_2),1,1), \\
	    f_{u_1v_2}(\ell_{p_1},\ell_{p_2}) &= f(\ell_{p_1},\ell_{p_2};d_G(u_1,v_2) + d_E(u_2,v_2),1,-1), \\
		f_{v_1u_2}(\ell_{p_1},\ell_{p_2}) &= f(\ell_{p_1},\ell_{p_2};d_E(u_1,v_1) + d_G(v_1,u_2),-1,1), \\
		f_{v_1v_2}(\ell_{p_1},\ell_{p_2}) &= f(\ell_{p_1},\ell_{p_2};d_E(u_1,v_1) + d_G(v_1,v_2) + d_E(u_2,v_2),-1,-1).
	\end{align}
\end{definition}

Note that $d_E(p_1,p_2)$, the numerator of $f$, contains only absolute coordinates of vertices and relative positions of non-vertex points, i.e. no absolute coordinates of non-vertex points (cf. Lemma~\ref{lem:EuclideanDistance}). 

\begin{theorem}[Straightness between two points]
	\label{lem:StraightnessRelative}
	Let $p_1 \in P$ be a point lying on an edge $(u_1,v_1) \in E$, at a distance $\ell_{p_1}$ from $u_1$, and $p_2 \in P$ be a point lying on an edge $(u_2,v_2) \in E$, at a distance $\ell_{p_2}$ from $u_2$.
	
	The Straightness $S(p_1,p_2)$ between $p_1$ and $p_2$ is defined as follows:
	\begin{itemize}
		\item If $d_G(p_1,p_2) = +\infty$, then $S(p_1,p_2) = 0$.
		\item If $p_1$ and $p_2$ lie on the same edge, then $S(p_1,p_2) = 1$.
		\item Otherwise, in the general case:
	\end{itemize}
	\begin{equation}
		S(p_1,p_2) = 
		\begin{cases}
	    	f_{u_1u_2}(\ell_{p_1},\ell_{p_2}) ,& \text{if } (\ell_{p_1} \leq \lambda_{u_2}, \lambda_{v_2} \wedge \ell_{p_2} \leq \lambda_{p_1}^{(1)}) \vee (\lambda_{v_2} < \ell_{p_1} \leq \lambda_{u_2} \wedge \ell_{p_2} \leq \lambda_{p_1}^{(2)}), \\
	    	f_{u_1v_2}(\ell_{p_1},\ell_{p_2}) ,& \text{if } (\ell_{p_1} \leq \lambda_{u_2}, \lambda_{v_2} \wedge \ell_{p_2} > \lambda_{p_1}^{(1)}) \vee( \lambda_{u_2} < \ell_{p_1} \leq \lambda_{v_2} \wedge \ell_{p_2} > \lambda_{p_1}^{(3)}), \\
		    f_{v_1u_2}(\ell_{p_1},\ell_{p_2}) ,& \text{if } (\lambda_{u_2} < \ell_{p_1} \leq \lambda_{v_2} \wedge \ell_{p_2} \leq \lambda_{p_1}^{(3)}) \vee (\ell_{p_1} > \lambda_{u_2}, \lambda_{v_2} \wedge \ell_{p_2} \leq \lambda_{p_1}^{(4)}), \\
		    f_{v_1v_2}(\ell_{p_1},\ell_{p_2}) ,& \text{if } (\lambda_{v_2} < \ell_{p_1} \leq \lambda_{u_2} \wedge \ell_{p_2} > \lambda_{p_1}^{(2)}) \vee (\ell_{p_1} > \lambda_{u_2}, \lambda_{v_2} \wedge \ell_{p_2} > \lambda_{p_1}^{(4)}). \\
		\end{cases}
		\label{eq:StraightnessRelative}
	\end{equation}
\end{theorem}

\begin{proof}[Proof of Theorem~\ref{lem:StraightnessRelative}]
	Let us first consider the particular cases. If $p_1$ and $p_2$ are disconnected, the Straightness is $0$ by definition (see Definition~\ref{def:Straightness}). If $p_1$ and $p_2$ lie on the same edge, we get $d_G(p_1,p_2) = d_E(p_1,p_2)$ by Lemma~\ref{lem:GraphDistancePoints}, so $S(p_1,p_2) = 1$. Note that this case subsumes $p_1 = p_2$ from Definition~\ref{def:Straightness}.
	
	We finally focus on the general case. Using Lemmas~\ref{lem:EuclideanDistance} and~\ref{lem:GraphDistancePoints} to replace $d_E$ and $d_G$ in Definition~\ref{def:Straightness} yields Theorem~\ref{lem:StraightnessRelative}.
\end{proof}

\section{Derivation of the Continuous Average Straightness}
\label{sec:ContinuousAverageStraightness}
In this section, we derive several continuous versions of the average Straightness. We take advantage of the reformulation of the Straightness we proposed in the previous section to average it through integration instead of summation. We present $5$ variants of the continuous average Straightness: the first two ones allow to characterize a point of interest (Section~\ref{sec:StraightnessPointLink}), the next two ones focus on edges instead and the last one on the whole graph (Section~\ref{sec:StraightnessLinkLink}).

\subsection{Average Straightness Relatively to a Point}
\label{sec:StraightnessPointLink}
In this section, we focus on measures expressed relatively to a point of interest. We first process the vertex-to-edge total Straightness before deriving two different average values.

\begin{lemma}[Total Straightness between a point and an edge]
	\label{lem:TotalStrPtLk}
	Let $p_1 \in P$ be a point lying on an edge $(u_1,v_1) \in E$ at a distance $\ell_{p_1}$ of $u_1$, and $(u_2,v_2) \in E$ be an edge. 
	
	The total Straightness $\hat{S}_{u_2v_2}(p_1)$ between $p_1$ and $(u_2,v_2)$ is:
	\begin{itemize}
		\item If $p_1$ lies on $(u_2,v_2)$, then $\hat{S}_{u_2v_2}(p_1) = d_E(u_2,v_2)$.
		\item If there is no path between $p_1$ and $(u_2,v_2)$, then $\hat{S}_{u_2v_2}(p_1) = 0$.
		\item Otherwise, in the general case:
	\end{itemize}	 
	\begin{align}
		\hat{S}_{u_2v_2}(p_1) = 
		\begin{cases}
			\displaystyle \int_0^{\lambda_{p_1}^{(1)}} f_{u_1u_2}(\ell_{p_1},\ell_{p_2}) \mathrm{d}\ell_{p_2} + \int_{\lambda_{p_1}^{(1)}}^{d_E(u_2,v_2)} f_{u_1v_2}(\ell_{p_1},\ell_{p_2}) \mathrm{d}\ell_{p_2},& \text{if } \ell_{p_1} \leq \lambda_{u_2}, \lambda_{v_2}, \\
			\displaystyle \int_0^{\lambda_{p_1}^{(2)}} f_{u_1u_2}(\ell_{p_1},\ell_{p_2}) \mathrm{d}\ell_{p_2} + \int_{\lambda_{p_1}^{(2)}}^{d_E(u_2,v_2)} f_{v_1v_2}(\ell_{p_1},\ell_{p_2}) \mathrm{d}\ell_{p_2},& \text{if } \lambda_{v_2} < \ell_{p_1} \leq \lambda_{u_2}, \\
			\displaystyle \int_0^{\lambda_{p_1}^{(3)}} f_{v_1u_2}(\ell_{p_1},\ell_{p_2}) \mathrm{d}\ell_{p_2} + \int_{\lambda_{p_1}^{(3)}}^{d_E(u_2,v_2)} f_{u_1v_2}(\ell_{p_1},\ell_{p_2}) \mathrm{d}\ell_{p_2},& \text{if } \lambda_{u_2} < \ell_{p_1} \leq \lambda_{v_2}, \\
			\displaystyle \int_0^{\lambda_{p_1}^{(4)}} f_{v_1u_2}(\ell_{p_1},\ell_{p_2}) \mathrm{d}\ell_{p_2} + \int_{\lambda_{p_1}^{(4)}}^{d_E(u_2,v_2)} f_{v_1v_2}(\ell_{p_1},\ell_{p_2}) \mathrm{d}\ell_{p_2},& \text{if } \ell_{p_1} > \lambda_{u_2}, \lambda_{v_2}.
		\end{cases}
		\label{eq:TotalStrPtLk}
	\end{align}
\end{lemma}

\begin{proof}[Proof of Lemma~\ref{lem:TotalStrPtLk}]
	The definition of the total Straightness $\hat{S}_{u_2v_2}(p_1)$ between a point $p_1 \in P$ and an edge $(u_2,v_2) \in E$ is:
	\begin{equation}
		\hat{S}_{u_2v_2}(p_1) = \int_0^{d_E(u_2,v_2)} f(\ell_{p_1},\ell_{p_2}) \mathrm{d}\ell_{p_2},
		\label{eq:AvgStrInt}
	\end{equation}
	where $p_2$ is the point used to integrate over $(u_2,v_2)$ and $\ell_{p_2}$ is its distance to $u_2$.
	
	We first consider the two particular cases. According to Theorem~\ref{lem:StraightnessRelative}, if $p_1$ lies on $(u_2,v_2)$, then $S(p_1,p_2)=1$ for any $p_2$ on $(u_2,p_2)$. Consequently, $\hat{S}_{u_2v_2}(p_1) = d_E(u_2,v_2)$. By the same Lemma, if $p_1$ and $(u_2,v_2)$ are not connected by any path, then $S(p_1,p_2)=0$ for any $p_2$ on $(u_2,p_2)$. Thus, $\hat{S}_{u_2v_2}(p_1) = 0$.
	
	We now focus on the general case. Let the break-even distance of $(u_2,v_2)$ for $p_1$ be $\lambda_{p_1}$. We first consider the integrability of $S(p_1,p_2)$ in this case. Suppose that $p_1$ is fixed whereas $p_2$ lies somewhere on the edge $(u_2,v_2)$. Then, $S(p_1,p_2)$ is piecewise (with two pieces) for $\ell_{p_2} \in [0, d_E(u_2,v_2)]$. This is due to it being the ratio of the Euclidean distance to the graph distance, which as mentioned before, is itself two-pieced. Moreover, $S(p_1,p_2)$ is continuous on the same interval, since it is the ratio of two continuous functions, whose denominator cannot be zero ($d_G(p_1,p_2)=0$ is handled by the second specific case in Theorem~\ref{lem:StraightnessRelative}, i.e. $p_1$ and $p_2$ lying on the same edge). 
The Straightness is therefore a continuous piecewise function on $[0, d_E(u_2,v_2)]$, so it can be integrated, as the sum of the integrals of its pieces:

	\begin{equation}
		\hat{S}_{u_2v_2}(p_1) = \int_0^{\lambda_{p_1}} S(p_1,p_2) \mathrm{d}\ell_{p_2} + \int_{\lambda_{p_1}}^{d_E(u_2,v_2)} S(p_1,p_2) \mathrm{d}\ell_{p_2}.
	\end{equation}
	
	Let the break-even distances of $(u_1,v_1)$ for $u_2$ and $v_2$ be $\lambda_{u_2}$ and $\lambda_{v_2}$, respectively, and $\lambda_{p_1}^{(i)}$ denote the value associated to the $i^{th}$ case of Lemma~\ref{lem:BreakEvenDistancePoint}. Theorem~\ref{lem:StraightnessRelative} tells us that the exact expression of $\lambda_{p_1}$ depends on the relative values of $\lambda_{u_2}$, $\lambda_{v_2}$ and $\ell_{p_1}$. This leads to the following $4$ cases:
	\begin{equation}
		\hat{S}_{u_2v_2}(p_1) = 
		\begin{cases}
			\displaystyle \int_0^{\lambda_{p_1}^{(1)}} S(p_1,p_2) \mathrm{d}\ell_{p_2} + \int_{\lambda_{p_1}^{(1)}}^{d_E(u_2,v_2)} S(p_1,p_2) \mathrm{d}\ell_{p_2},& \text{if } \ell_{p_1} \leq \lambda_{u_2}, \lambda_{v_2}, \\
			\displaystyle \int_0^{\lambda_{p_1}^{(2)}} S(p_1,p_2) \mathrm{d}\ell_{p_2} + \int_{\lambda_{p_1}^{(2)}}^{d_E(u_2,v_2)} S(p_1,p_2) \mathrm{d}\ell_{p_2},& \text{if } \lambda_{v_2} < \ell_{p_1} \leq \lambda_{u_2}, \\
			\displaystyle \int_0^{\lambda_{p_1}^{(3)}} S(p_1,p_2) \mathrm{d}\ell_{p_2} + \int_{\lambda_{p_1}^{(3)}}^{d_E(u_2,v_2)} S(p_1,p_2) \mathrm{d}\ell_{p_2},& \text{if } \lambda_{u_2} < \ell_{p_1} \leq \lambda_{v_2}, \\
			\displaystyle \int_0^{\lambda_{p_1}^{(4)}} S(p_1,p_2) \mathrm{d}\ell_{p_2} + \int_{\lambda_{p_1}^{(4)}}^{d_E(u_2,v_2)} S(p_1,p_2) \mathrm{d}\ell_{p_2},& \text{if } \ell_{p_1} > \lambda_{u_2}, \lambda_{v_2}.
		\end{cases}
		\label{eq:AvgStrProof2}
	\end{equation}
	 
	 Finally, $S(p_1,p_2)$ also depends on $\lambda_{u_2}$, $\lambda_{v_2}$ and $\ell_{p_1}$. Replacing the $S(p_1,p_2)$ terms in Equation (\ref{eq:AvgStrProof2}) by the appropriate instances of the auxiliary function $f$ from Theorem~\ref{lem:StraightnessRelative} yields the expression given in Lemma~\ref{lem:TotalStrPtLk}.
\end{proof}

Based on the total values from Lemma~\ref{lem:TotalStrPtLk}, we can define the average Straightness between a point of interest and all the points constituting some edge, as well as all the points constituting the graph.

\begin{definition}[Average Straightness between a point and an edge]
	\label{def:AvgStrPtLk}
	The continuous average Straightness between a vertex $p \in P$ and all the points constituting an edge $(u,v) \in E$ is:
	\begin{equation}
		S_{uv}(p) = \frac{\hat{S}_{uv}(p)}{d_E(u,v)}.
	\end{equation}
\end{definition}

As illustrated later in the experimental evaluation (Section~\ref{sec:ComparisonNode}), this particular average value is convenient to characterize the accessibility of an edge from one point of the graph.

\begin{definition}[Average Straightness between a point and the graph]
	\label{def:AvgStrPtGr}
	The continuous average Straightness between a vertex $p \in P$ and all the points constituting the graph $G$ is:
	\begin{equation}
		S_G(p) = \frac{\sum_{(u,v) \in E} \hat{S}_{uv}(p)}{\sum_{(u,v) \in E} d_E(u,v)}.
	\end{equation}
\end{definition}

This average value is particularly interesting, because it can be used as some kind of vertex centrality measure, representing how accessible the considered vertex is from the rest of a the graph.

\subsection{Average Straightness Relatively to an Edge}
\label{sec:StraightnessLinkLink}
In this section, we focus on average measures expressed relatively to an edge of interest. Like before, we first process the total Straightness, this time between two edges, before deriving several distinct average values.

\begin{lemma}[Total Straightness between two edges]
	\label{lem:TotalStrLkLk}
	Let $(u_1,v_1) \in E$ and $(u_2,v_2) \in E$ be two edges.
	
	The total Straightness $\hat{S}_{u_2v_2}(u_1,v_1)$ between $(u_1,v_1)$ and $(u_2,v_2)$ is:
	\begin{itemize}
		\item If $(u_1,v_1) = (u_2,v_2)$, then $\hat{S}_{u_2v_2}(u_1,v_1) = d_E(u_1,v_1)^2/2$.
		\item If there is no path between $(u_1,v_1)$ and $(u_2,v_2)$, then $\hat{S}_{u_2v_2}(u_1,v_1) = 0$.
		\item Otherwise, in the general case:
	\end{itemize}
	\begin{align}
		\begin{split}
		&\hat{S}_{u_2v_2}(u_1,v_1) = \\
		&\quad\begin{cases}
			\displaystyle \int_0^{\lambda_{u_2}} \hat{S}_{u_2v_2}(p_1) \mathrm{d}\ell_{p_1} + \int_{\lambda_{u_2}}^{\lambda_{v_2}} \hat{S}_{u_2v_2}(p_1) \mathrm{d}\ell_{p_1} + \int_{\lambda_{v_2}}^{d_E(p_1)} \hat{S}_{u_2v_2}(p_1) \mathrm{d}\ell_{p_1},& \text{if } \lambda_{u_2} \leq \lambda_{v_2}, \\
			\displaystyle \int_0^{\lambda_{v_2}} \hat{S}_{u_2v_2}(p_1) \mathrm{d}\ell_{p_1} + \int_{\lambda_{v_2}}^{\lambda_{u_2}} \hat{S}_{u_2v_2}(p_1) \mathrm{d}\ell_{p_1} + \int_{\lambda_{u_2}}^{d_E(p_1)} \hat{S}_{u_2v_2}(p_1) \mathrm{d}\ell_{p_1},& \text{if } \lambda_{u_2} > \lambda_{v_2}. \\
		\end{cases}
		\end{split}
		\label{eq:TotalStrLkLk}
	\end{align}
\end{lemma}

Being a double integration of the point-to-point Straightness over both considered edges, $\hat{S}_{u_2v_2}(u_1,v_1)$ is symmetric with respect to these edges, i.e. $\hat{S}_{u_2v_2}(u_1,v_1)=\hat{S}_{u_1v_1}(u_2,v_2)$. 

\begin{proof}[Proof of Lemma~\ref{lem:TotalStrLkLk}]
	The definition of the total Straightness $\hat{S}_{u_2v_2}(u_1,v_1)$ between two edges $(u_1,v_1)$ and $(u_2,v_2)$ is:
	\begin{equation}
		\hat{S}_{u_2v_2}(u_1,v_1) = \int_0^{d_E(u_1,v_1)} \hat{S}_{u_2v_2}(p_1) \mathrm{d}\ell_{p_1},
		\label{eq:AvgStrInt2}
	\end{equation}
	where $p_1$ is the point used to integrate over $(u_1,v_1)$ and $\ell_{p_1}$ is its distance to $u_1$.
	
	We first consider the two particular cases. According to Lemma~\ref{lem:TotalStrPtLk}, if $(u_1,v_1) = (u_2,v_2)$, then $\hat{S}_{u_2v_2}(p_1)=d_E(u_2,v_2)$ for any $p_1$ on $(u_1,p_1)$. Consequently, the integral amounts to $d_E(u_1,v_1) \cdot d_E(u_2,v_2)$. However, each pair of vertices is accounted for twice, since $(u_1,v_1)$ and $(u_2,v_2)$ are the same edge. For this reason, we take only half this product and get $\hat{S}_{u_2v_2}(u_1,v_1) = d_E(u_1,v_1) \cdot d_E(u_2,v_2)/2 = d_E(u_1,v_1)^2/2$. By the same Lemma, if $p_1$ and $(u_2,v_2)$ are not connected by any path, then $S_{u_2,v_2}(p_1)=0$ for any $p_1$ on $(u_1,p_1)$. Thus, $\hat{S}_{u_2v_2}(u_1,v_1) = 0$, too.
	
	We now focus on the general case. Let $\lambda_{u_2}$ and $\lambda_{v_2}$ be the break-even distances of $(u_1,v_1)$ for $u_2$ and $v_2$, respectively. From Lemma~\ref{lem:TotalStrPtLk}, we know that $S_{u_2v_2}(p_1)$ is a three-part piecewise function. It can therefore be integrated as the sum of the integrals of its three pieces. Moreover, Lemma~\ref{lem:TotalStrPtLk} also tells us it can take two different forms, depending on how $\lambda_{u_2}$ compares to $\lambda_{v_2}$ (i.e. either $\lambda_{u_2} \leq \lambda_{v_2}$ or $\lambda_{u_2} > \lambda_{v_2}$). This leads to the expression given in Lemma~\ref{lem:TotalStrLkLk}.
\end{proof}

\begin{definition}[Average Straightness between two edges]
	\label{def:AvgStrLkLk}
	The continuous average Straightness between all the points constituting an edge $(u_1,v_1) \in E$ on one side, and all the points constituting an edge $(u_2,v_2) \in E$ on the other side, is:
	\begin{equation}
		S_{u_2v_2}(u_1,v_1) =
		\begin{cases}
			\displaystyle \frac{\hat{S}_{u_2v_2}(u_1,v_1)}{d_E(u_1,v_1)^2/2} = 1,& \text{if } (u_1,v_1) = (u_2,v_2),\\
			\displaystyle \frac{\hat{S}_{u_2v_2}(u_1,v_1)}{d_E(u_1,v_1)d_E(u_2,v_2)},& \text{otherwise.}
		\end{cases}
	\end{equation}
\end{definition}

When we consider the straightness between an edge and itself, the division by 2 prevents from counting each pair of points twice. Note that like $\hat{S}_{u_2v_2}(u_1,v_1)$, the measure $S_{u_2v_2}(u_1,v_1)$ is symmetric relatively to the concerned edges, i.e. $S_{u_2v_2}(u_1,v_1)=S_{u_1v_1}(u_2,v_2)$. As illustrated later in Section~\ref{sec:ComparisonNode}, $S_{u_2v_2}(u_1,v_1)$ can be used to characterize the mutual accessibility of the considered edges.

\begin{definition}[Average Straightness between an edge and the graph]
	\label{def:AvgStrLkGr}
	The continuous average Straightness between an edge $(u_1,v_1) \in E$ and all the points constituting the graph $G$ is:
	\begin{equation}
		S_G(u_1,v_1) = 
			\frac{\displaystyle \sum_{(u_2,v_2) \in E} \hat{S}_{u_2v_2}(u_1,v_1)}
			{\displaystyle  \sum_{(u_2,v_2) \in E} d_E(u_2,v_2)d_E(u_1,v_1) - d_E(u_1,v_1)^2/2}.
	\label{eq:AvgStrLkGr}
	\end{equation}
\end{definition}

Like with Definition~\ref{def:AvgStrPtGr}, this average Straightness can be used as a centrality measure, but this time to describe edges in place of vertices. The value $S_{G}(u_1,v_1)$ represents the accessibility of the edge from the points constituting the whole graph. 

Note that Equation (\ref{eq:AvgStrLkGr}) includes the points constituting the edge of interest itself (i.e. the Straightness between the edge and itself). If one wants to ignore them, one should use the following expression instead:
\begin{equation}
	S_G(u_1,v_1) = 
		\frac{\displaystyle \sum_{(u_2,v_2) \neq (u_1,v_1)} \hat{S}_{u_2v_2}(u_1,v_1)}
		{\displaystyle  \sum_{(u_2,v_2) \neq (u_1,v_1)} d_E(u_2,v_2)d_E(u_1,v_1)}.
\end{equation}

\begin{definition}[Average Straightness for the whole graph]
	\label{def:AvgStrGr}
	The continuous average Straightness between all pairs of points constituting the graph $G$ is:
	\begin{equation}
		S_G(G) = 
			\frac{\displaystyle \sum_{(u_1,v_1) \in E} \sum_{(u_2,v_2) \geq (u_1,v_1)} \hat{S}_{u_2v_2}(u_1,v_1)}
			{\displaystyle  \sum_{(u_1,v_1) \in E} \sum_{(u_2,v_2) \geq (u_1,v_1)} d_E(u_2,v_2)d_E(u_1,v_1) - \sum_{(u_1,v_1) \in E} d_E(u_1,v_1)^2/2},
	\label{eq:AvgStrGr}
	\end{equation}
	where $\geq$ corresponds to the lexicographical order between edges.
\end{definition}

Again, we include in Equation (\ref{eq:AvgStrGr}) the paths between points located on the same edge. If one wants to ignore them, one should use instead:
\begin{equation}
	S_G(G) = 
		\frac{\displaystyle \sum_{(u_1,v_1) \in E} \sum_{(u_2,v_2) > (u_1,v_1)} \hat{S}_{u_2v_2}(u_1,v_1)}
		{\displaystyle  \sum_{(u_1,v_1) \in E} \sum_{(u_2,v_2) > (u_1,v_1)} d_E(u_2,v_2)d_E(u_1,v_1)},
\end{equation}
where $>$ corresponds to the lexicographical order between edges.

\section{Algorithmic Complexity and Discrete Approximations}
\label{sec:ComplexityApproximation}
To summarize the previous section: we introduced $5$ distinct continuous average Straightness measures, allowing to characterize various elements of a spatial graph (vertices, edges and the whole graph). 
As a reminder, they are described in Table \ref{tab:Complexities} (top $5$ rows).

In this section, we first derive the algorithmic complexity of these $5$ variants. Then, we describe how the two traditional approaches for averaging the Straightness can be considered as discrete approximations of $S_G(u)$ (average Straightness between a vertex $u$ and the graph $G$) and $S_G(G)$ (average Straightness over the whole graph), and discuss their complexity.

We implemented all our variants of the continuous average Straightness, as well as both discrete approximations, using the R language, and based on the R version of the igraph library \parencite{Csardi2006}. The interest of this library is that it offers an easy access to graph-related operations through its R interface, while providing fast computations thanks to its underlying C implementation. Our source code is publicly available online\footnote{\url{https://github.com/CompNet/SpatialMeasures}}.

\subsection{Continuous Average Straightness}
\label{sec:AlgorithmicComplexity}
\paragraph{Straightness of a Point} 
Let us first consider the simplest of the $5$ variants: $S_{uv}(p)$, i.e. the average Straightness between a point $p$ and an edge $(u,v)$. Processing $S_{uv}(p)$ requires both the Euclidean and graph distances between the point $p$ and $u$ and $v$, the two end-vertices of the edge, as well as the Euclidean distance between $u$ and $v$ (i.e. the edge length). The Euclidean distances can be obtained in constant time. However, processing the graph distances is computationally more demanding. The igraph library uses Dijkstra's algorithm for this purpose, and the implementation it proposes runs in $O(m \log m + n)$ according to its documentation, 
where $n$ and $m$ are the numbers of vertices and edges in the graph, respectively. 
The next step is to integrate the auxiliary function $f$ (Definition \ref{def:AuxiliaryFunction}) over the edge, which can be done in constant time. Moreover, we used Wolfram Mathematica to process the closed form of $F$, its antiderivative, and implemented it as an R function in order to speed up the calculations, so this step is very short in practice. The total time complexity is therefore of $O(m \log m + n)$ for this measure, while the space complexity is $O(m + n)$, i.e. the memory required to represent the graph. It is interesting to remark that if the graph is \textit{sparse} (i.e. $m=O(n)$), which (by Euler's Formula) is the case of planar graphs, then these expressions become $O(n \log n)$ and $O(n)$, respectively.

When processing $S_G(p)$, the average Straightness between a point $p$ and the whole graph $G$, we can repeat roughly the same process over all edges. So, for a graph containing $m$ edges, we get a time complexity of $O(m(m \log m + n))=O(m^2 \log m + mn)$ (or $O(n^2 \log n)$ for a sparse graph), while the space complexity stays $O(m + n)$. However, it is worth noticing that, with this approach, some of the distances will be processed several times (since a vertex can be attached to several distinct edges), and as we already mentioned, this has a high computational cost. So, instead of processing the graph distances between $p$ and the end-vertices on-demand, i.e. at each of the $m$ iterations, an alternative approach consists in computing them separately, and once and for all. This can also be done using Dijkstra's algorithm, resulting in a time complexity of $O((m \log m + n) + m)=O(m \log m + n)$ (or $O(n \log n)$ for a sparse graph). The space complexity is the same, although in practice this method requires more memory, since we need to store all $n$ distance values. Our implementation proposes both these versions, since the slower one allows handling large graphs, for which the faster one takes too much memory. As we will see, the same choice can be made for all the other average Straightness variants involving loops over vertices or edges.
 
\paragraph{Straightness of an Edge} 
We now focus on $S_{u_2v_2}(u_1,v_1)$, the average Straightness between two edges $(u_1,v_1)$ and $(u_2,v_2)$. We need the Euclidean and graph distances between their end-vertices: like before, this step has a time complexity of $O(m \log m + n)$. Because we must now average over two edges, function $f$ has to be integrated twice. As mentioned before, we already have the closed form of $F$, its antiderivative. However, we were not able to obtain the antiderivative of $F$ itself, which is why we used a numerical approach to perform this integration, through the default R \texttt{integrate} function. Note that, from the computational complexity perspective, this numerical integration is still considered to be performed in constant time, but in practice it takes significantly longer. In total, we have a time complexity of $O(m \log m + n)$ for this measure (the same than for $S_{uv}(p)$, although it actually takes longer to process) and a space complexity of $O(m + n)$. If the graph is sparse, we respectively get $O(n \log n)$ and $O(n)$ instead.

When processing $S_G(u,v)$, the average Straightness between an edge $(u,v)$ and the graph $G$, we roughly repeat the same process for all $m$ edges. So, its time complexity is $O(m(m \log m + n))=O(m^2 \log m + mn)$ (or $O(n^2 \log n)$ if the graph is sparse), i.e. it is similar to that of $S_G(p)$. The space complexity is also similar, i.e. $O(m+n)$. Like before, if we have enough memory to process the graph distances separately, the time complexity can be reduced to $O(m \log m + n)$ (or $O(n \log n)$ for sparse graphs). This approach requires storing all $2n$ additional values, but this does not increase the asymptotic space complexity. 

\paragraph{Straightness of the Graph} 
Finally, processing $S_G(G)$, the average Straightness over the whole graph $G$, is also a similar process, except this time the sum is performed over $O(m^2)$ pairs of edges. This results in a time complexity of $O(m^2 (m \log m + n))=O(m^3 \log m + m^2n)$ (or $O(n^3 \log n)$ for sparse graphs). The space complexity stays the same, i.e. $O(m + n)$ in general, and $O(n)$ for sparse graphs. Like before, we can speed things up by computing separately the vertex-to-vertex distances. However, this time we need to consider $O(n^2)$ pairs of vertices. The igraph libray can process them in $O(mn)$, according to its documentation, 
leading to a total time complexity of $O(m^2 + mn)$ (or $O(n^2)$ for a sparse graph). Memory-wise, igraph needs a $n \times n$ matrix to store all the distance values, resulting in a total space complexity of $O(m + n^2)$ (or $O(n^2)$ for sparse graphs).
The first $5$ rows of Table \ref{tab:Complexities} summarize these calculations \textit{for sparse graphs}.

\begin{table}[!th]
    \caption{Summary of the considered measures: the $5$ variants of continuous average Straightness proposed in Section~\ref{sec:ContinuousAverageStraightness}, and the $2$ discrete approximations described in Section~\ref{sec:DiscreteApproximations}. The algorithmic complexity of certain measures is described by two expressions: the second corresponds to the case where the graph distances are processed separately (as explained in the text, this speeds up the calculations but requires more memory). All expressions describe the complexity for \textit{sparse} graphs (the general case is described in the text).}
	\label{tab:Complexities}
	\centering
	\begin{tabular*}{\textwidth}{@{\extracolsep{\fill}}l@{}p{8cm}@{}p{2.5cm}@{}p{1.5cm}@{}}
		\hline
		Measure & Description & Complexity & \\
		 &  & Time & Space \\
		\hline
		$S_{uv}(p)$ & Continuous average Straightness between a point $p$ and an edge $(u,v)$ (Definition~\ref{def:AvgStrPtLk}). & $O(n \log n)$ & $O(n)$ \\
		$S_G(p)$ & Continuous average Straightness between a point $p$ and the graph $G$ (Definition~\ref{def:AvgStrPtGr}). & $O(n^2 \log n)$ \newline $O(n \log n)$ & $O(n)$ \newline $O(n)$ \\
		$S_{u_2v_2}(u_1,v_1)$ & Continuous average Straightness between two edges $(u_1,v_1)$ and $(u_2,v_2)$ (Definition~\ref{def:AvgStrLkLk}). & $O(n \log n)$ & $O(n)$ \\
		$S_G(u,v)$ & Continuous average Straightness between an edge $(u,v)$ and the graph $G$ (Definition~\ref{def:AvgStrLkGr}). & $O(n^2 \log n)$ \newline $O(n \log n)$ & $O(n)$ \newline $O(n)$ \\
		$S_G(G)$ & Continuous average Straightness over all pairs of points constituting the graph $G$ (Definition~\ref{def:AvgStrGr}). & $O(n^3 \log n)$ \newline $O(n^2)$ & $O(n)$ \newline $O(n^2)$ \\
		\hline
		$\sigma_\theta(u)$ & Discrete approximation of $S_G(p)$ (Section~\ref{sec:DiscreteApproximations}). & $O(\theta^2 n^2 \log \theta n)$ \newline $O(\theta n \log \theta n)$ & $O(\theta n)$ \newline $O(\theta n)$ \\
		$\sigma_\theta(G)$ & Discrete approximation of $S_G(G)$ (Section~\ref{sec:DiscreteApproximations}). & $O(\theta^3 n^3 \log \theta n)$ \newline $O(\theta^2 n^2)$ & $O(\theta n)$ \newline $O(\theta^2 n^2)$ \\
		\hline
	\end{tabular*}
\end{table}

\subsection{Discrete Approximations}
\label{sec:DiscreteApproximations}
As explained in the introduction, by comparison with the continuous point-to-point approach that we propose in this article, the traditional approach consists in averaging vertex-to-vertex Straightness values, either by fixing one vertex $u$ and considering all the others in the graph $G$ (which allows characterizing this vertex), or by considering all pairs of vertices in the graph or a subgraph (which allows characterizing the graph or subgraph). For convenience, we note these measures $\sigma(u)$ and $\sigma(G)$, respectively. They can be considered as very rough approximations of $S_G(u)$ and $S_G(G)$, respectively, in which a vertex is used as a proxy for the whole edge to which it is attached. However, note that the three other average measures which we introduced have no counterparts in previous works. 


\paragraph{Approximation Method} 
\textcite{Josselin2016} proposed a method allowing to obtain more reliable approximations than $\sigma(u)$ and $\sigma(G)$. We use it in the experimental part of this article (Section \ref{sec:ComparisonDiscrete}). This method is based on the discretization of the original edges. First, additional $2$-degree vertices are added on the graph edges, in a way such that no remaining edge is longer than some predefined parameter $\varepsilon$. This amounts to splitting the edges into shorter pieces, resulting in a modified graph $G'=(V',E')$ (with $V \subseteq V'$). Second, the classic vertex-to-vertex average Straightness values $\sigma(u)$ and $\sigma(G')$ are processed on this modified graph. For the former, one fixes a vertex $u \in V$ in the original graph $G$, averages over $V'$ in the modified graph $G'$ and gets an approximation of $S_G(u)$. For the latter, one averages over all pairs of vertices in $V'$ and obtains an approximation of $S_G(G)$.

The smaller the $\varepsilon$ value and the better the approximation, but this also has an effect on the processing time and memory usage, since there are more and more additional vertices and edges to handle. In the following, instead of using $\varepsilon$ to characterize the discretization, we prefer to use the \textit{average edge segmentation}, noted $\theta$, which corresponds to the average number of times the original edges are split. For instance, a value of $\theta = 4$ means an edge of the original graph $G$ is split in $4$ segments, in average, in the modified graph $G'$. We note $\sigma_\theta(u)$ and $\sigma_\theta(G)$ the discrete approximations of $S_G(u)$ and $S_G(G)$ obtained with an \textit{average edge segmentation} of $\theta$, respectively. According to this notation, we have $\sigma_0(u) = \sigma(u)$ and $\sigma_0(G) = \sigma(G)$, respectively. In the rest of this section, we derive the algorithmic complexity of these measures by first considering separately the discretization process and the computation of $\sigma(u)$ and $\sigma(G')$ on the resulting modified graph, before combining them to get the total expressions.

\paragraph{Complexity of the Discretization Process} 
The process of discretizing the graph edges requires to split each edge into $\theta$ smaller parts. For this purpose, the first step is to build a new graph containing only the original vertices (without any edge). Then, the second step would be to iteratively process each original edge, inserting new vertices at the splitting points and adding the corresponding new edges to connect them. 

However, with the igraph library, the asymptotic complexity of inserting new objects in the graph does not depend on the number of inserted objects, be it vertices or edges. Inserting one or several new vertices runs in $O(n)$, 
whereas for new edges it is $O(m+n)$, 
according to the official documentation. So, instead of inserting new vertices and edges at each iteration, it is faster to iteratively build the lists of new vertices and edges for the whole graph, and insert them all simultaneously, once all the original edges have been processed. This way, the time complexity of the whole discretization process is $O(m + n + (m+n)))=O(m + n)$. The graph $G'$ resulting from the discretization contains $n' = n+(\theta-1) m$ vertices and $m' = \theta m$ edges, so we have $n' = O(n +\theta m)$ and $m' = O(\theta m)$. The discretization process requires to represent simultaneously the original ($G$) and modified ($G'$) graphs, as well as the lists of new vertices and edges, so its space complexity is $O((n+m)+((n+\theta m) + \theta m))=O(n + \theta m)$.

\paragraph{Complexity of the Averaging Process} 
Processing the Straightness between two vertices requires computing both the graph and Euclidean distances between them. As explained before, the former can be handled in constant time, whereas the latter has a $O(m \log m + n)$ time complexity in igraph. Processing $\sigma(u)$, the discrete average Straightness between a fixed vertex $u$ and all the other vertices in the original graph $G$, consists in repeating the same process $n$ times, so we get a total time complexity of $O(n(m \log m + n)) = O(mn \log m + n^2)$. The space complexity is $O(m + n)$, which corresponds to the memory required to represent the graph. Like for the continuous average measures, if the available memory is large enough, one can process all the necessary distances separately, using Dijkstra's algorithm. This leads to a total time complexity of $O((m \log m + n) + n)=O(m \log m + n)$. Memory-wise, we need to store $n$ distance values, in addition to the graph, so the memory usage increases in practice, but the asymptotic space complexity stays $O(m + n)$. 

For $\sigma(G)$, the discrete average Straightness over the \textit{whole} graph, the average must be processed over $O(n^2)$ pairs of vertices, resulting in a time complexity of $O(n^2(m \log m + n)) = O(mn^2 \log m + n^3)$, while the space complexity is $O(m + n)$. If there is enough memory to process all vertex-to-vertex graph distances separately, the time complexity can be decreased to $O((m \log m + n) + n^2) = O(m \log m + n^2)$ whereas the space complexity increases to $O(m + n^2)$.

\paragraph{Total Complexity} 
Based on our previous calculations, we can now derive the total complexity of $\sigma_\theta(u)$ and $\sigma_\theta(G)$: we add the cost of the discretization process to that of the processing of $\sigma(u)$ or $\sigma(G')$ \textit{on the modified graph} $G'$. 
For $\sigma_\theta(u)$, we get a time complexity of $O((m+n) + (m'n' \log m' + n'^2)) = O((n + \theta m) \theta m \log \theta m + n^2)$ and a space complexity of $O((\theta m + n) + (m' + n')) = O(\theta m + n)$. When considering sparse graphs, these expressions can be simplified to $O(\theta^2 n^2 \log \theta n)$ and $O(\theta n)$, respectively.
If the distances are processed separately, the time complexity becomes 
$O((m+n) + (m' \log m' + n')) = O(\theta m \log \theta m + n)$ (or $O(\theta n \log \theta n)$ for sparse graphs), whereas the space complexity stays the same.

For $\sigma_\theta(G)$, the time complexity is $O((m+n) + (m'n'^2 \log m' + n'^3)) = O((\theta m n^2 + \theta^3 m^3) \log \theta m + n^3)$ and the space complexity is $O((m+n) + (m'+n')) = O(\theta m + n)$. When considering sparse graphs, these expressions can be simplified to $O(\theta^3 n^3 \log \theta n)$ and $O(\theta n)$, respectively.
If the distances are processed separately, the time complexity decreases to $O((m+n) + (m' \log m' + n'^2)) = O(\theta^2 m^2 + n^2)$ and the space complexity increases to $O((m+n) + (m'+n'^2)) = O(\theta^2 m^2 + n^2)$. For sparse graphs, we can simplify both these expressions to $O(\theta^2 n^2)$.
The last $2$ rows of Table \ref{tab:Complexities} summarize the main complexity results for these discrete approximations, on \textit{sparse graphs}.

\subsection{Comparison and Practical Considerations}
In this subsection, we compare the complexities of the continuous average Straightness measures $S_G(p)$ and $S_G(G)$, with those of their respective discrete approximations $\sigma_\theta(u)$ and $\sigma_\theta(G)$. We then discuss some important practical computational aspects which do not appear when only considering the asymptotic complexity.

\paragraph{Complexity Comparison}
We first consider the average Straightness for a fixed vertex. With the method consisting in computing the distances on-demand, we have a time complexity of $O(m^2 \log m + mn)$ for $S_G(u)$ vs. $O((n + \theta m) \theta m \log \theta m + n^2)$ for $\sigma_\theta(u)$. The latter can be rewritten as $O(\theta^2 m^2 \log \theta m + (\theta m \log \theta m +n)n)$, and the discrete approximation has therefore a larger time complexity. When processing the distances separately, we have $O(m \log m + n)$ vs. $O(\theta m \log \theta m + n)$, so the conclusion is the same. The space complexities are $O(m+n)$ vs. $O(\theta m + n)$, independently from the way we handle the distances, so again, the continuous approach has a lower complexity.

We now switch to the average Straightness for the whole graph. When processing the distances on-demand, the time complexity is $O(m^3 \log m + m^2 n)$ for $S_G(G)$ vs. $O((\theta m n^2 + \theta^3 m^3) \log \theta m + n^3)$ for $\sigma_\theta(G)$. Unlike before, these expressions cannot be compared directly. We will consider instead separately the three cases corresponding to the possible orders of magnitude of $m$: $O(1)$ (quasi-empty graph), $O(n)$ (sparse graph) and $O(n^2)$ (dense graph). In the first case, the time complexity of $S_G(G)$ simplifies to $O(n)$ whereas that of $\sigma_\theta(G)$ contains a $n^3$ term. In the second case, we obtain $O(n^3 \log n)$ vs. $O(\theta^3 n^3 \log \theta n)$ (cf. Table \ref{tab:Complexities}). In the third case, it is $O(n^6 \log n^2)$ vs. $O(\theta^3 n^6 \log \theta n^2)$. Therefore, in all cases the complexity is lower for the continuous average. Memory-wise, we have $O(m + n)$ for $S_G(G)$ vs. $O(\theta m + n)$ for $\sigma_\theta(G)$, so the space complexity is also lower for the continuous average.

When processing the distances separately, the time complexities are $O(m^2 + mn)$ for $S_G(G)$ vs. $O(\theta^2 m^2 + n^2)$ for $\sigma_\theta(G)$. By using the same approach as before (based on the three possible orders of magnitude of $m$), we conclude the time complexity of the continuous average is lower. For the space complexity, we have $O(m + n^2)$ vs. $O(\theta^2 m^2 + n^2)$, so the complexity is also lower here for the continuous average. In the end, the asymptotic complexities of the continuous average measures are always lower than for their discrete approximations, be it in terms of processing time or memory usage. 

\paragraph{Memory usage}
It is important to notice that, for all considered measures (including the $5$ continuous variants and $2$ discrete approximations), the most costly step in terms of both computational time and memory usage, is the processing and storage of the graph distances. In comparison, the computation of the average Straightness itself is negligible. As explained before, we can speed up the overall processing by computing the graph distances separately, i.e. before iterating over the concerned vertices or edges. However, this requires to store all these data, which in turn can raise a memory problem. The need to use this extra space does not necessarily affect the asymptotic space complexity, but in practice the memory usage still increases significantly. 

For instance, storing all the vertex-to-vertex distances for a $100,000$ vertices graph requires approximately $75$~GB of RAM on a standard programming environment. And this is in addition to the memory already used to represent the graph itself. Processing distances on-demand allows saving a lot of memory, at the price of some extra processing, since the same distances will be processed several times when iterating over the vertices or edges. The difference is that the memory restriction is a hard limit, whereas one can always decide to let the program run longer. So in practice, memory usage, and not processing time, quickly becomes the limitation preventing from handling large graphs.

\section{Empirical Validation}
\label{sec:EmpiricalValidation}
In this section, we study experimentally our $5$ variants of the continuous average Straightness. First, in Section~\ref{sec:ComparisonDiscrete}, we compare the performance of our continuous average variants with that of their discrete counterparts, in terms of computational time, memory usage and precision. Second, in Section~\ref{sec:ComparisonNode} we describe their behavior on a collection variety of graphs, and highlight how it differ from the traditional vertex-to-vertex average Straightness. We summarize and discuss our main findings in Section~\ref{sec:FinalObs}.

\subsection{Performance Study}
\label{sec:ComparisonDiscrete}
We first empirically evaluate the performance of our continuous approach in terms of computational time and memory usage, and compare it to its discrete approximations. Our study is conducted on artificially generated networks.

\paragraph{Methods} 
We randomly generate spatial graphs using the following method. First, we draw $n$ vertices uniformly distributed over a predefined square. The exact dimension of this square is unimportant, since the Straightness is a ratio of distances. Second, we connect these vertices using two alternative methods: on the one hand, Delaunay's triangulation \parencite{Delaunay1934} to obtain a planar graph, and on the other hand, the Erd\H{o}s-R\'enyi model \parencite{Erdos1959} to get a non-planar one. We do not observe any noticeable difference between the performances obtained on both types of graphs. Moreover, the Straightness is generally used to characterize physical networks modeled by planar graphs, such as road networks. For these reasons, in the rest of this section we only present the detailed results obtained for the planar graphs.

To get representative results, we generate $10$ instances of graphs for $n=10$, $25$, $50$ and $100$. We observe a similar behavior for all the considered graph sizes $n$, so in the rest of this section we only show the $n=50$ plots. For each generated graph, we compute two variants of the average continuous Straightness: $S_G(u)$, which characterizes an individual vertex $u$ by its average Straightness with the rest of the points constituting the graph $G$ (Definition~\ref{def:AvgStrPtGr}), and $S_G(G)$, which characterizes the whole graph by the average Straightness between all pairs of its constituting points (Definition~\ref{def:AvgStrGr}). We also process their respective discrete approximations $\sigma_\theta(u)$ and $\sigma_\theta(G)$ using an increasing average edge segmentation $\theta$. In other words, we compute vertex-to-vertex Straightness with an increasing number of vertices. This includes $\sigma(u)$ and $\sigma(G)$, the discrete average processed on the original graph (i.e. we start with $\theta=0$ then increase $\theta$). As explained in Section~\ref{sec:ComplexityApproximation}, we have proposed two approaches to process each one of these four measures: one is time-efficient, the other is memory-efficient. In this part of our experiments, we used only the fastest method, since the size of the graphs allows it.
\begin{figure}[ht!]
	\centering
	\includegraphics[width=0.49\textwidth]{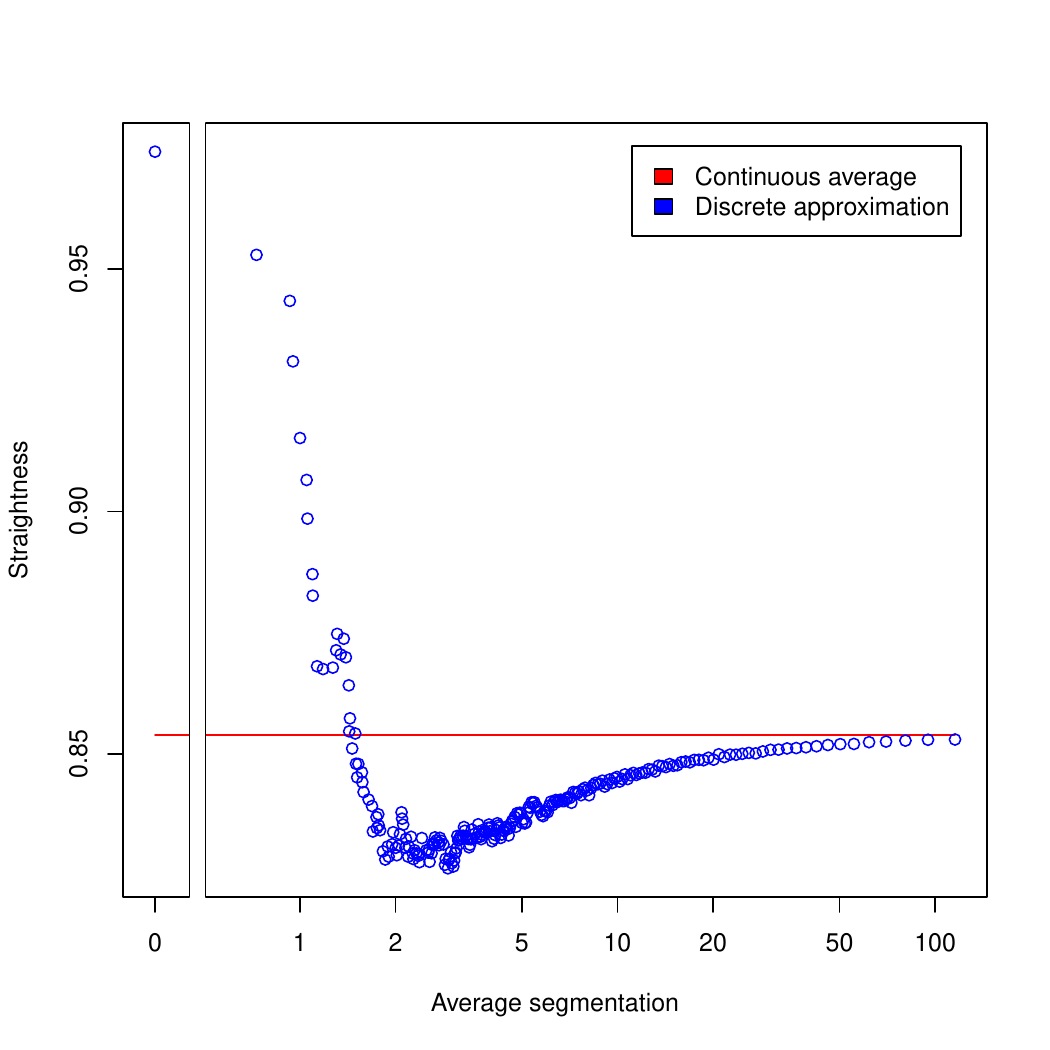}
	\hfill
	\includegraphics[width=0.49\textwidth]{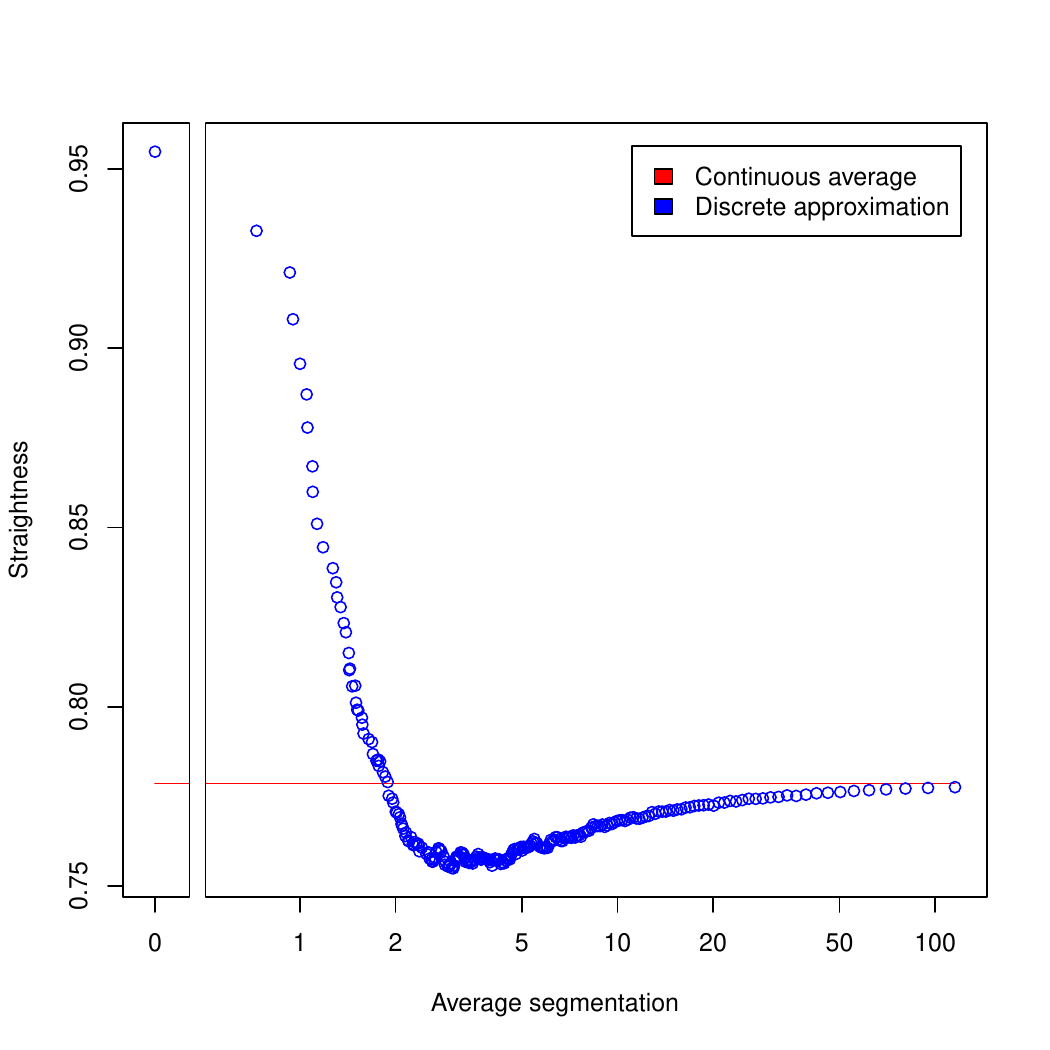}
    \caption{Comparison of the Straightness obtained through the discrete approximation (blue) and the continuous average (red), on random planar graphs. The values are displayed for a given vertex ($S_G(u)$ and $\sigma_\theta(u)$, left) and a given graph ($S_G(G)$ and $\sigma_\theta(G)$, right), as functions of the average edge segmentation $\theta$.}
    \label{fig:ExpRandplanarStraightness}
\end{figure}

\paragraph{Quality of Approximation} 
Figure~\ref{fig:ExpRandplanarStraightness} compares the continuous average Straightness (in red) and its discrete approximation (in blue) for a \textit{single} vertex ($S_G(u)$ and $\sigma_\theta(u)$, left plot) and a \textit{single} graph instance ($S_G(G)$ and $\sigma_\theta(G)$, right plot). These values are plotted as functions of the average edge segmentation $\theta$. Of course, only the discrete approximations depend on this parameter, so the continuous averages take the form of straight lines. Remark that both $x$ axes use a logarithmic scale, which is also the case for the other plots presented in this section. This explains why the values obtained for $\theta=0$ are separated from the rest of the data.

We selected these two specific plots because they are very representative of the other vertices and graphs generated during this experiment. As mentioned before, we ran the same experiments $10$ times, so we could have represented all of them on the same plots. However, doing so is not illustrative at all, since we can get very different Straightness values from one run to the other, due to the stochastic nature of our graph generation process. This is why we first focus on a specific vertex and a specific graph in this figure, whereas the results for all $10$ runs are presented under a different form later, in Figure~\ref{fig:ExpRandplanarDiff}.

Our first observation is that, when averaged over only the original vertices (i.e. for an average edge segmentation of $\theta=0$), $\sigma(u)$ and $\sigma(G)$ are poor approximations of $S_G(u)$ and $S_G(G)$. As explained before, these values are separated from the others on the plots, due to the logarithmic scale. They generally reach very high Straightness values (close to $1$), as illustrated by the examples of Figure~\ref{fig:ExpRandplanarStraightness}. The approximation gets better when we start discretizing the edges: the values quickly decrease. However, they drop to the point where they underestimate the continuous average. Then, they get asymptotically closer to it when the average edge segmentation increases, i.e. when the edges are split into smaller and smaller pieces. Incidentally, this is a good empirical validation of our implementation of the continuous average Straightness, since there seems to be convergence.
\begin{figure}[ht!]
	\centering
	\includegraphics[width=0.49\textwidth]{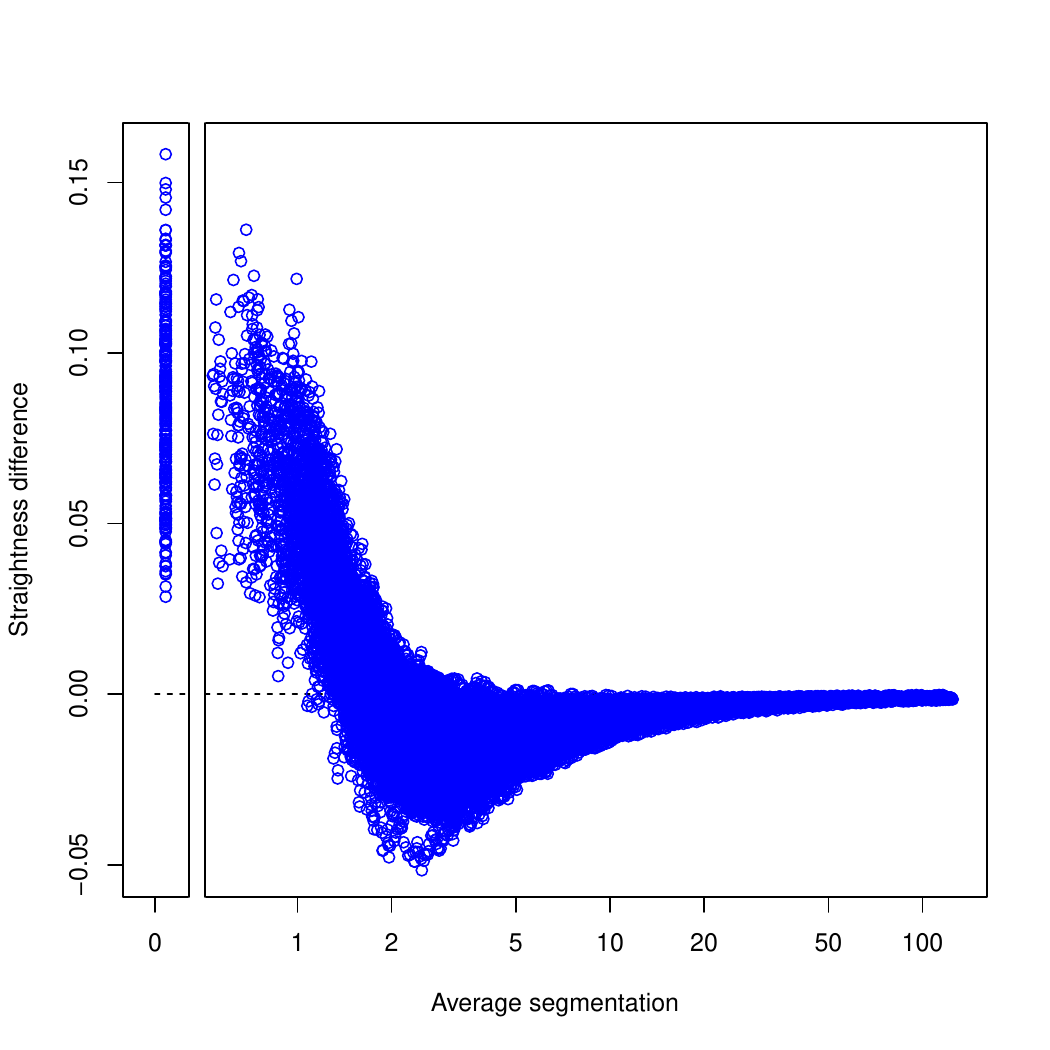}
	\hfill
	\includegraphics[width=0.49\textwidth]{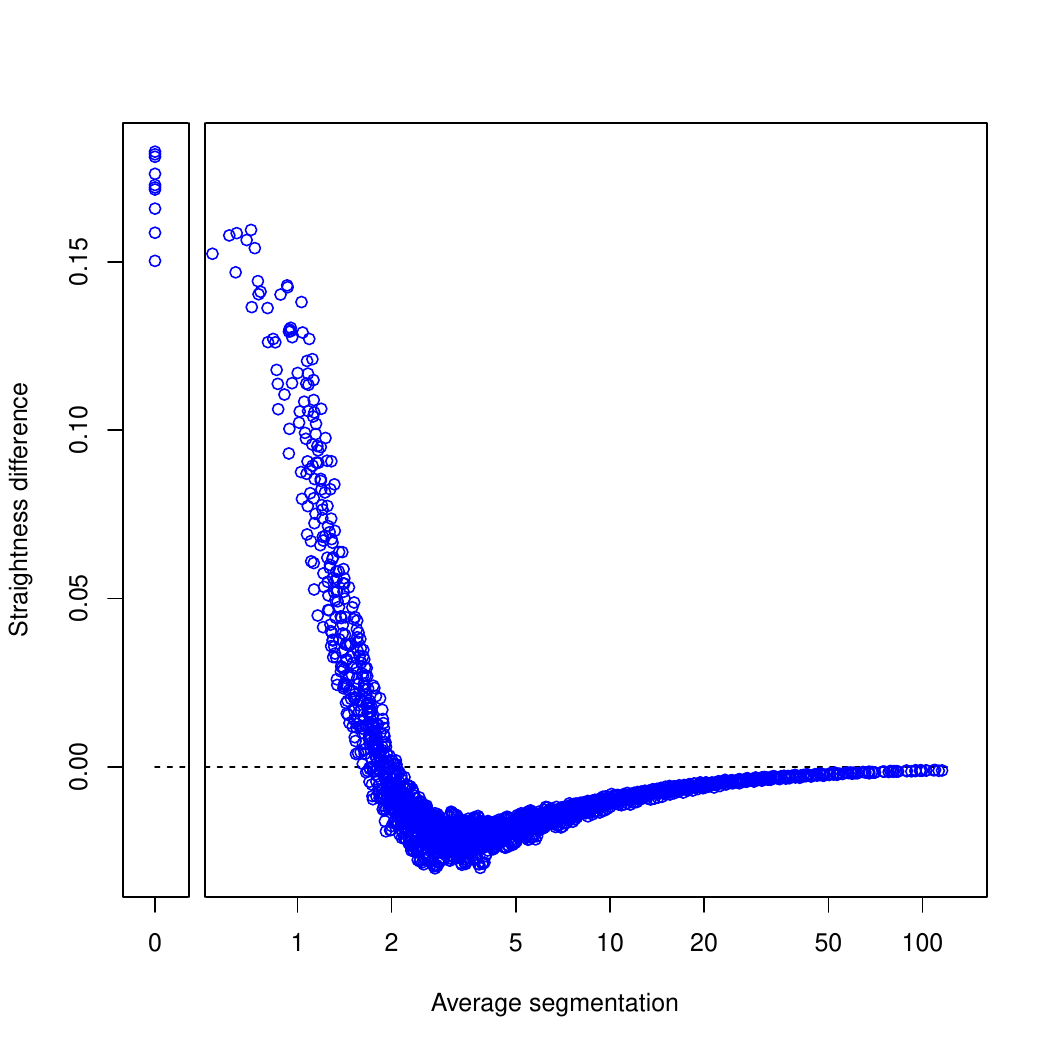}
    \caption{Difference between the discrete approximation and the continuous Straightness, on random planar graphs. The displayed values correspond to individual vertices ($\sigma_\theta(u) - S_G(u)$, left) and whole graphs ($\sigma_\theta(G) - S_G(G)$, right), for all $10$ runs, as functions of the average edge segmentation $\theta$.}
    \label{fig:ExpRandplanarDiff}
\end{figure}

Figure~\ref{fig:ExpRandplanarDiff} shows the results obtained for all $10$ runs (and not a single vertex or graph like before), but under a slightly different form. The $x$ axis has the same meaning as in Figure~\ref{fig:ExpRandplanarStraightness}, but the $y$ axis now represents the \textit{difference} between the discrete approximation and the continuous average Straightness. All vertices are represented in the left plot by their corresponding $\sigma_\theta(u) - S_G(u)$ values, and all graphs in the right one by their $\sigma_\theta(G) - S_G(G)$ values. In both plots, the dashed line materializes a zero difference.    

Both plots confirm the trend observed on the individual vertex and graph plots from Figure \ref{fig:ExpRandplanarStraightness}. Intuitively, when considering Figure~\ref{fig:ExpRandplanarDiff}, the discrete approximation seems to become fairly close to the continuous average when the edges are split in $40$--$50$ pieces or more. However, it seems difficult to generalize this threshold, because Figure~\ref{fig:ExpRandplanarDiff} shows a considerable variance among the considered vertices and graphs, even though they were generated using the same procedure and parameters. This difference could be even stronger with graphs of completely different origins. To summarize, the discrete approximation is unreliable for low average segmentation values, and it is hard to determine an ideal threshold for which it becomes a good approximation of the continuous average.
\begin{figure}[ht!]
	\centering
	\includegraphics[width=0.49\textwidth]{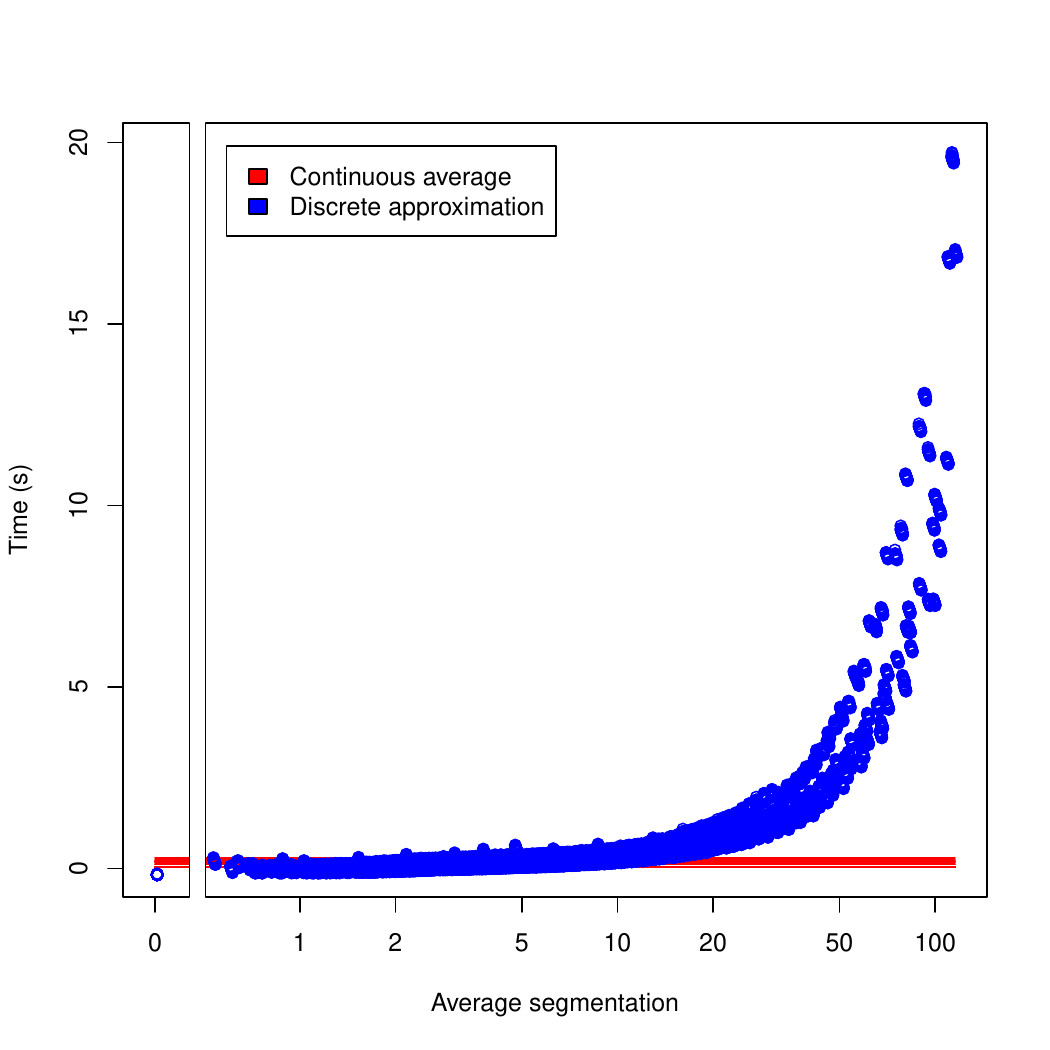}
	\hfill
	\includegraphics[width=0.49\textwidth]{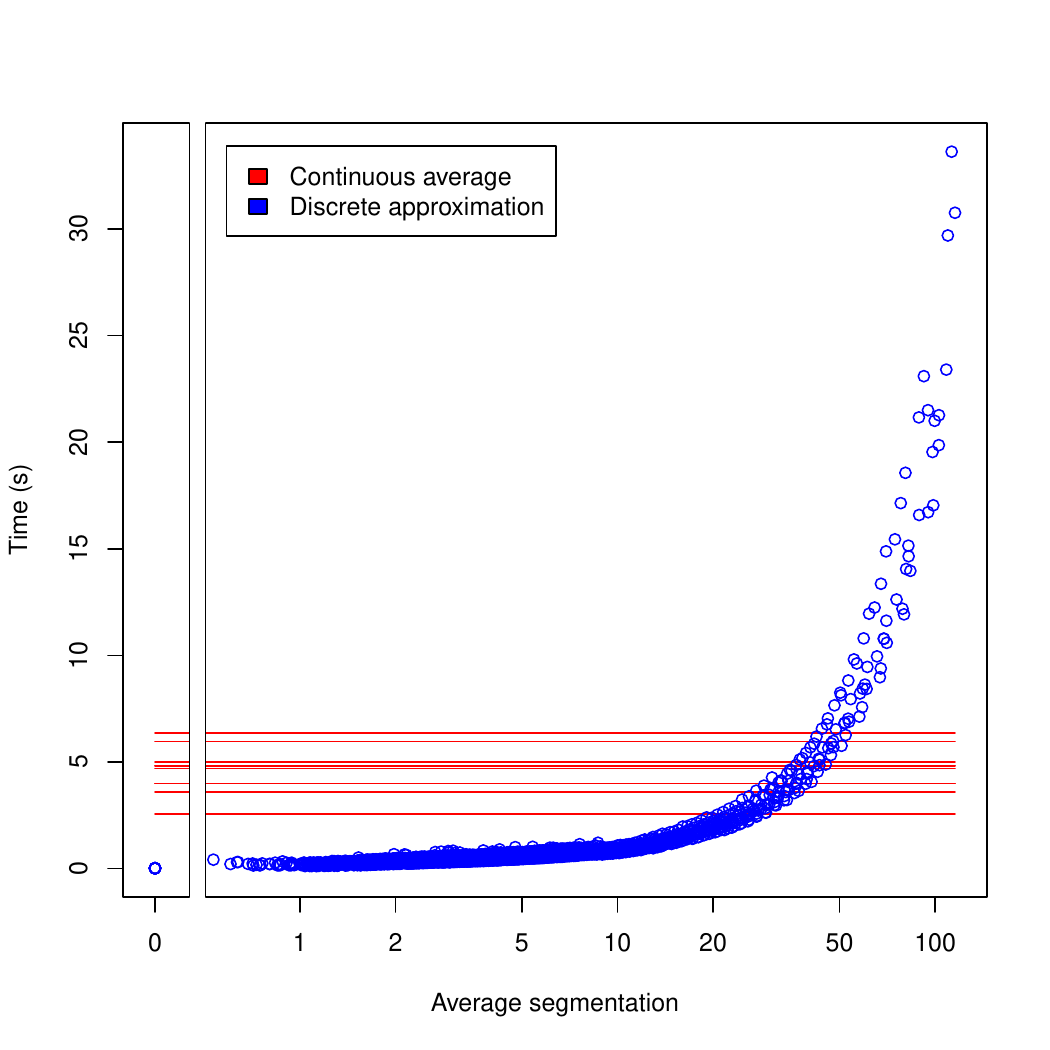}
    \caption{Comparison of the processing times required by the continuous average Straightness (red) and its discrete approximation (blue), on random planar graphs. The displayed values represent individual vertices (left) and whole graphs (right), as functions of the average edge segmentation $\theta$, for all $10$ runs.}
    \label{fig:ExpRandplanarTime}
\end{figure}

\paragraph{Computational Time} 
Figure \ref{fig:ExpRandplanarTime} displays the computational times measured when processing the results shown in Figure~\ref{fig:ExpRandplanarDiff}: individual vertices on the left, whole graphs on the right, all $10$ runs at once. Like before, the red lines represent the continuous average Straightness whereas the blue dots are the discrete approximations. The times are expressed in seconds, and were obtained using a plain desktop Intel i5 3Ghz 8GB machine running Ubuntu 16.10. When dealing with individual vertices ($S_G(u)$ and $\sigma_\theta(u)$, left plot), processing the continuous average is generally much faster. The discrete approximation is actually faster when dealing with the original graph, and when the discretization is very rough, but as mentioned before, this leads to very poor estimations of the continuous average. When reaching the previously established intuitive threshold of a $40$--$50$ average edge segmentation, the processing times are already much higher for the discrete approximation.

The results are not as favorable to the continuous average when considering the whole graph ($S_G(G)$ and $\sigma_\theta(G)$, right plot). As explained in Section~\ref{sec:AlgorithmicComplexity}, this is due to the fact the auxiliary function $f$ (from Definition~\ref{def:AuxiliaryFunction}) is integrated analytically to get $S_G(u)$, whereas we need to use a numerical, slower approach when computing $S_G(G)$. Consequently, the processing times reached by the discrete averages are comparable to those of the continuous ones when reaching the $40$--$50$ threshold. However, additional results obtained on real-world road networks using this threshold (cf. Table~\ref{tab:RealWorldPerf}) show that, in more realistic situations, the processing time of the discrete approximation increases much faster. Moreover, the processing time of the continuous average Straightness could be made shorter by taking advantage of a faster tool to perform the numerical integration, rather than the default R function we used in this experiment.
\begin{figure}[ht!]
	\centering
	\includegraphics[width=0.49\textwidth]{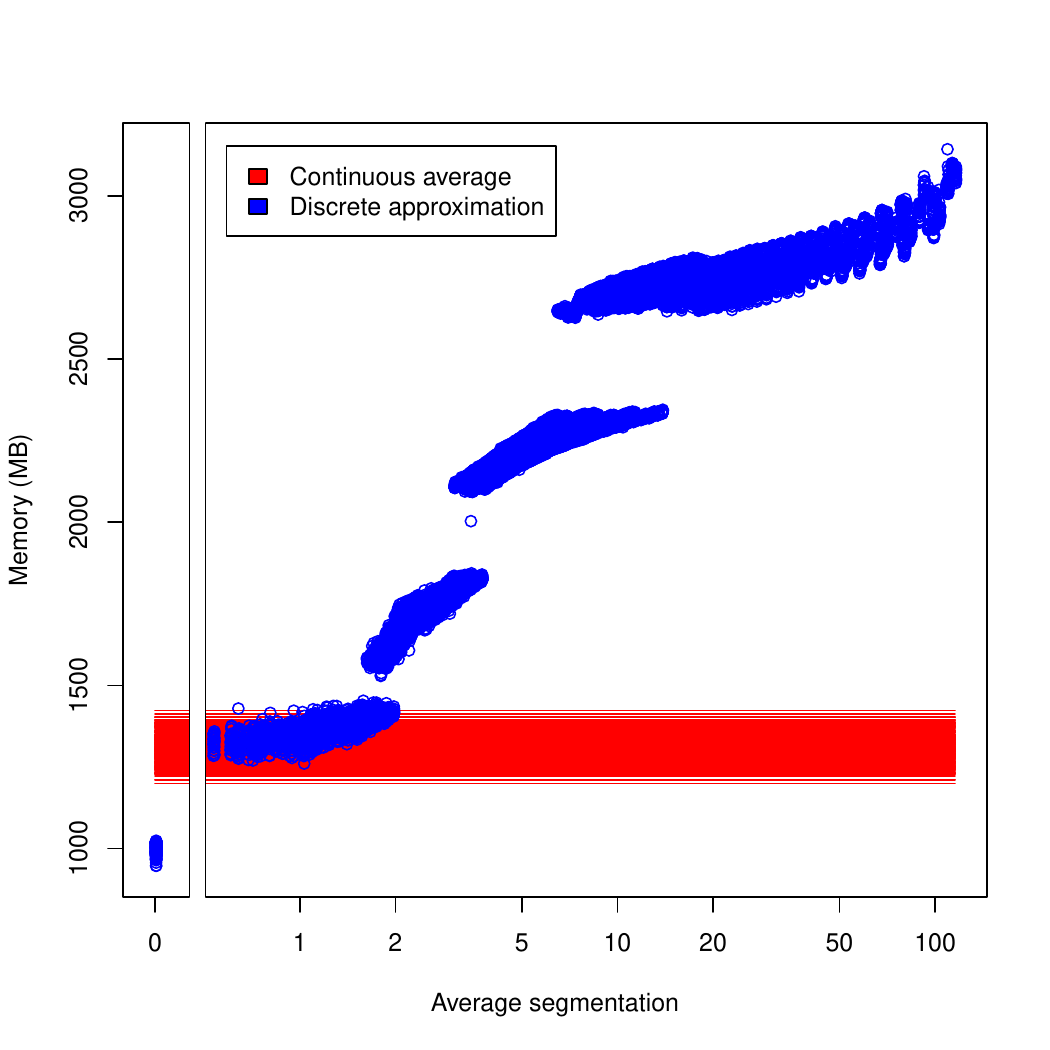}
	\hfill
	\includegraphics[width=0.49\textwidth]{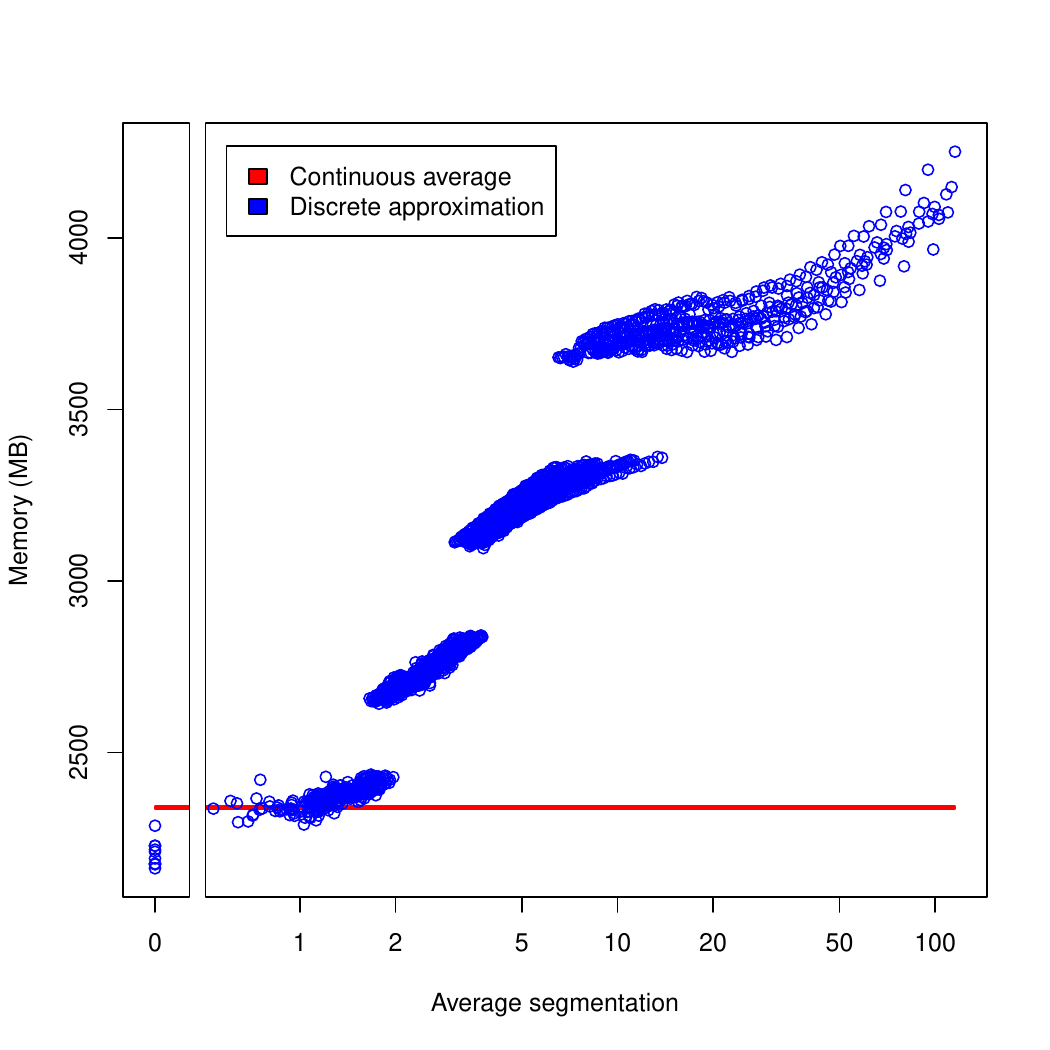}
    \caption{Comparison of the memory used when processing the continuous average Straightness (red) and its discrete approximation (blue), on random planar graphs, for all $10$ runs. The layout is similar to that of Figure \ref{fig:ExpRandplanarTime}.}
    \label{fig:ExpRandplanarMemory}
\end{figure}

\paragraph{Memory Usage} 
Figure \ref{fig:ExpRandplanarMemory} displays the memory usage, expressed in MB, for the same data than Figures~\ref{fig:ExpRandplanarDiff} and~\ref{fig:ExpRandplanarTime}: all $10$ runs at once, individual vertices on the left plot and whole graphs on the right one, continuous average Straightness in red and discrete approximations in blue. Each plot presents three discontinuities for the discrete approximation, which we suppose to be caused by the garbage collector of the R environment: it is automatically triggered, and the user has no control over it.

The comparison is clearly in favor of the continuous average Straightness, both for single vertices and whole graphs. Its memory usage is much lower than that of the discrete approximation, long before the $40$--$50$ average segmentation limit we previously identified. This is an important point, since the discrete approximation could quickly become intractable in practice on a computer with insufficient memory. This is mainly due to the large number of vertices and edges introduced by the edge discretization process described in Section~\ref{sec:DiscreteApproximations}.

\subsection{Behavior of the Measures}
\label{sec:ComparisonNode}
We now study how the various versions of the average Straightness behave over regular graphs, before switching to non-regular and random graphs, and finally real-world road networks.
\begin{figure}[ht!]
	\centering
    \makebox[0.9cm][c]{}
	\hfill
    \makebox[4.4cm][c]{\footnotesize $S(v_1,v_i)$}
	\hfill
    \makebox[4.4cm][c]{\footnotesize $S_{v_1v_2}(v_i)$}
	\hfill
    \makebox[4.4cm][c]{\footnotesize $S_{v_iv_j}(v_1)$}
    
	\includegraphics[height=4.4cm]{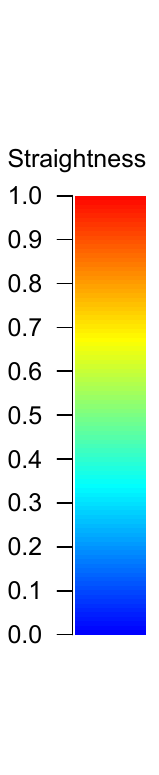}
	\hfill
	\includegraphics[height=4.4cm]{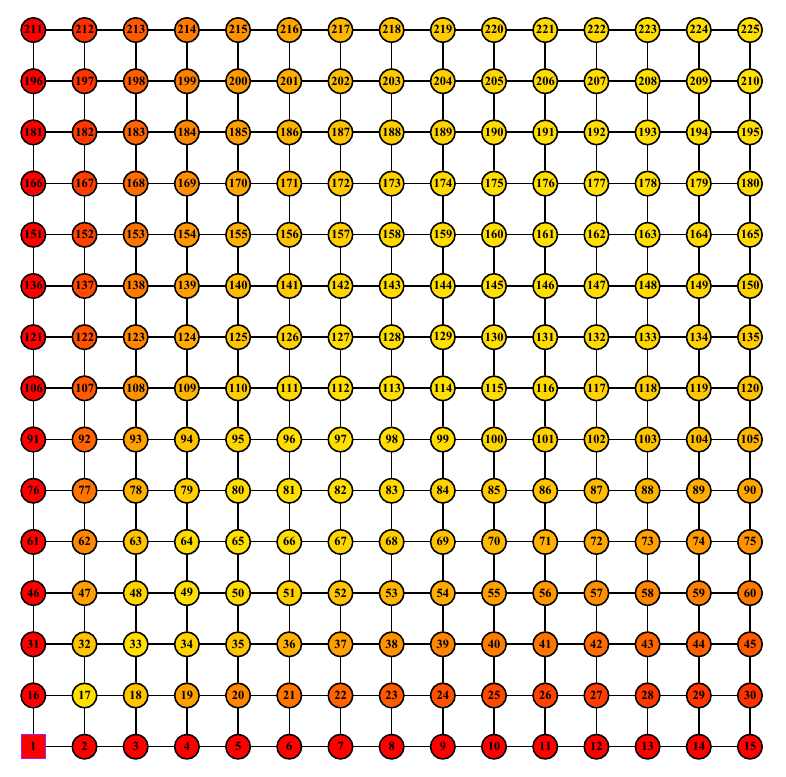}
	\hfill
	\includegraphics[height=4.4cm]{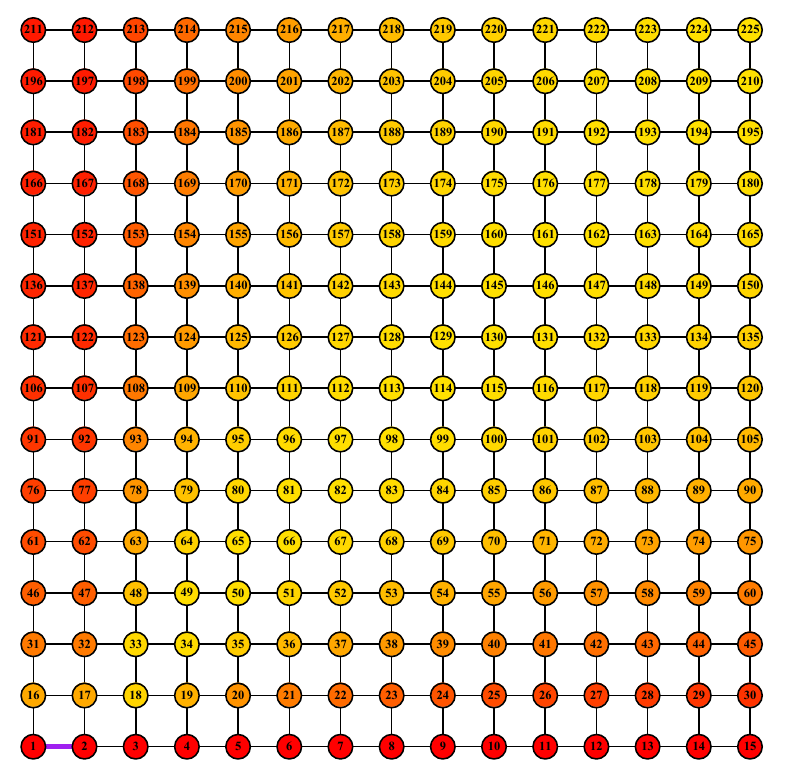}
	\hfill
	\includegraphics[height=4.4cm]{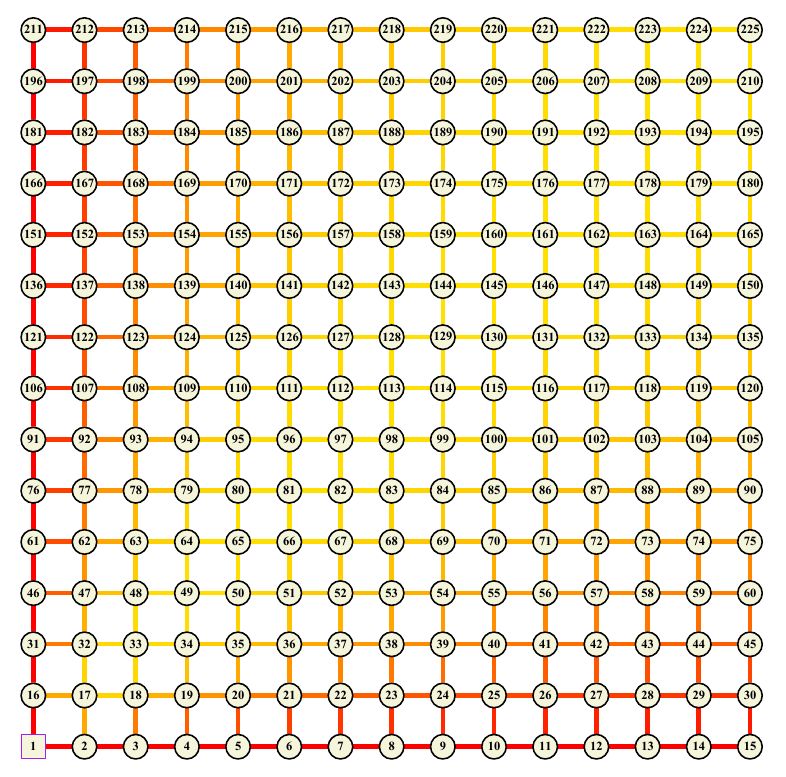}
    \caption{Straightness obtained on a square regular graph, represented through a color gradient. In the left graph, each vertex $v_i \in V$ is colored depending on $S(v_1,v_i)$, the vertex-to-vertex Straightness between this vertex $v_i$ and the fixed vertex $v_1$ (bottom-left vertex in the graph). The center graph displays $S_{v_1v_2}(v_i)$, the continuous average Straightness between $v_i$ and the fixed edge $(v_1,v_2)$ in a similar fashion. In the right graph, each edge $(v_i,v_j) \in E$ is colored depending on $S_{v_iv_j}(v_1)$, the continuous average Straightness between the fixed vertex $v_1$ and this edge $(v_i,v_j)$. The fixed vertices and edge are represented in purple.}
    \label{fig:Squares}
\end{figure}

Figure~\ref{fig:Squares} displays some results obtained for square grids, using a color gradient to represent the Straightness of vertices and edges. The left graph displays $S(v_1,v_i)$, i.e. the traditional Straightness between two vertices (so, \textit{not} an average value). Notice that the first vertex is fixed to $v_1$, which is located at the bottom-left corner of the graph and represented in purple. So, the color of a given vertex v$_i$ corresponds to its Straightness with $v_1$, as traditionally defined (Definition~\ref{def:Straightness}). The center and right graphs show $S_{v_1v_2}(v_i)$ and $S_{v_iv_j}(v_1)$, respectively, i.e. they both represent the continuous average Straightness between a vertex and an edge, but under two different forms. In the former, the edge is fixed to $(v_1,v_2)$ (represented in purple) and each vertex is colored depending to its average Straightness with this edge. In the latter, the vertex is fixed to $v_1$ (also in purple) and each edge is colored depending on its average Straightness with $v_1$.

The left graph illustrates the behavior of the traditional Straightness on such a grid: the vertices located on the same column or row as $v_1$ have a maximal Straightness, whereas it gets lower and lower when getting closer to the diagonal. This is due to the fact there is no diagonal edges in this graph (only vertical and horizontal ones), so being on the diagonal means making more detours. As expected, the continuous average Straightness exhibit a similar behavior. For the center graph, which focuses on edge $(v_1,v_2)$, the difference is that the orientation of the edges affects the obtained values. For instance, $v_{16}$ gets a lower value than in the left graph, whereas $v_{17}$ gets a higher value. This is because $v_1$, $v_2$ and $v_{16}$ are not collinear, so going from a point on $(v_1,v_2)$ to a point on $(v_1,v_{16})$ is not a straight path (unlike going from $v_1$ to $v_{16}$, which is the path considered in the left graph). For the right graph, the difference with the traditional Straightness (left graph) is more visible: it shows the average Straightness of edges, instead of vertices. There is not much difference, in terms of color distribution, with the left graph, due to the regular nature of the graph. However, we will see later that this remark does not necessarily apply in general: the continuous average Straightness of an edge is not necessary similar to those of its end-vertices.

Figure~\ref{fig:Squares} also allows highlighting an important property of the Straightness: it tends to increase with the Euclidean distance to the vertex/edge of interest. Indeed, when considering vertices farther and farther from the vertex/edge of interest, for the ratio defining the Straightness (Definition~\ref{def:Straightness}) to be constant, the detour caused when following the graph edges must grow proportionally to the Euclidean distance, which is often not the case (it grows slower). This is visible on the figure: in the center graph for instance, $v_{30}$ is located one hop away from the row containing the edge of interest $(v_1,v_2)$, but its Straithgness is higher than that of $v_{17}$ which is also located one hop from the row (but much closer from the edge of interest, in terms of Euclidean distance).
\begin{figure}[ht!]
	\centering
    \makebox[0.9cm][c]{}
	\hfill
    \makebox[4.4cm][c]{\footnotesize $S_G(v_i,v_j)$}
	\hfill
    \makebox[4.4cm][c]{\footnotesize $S_G(v_i,v_j)$}
	\hfill
    \makebox[4.4cm][c]{\footnotesize $S_G(v_i,v_j)$}
    
	\includegraphics[height=4.4cm]{legend.pdf}
	\hfill
	\includegraphics[height=4.4cm]{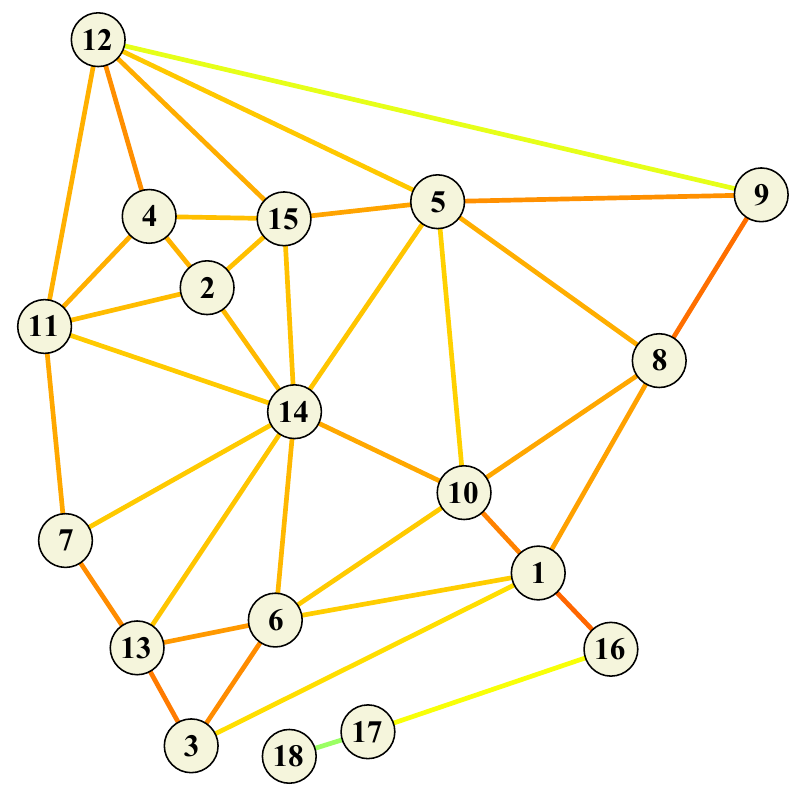}
	\hfill
	\includegraphics[height=4.4cm]{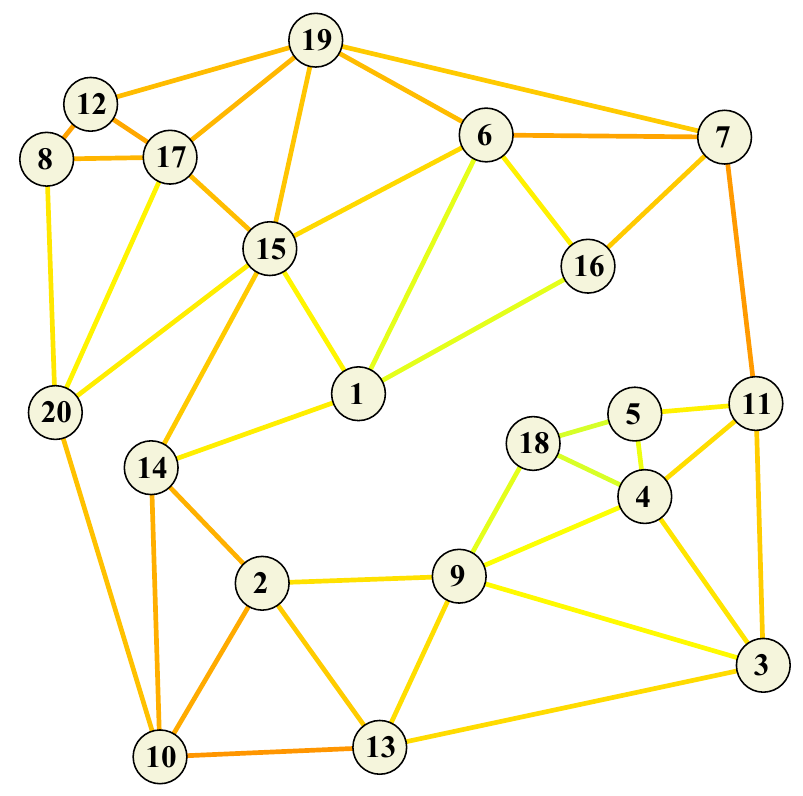}
	\hfill
	\includegraphics[height=4.4cm]{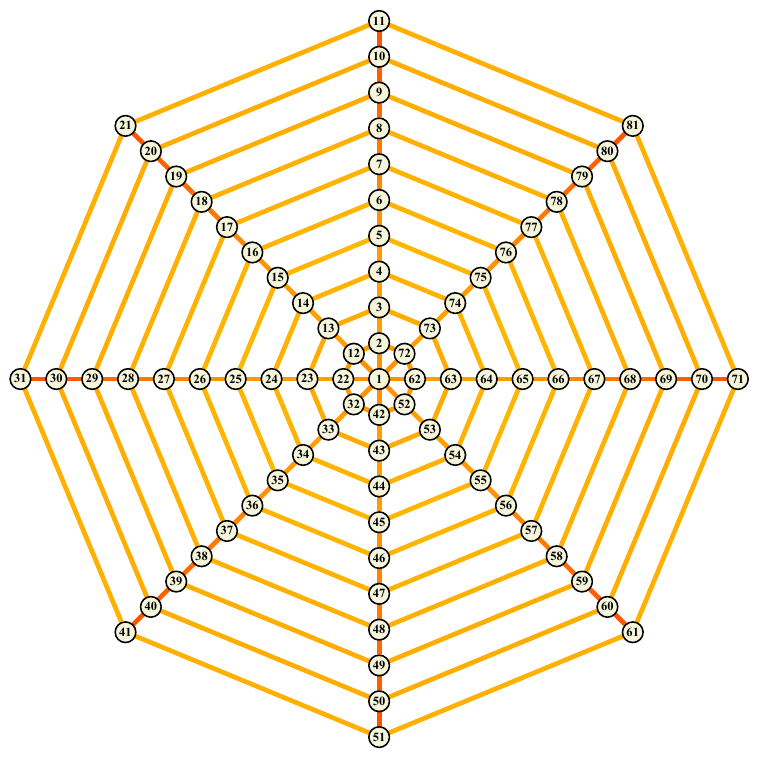}
    \caption{Straightness obtained on two types of non-regular graphs: heterogeneous (left and center) and radio-concentric (right). The color gradient corresponds to $S_G(v_i,v_j)$, the continuous average Straightness between each edge and the rest of the graph.}
    \label{fig:HeteroGrph}
\end{figure}

Of course, this property is present in the original, vertex-to-vertex, Straightness (as shown by the left graph of Figure~\ref{fig:Squares}). But since we consider averages over whole edges in the measures we propose, this distance effect is coupled to another factor: the length of the considered edges. A long edge is more likely to get a low Straightness (and vice versa), since reaching its most inner points requires a longer detour. This does not appear on the regular graphs presented in the previous figures, because their edges are all of relatively similar lengths. However, it shows in the non-regular graphs appearing in Figure~\ref{fig:HeteroGrph}, e.g. edge $(v_9,v_{12})$ in the left graph. This is not to say that only long edges have low Straightness though, as illustrated by edge $(v_{17},v_{18})$ in the same graph, nor that all long edges necessarily have low Straightness, e.g. the peripheral spires in the radioconcentric graph (right graph). 
\begin{figure}[ht!]
	\centering
    \makebox[0.9cm][c]{}
	\hfill
    \makebox[4.4cm][c]{\footnotesize $\sigma(v_i)$}
	\hfill
    \makebox[4.4cm][c]{\footnotesize $S_G(v_i)$}
	\hfill
    \makebox[4.4cm][c]{}
    
    \vspace{-0.5cm}
	\includegraphics[height=4.4cm]{legend.pdf}
	\hfill
	\includegraphics[height=4.4cm]{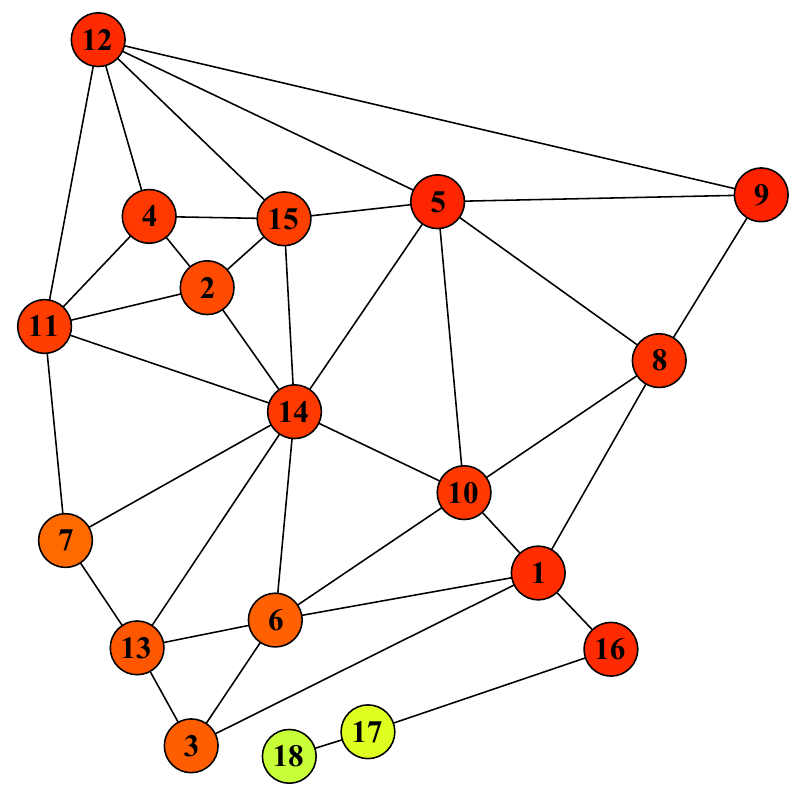}
	\hfill
	\includegraphics[height=4.4cm]{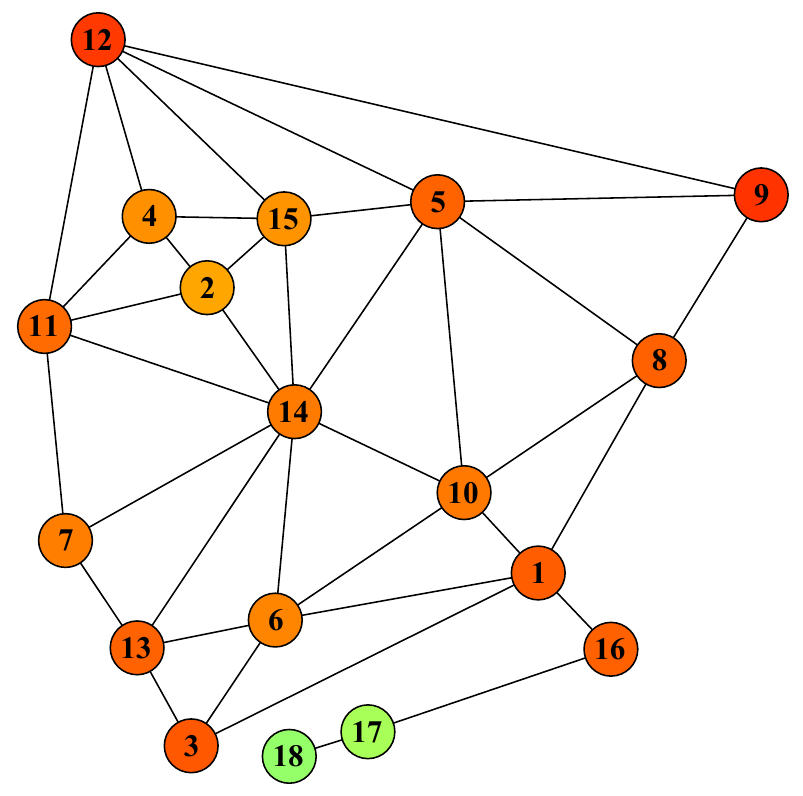}
	\hfill
	\includegraphics[height=4.4cm]{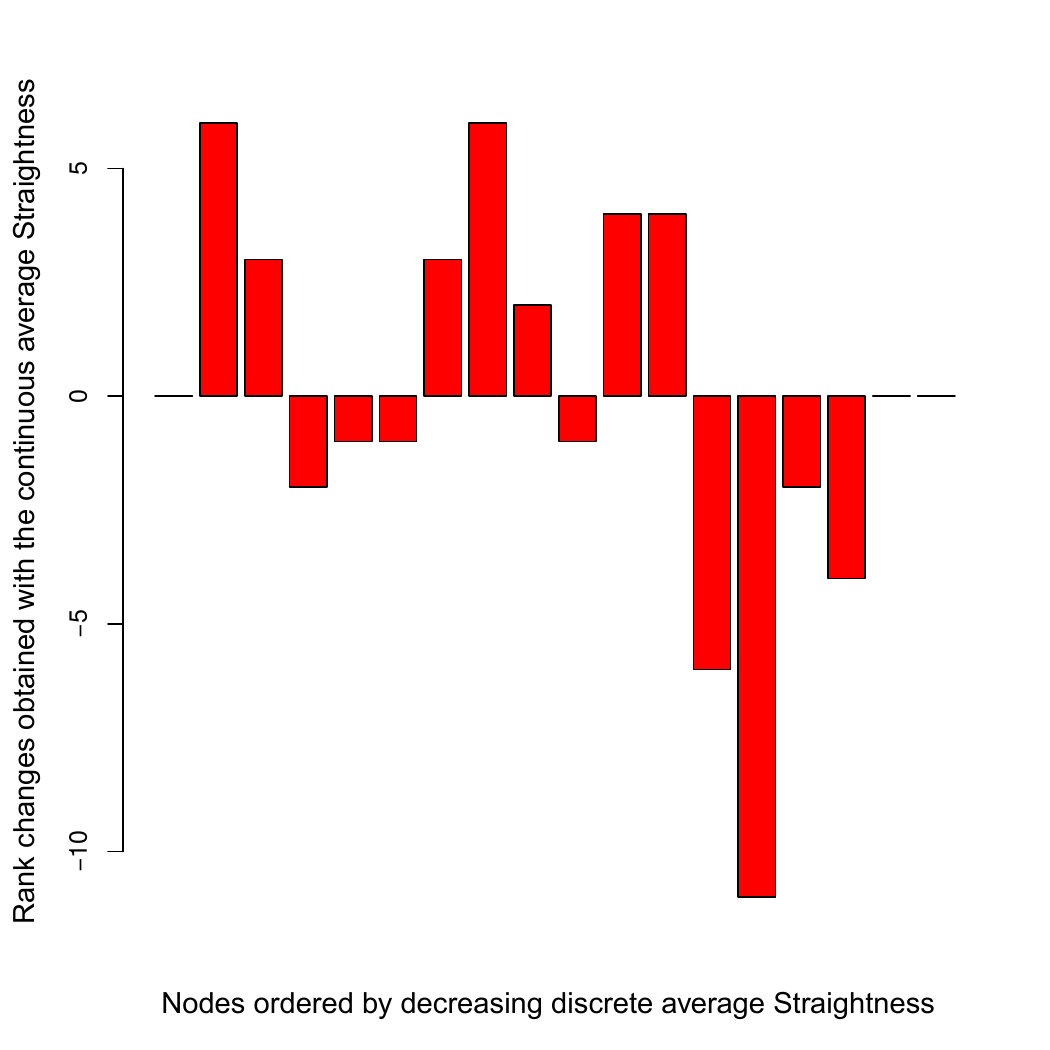}
    \caption{Focus on the left graph of Figure~\ref{fig:HeteroGrph}. The color gradient corresponds to $\sigma(v_i)$, the vertex-to-vertex average Straightness (left), and $S_G(v_i)$, the continuous average Straightness (center) between each vertex and the rest of the graph. The right plot represents the changes in vertex ranking when comparing $\sigma(v_i)$ and $S_G(v_i)$. Each bar corresponds to a specific vertex, and they are ordered in terms of decreasing $\sigma(v_i)$. The bar height displays how the vertex rank changes when considering $S_G(v_i)$ instead of $\sigma(v_i)$.}
    \label{fig:RandomGrph}
\end{figure}

Figures~\ref{fig:RandomGrph} and~\ref{fig:RadioGrph} focus on the left and right graphs from Figure \ref{fig:HeteroGrph}, respectively. In both of them, The left and center graphs display $\sigma(v_i)$ and $S_G(v_i)$, the classic vertex-to-vertex average Straightness and the continuous average Straightness between each vertex and the rest of the graph, respectively. Two points are worth mentioning. First, the continuous average Straightness of an edge and its attached vertices are not necessarily similar. This was not obvious based on Figure~\ref{fig:Squares}, but is much clearer when considering the heterogeneous graphs. For instance, in the left graph of Figure~\ref{fig:HeteroGrph}, we can see that $S_G(v_9,v_{12})$, the continuous average Straightness between edge $(v_9,v_{12})$ and the rest of the graph, is under $0.7$. In  the center graph of Figure~\ref{fig:RandomGrph}, the values observed for $S_G(v_9)$ and $S_G(v_{12})$, the continuous average Straightness between these vertices and the rest of the graph, is clearly higher, reaching approximately $0.9$.
\begin{figure}[ht!]
	\centering
    \makebox[0.9cm][c]{}
	\hfill
    \makebox[4.4cm][c]{\footnotesize $\sigma(v_i)$}
	\hfill
    \makebox[4.4cm][c]{\footnotesize $S_G(v_i)$}
	\hfill
    \makebox[4.4cm][c]{}
    
	\includegraphics[height=4.4cm]{legend.pdf}
	\hfill
	\includegraphics[height=4.4cm]{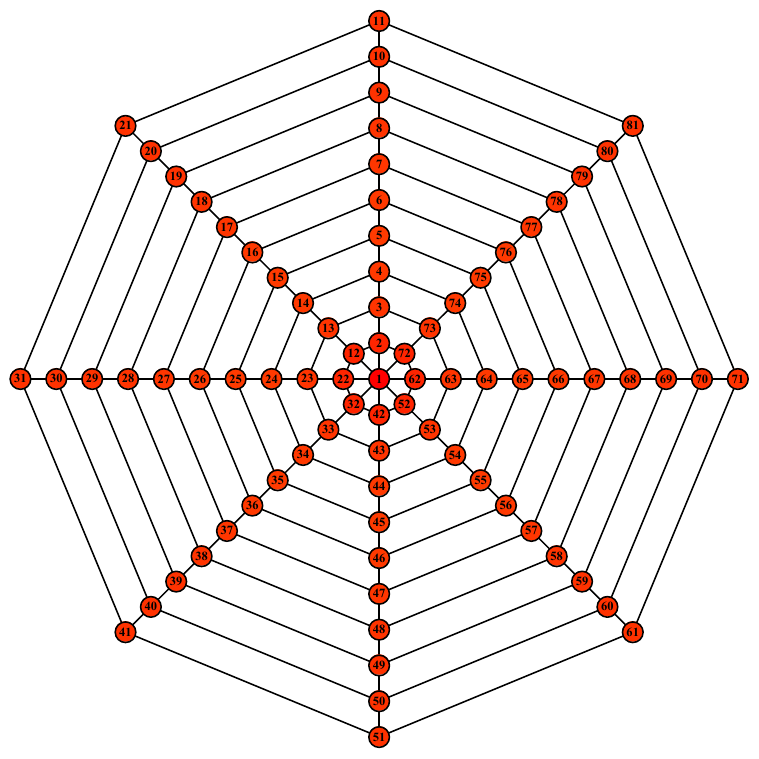}
	\hfill
	\includegraphics[height=4.4cm]{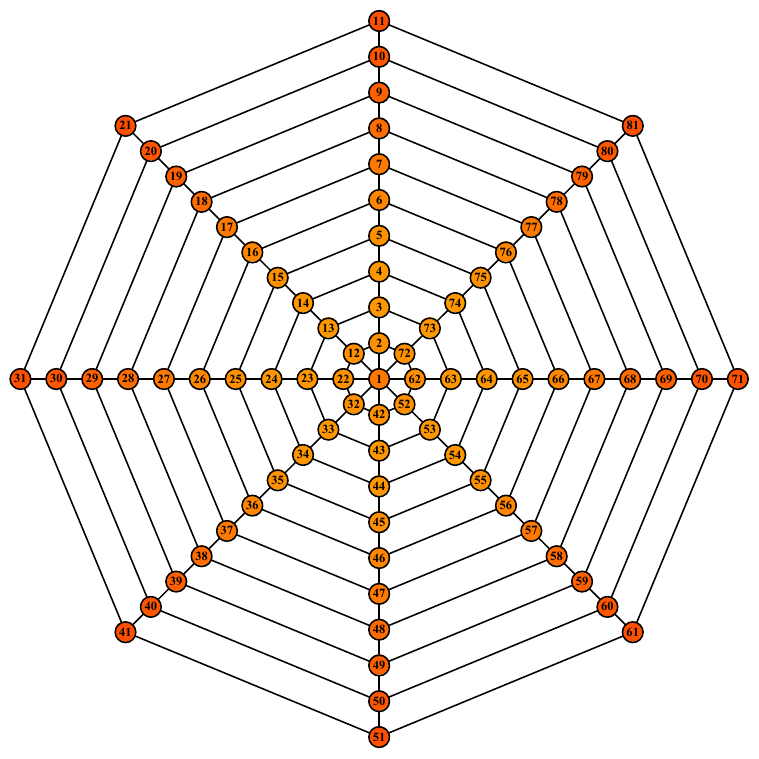}
	\hfill
	\includegraphics[height=4.4cm]{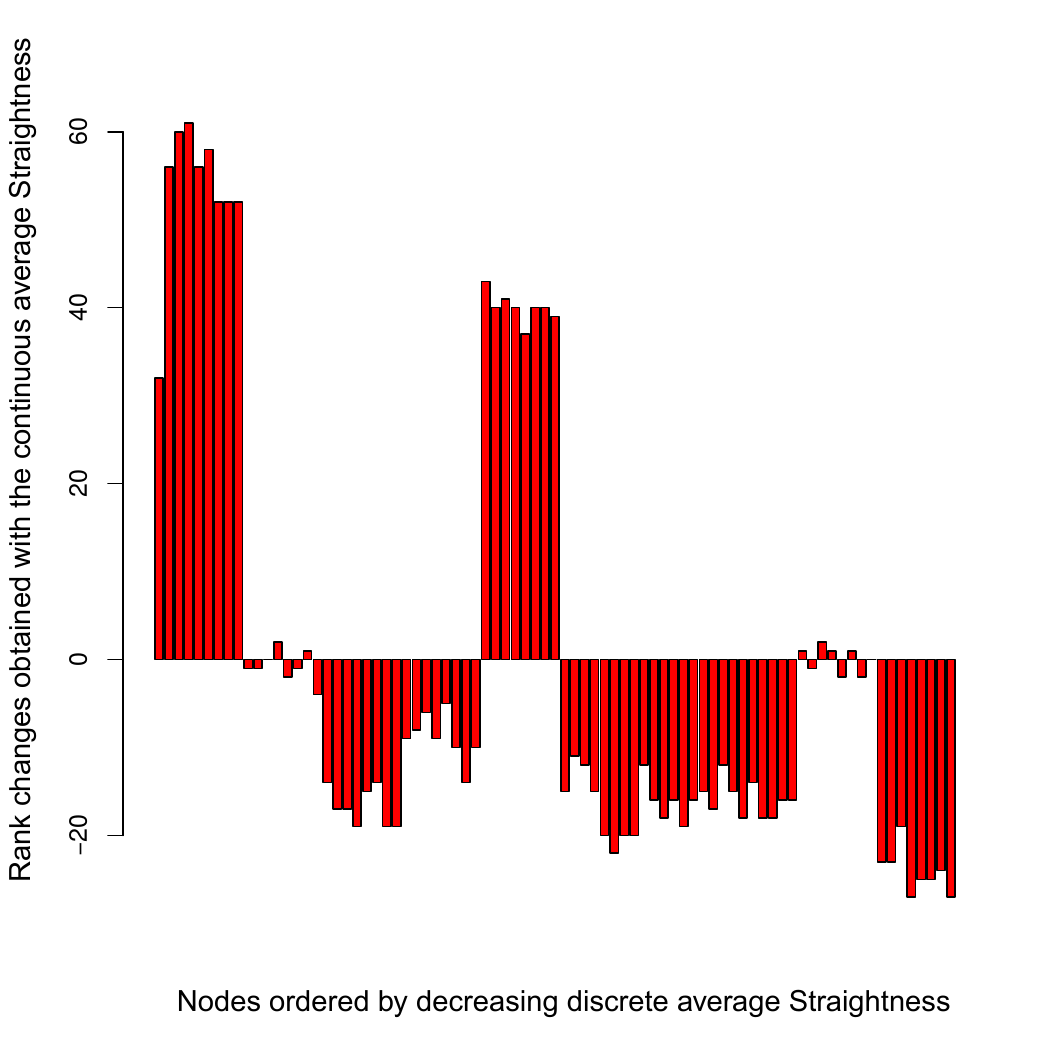}
    \caption{Focus on the right graph of Figure~\ref{fig:HeteroGrph}. The layout and color code are similar to those of Figure~\ref{fig:RandomGrph}.}
    \label{fig:RadioGrph}
\end{figure}

Second, we observe that the values obtained with $\sigma(v_i)$ are generally larger, which is consistent with what we previously noticed during the performance study (Figure~\ref{fig:ExpRandplanarStraightness}), but was not obvious on Figure~\ref{fig:Squares}. However, the continuous average is not just a monotic function of its discrete counterpart. For instance, in Figure~\ref{fig:RandomGrph}, $\sigma(v_{15}) > \sigma(v_3)$, but $S_G(v_{15}) < S_G(v_3)$. In Figure~\ref{fig:RadioGrph}, the vertex with the highest average Straightness is the center when considering $\sigma(v_i)$, but the peripheral vertices reach a higher value than the center with $S_G(v_i)$. This is due to the previous observation, regarding the fact the continuous average Straightness is affected by the edge length. In the case of the radioconcentric graph, the peripheral vertices have the best access to a number of peripheral spires (in addition to certain radii), which happen to be the longest edges in the graph. 

This aspect is studied further through the right plots in Figures~\ref{fig:RandomGrph} and~\ref{fig:RadioGrph}, which both display the changes in vertex ranks when switching from $\sigma(v_i)$ to $S_G(v_i)$. Each bar represents a vertex, and the vertices are horizontally sorted by decreasing $\sigma(v_i)$, i.e. the top (resp. bottom) vertex in terms of $\sigma(v_i)$ is located on the left (resp. right) side of the plot. The $y$ axis corresponds to the rank changes undergone by the vertices when considering $S_G(v_i)$ instead of $\sigma(v_i)$. For a given vertex $v_i$, a positive value means it is better ranked when considering $\sigma(v_i)$ than $S_G(v_i)$, whereas a negative value means the opposite. On the heterogeneous graph from Figure~\ref{fig:RandomGrph}, one can see that the most extreme vertices in terms of $\sigma(v_i)$ are ranked similarly in terms of $S_G(v_i)$, but that there are important changes from most other vertices. On the radioconcentric graph from Figure~\ref{fig:RadioGrph}, it appears that the vertices best ranked according to $\sigma(v_i)$ (left-side bars), i.e. the most central vertices, are the worst ranked in terms of $S_G(v_i)$ (large positive values). On the contrary, the worst ranked vertices according to $\sigma(v_i)$ (right-side bars) get better ranks with $S_G(v_i)$ (large negative values).

\begin{table}[!th]
    \caption{Main topological properties of the considered real-world road networks, and corresponding processing times. These are indicated for both $S_G(u)$, the continuous average Straightness between a vertex and the rest of the graph, and $\sigma_\theta(u)$, its discrete approximation. For both measures, the first specified value is an average over a $25$-vertex sample, and the second one (between parenthesis) is the associated standard deviation. Both are expressed in seconds.}
	\label{tab:RealWorldPerf}
	\centering
	\begin{tabular*}{\textwidth}{@{\extracolsep{\fill}}l@{}r@{}r@{}r@{}r@{}l@{}r@{}l@{}}
		\hline
		City & Vertex count & Edge count & Density & \multicolumn{4}{c}{Processing time} \\
		 & $n$ & $m$ & & \multicolumn{2}{c}{For $\sigma_{\theta}(u)$} & \multicolumn{2}{c}{For $S_G(u)$} \\
		\hline
		Abidjan & 2,577 & 2,908 & $8.76\cdot 10^{-4}$ & 490.12 & (0.02) & 4.00 & ~(0.08) \\
		Karlskrona & 2,619 & 2,999 & $8.75\cdot 10^{-4}$ & 672.36 & (0.02) & 4.11 & ~(0.06) \\
		Soustons & 6,803 & 7,351 & $3.18\cdot 10^{-4}$ & 3,889.20 & (0.10) & 18.13 & ~(0.10) \\
		Maastricht & 9,539 & 10,929 & $2.40\cdot 10^{-4}$ & 7,379.10 & (0.08) & 35.84 & ~(0.18) \\
		Trois-Rivi\`eres & 12,001 & 14,014 & $1.95\cdot 10^{-4}$ & 11,504.23 & (0.12) & 55.54 & ~(0.26) \\
		Alice Springs & 17,011 & 17,790 & $1.23\cdot 10^{-4}$ & 16,765.85 & (0.17) & 93.38 & ~(0.56) \\
		Sfax & 17,152 & 19,702 & $1.34\cdot 10^{-4}$ & 21,122.19 & (0.14) & 131.18 & ~(0.42) \\
		Avignon & 19,481 & 21,898 & $1.15\cdot 10^{-4}$ & 25,438.31 & (0.20) & 136.30 & ~(0.77) \\
		Liverpool & 28,739 & 33,424 & $8.09\cdot 10^{-5}$ & 52,773.02 & (0.20) & 313.55 & ~(3.32) \\
		Ljubljana & 30,854 & 34,684 & $7.29\cdot 10^{-5}$ & 60,572.10 & (0.31) & 369.35 & (26.91) \\
		Lisbon & 35,231 & 40,853 & $6.58\cdot 10^{-5}$ & 79,702.34 & (0.44) & 495.99 & ~(9.46) \\
		Dakar & 36,561 & 45,041 & $6.74\cdot 10^{-5}$ & 98,568.74 & (0.40) & 530.90 & ~(7.62) \\
		Hong Kong & 46,145 & 49,559 & $4.65\cdot 10^{-5}$ & 111,724.71 & (0.36) & 790.45 & ~(8.67) \\
		\hline
	\end{tabular*}
\end{table}

Finally, we processed the continuous average Straightness of a selection of $13$ real-world road networks of various sizes and shapes. These were obtained using the OpenStreetMap website\footnote{\url{https://www.openstreetmap.org}}. The processing times are much larger than for the artificially generated networks considered before, due to the size of the real-world networks. Table~\ref{tab:RealWorldPerf} describes the networks (numbers of vertices $n$ and edges $m$, and edge density), and gives the processing times for $S_G(u)$, the continuous average Straightness between a vertex and the rest of the graph, and $\sigma_{\theta}(u)$, its discrete approximation. We used $\theta = 50$, as this is the threshold empirically identified as appropriate to get a reliable estimation (cf. Section~\ref{sec:ComparisonDiscrete}). These measures were processed for a sample of $25$ vertices in each network, and the table therefore displays average values and standard deviations. The durations increase much faster for the discrete approximation. With regards to the obtained Straightness values, the difference between $\sigma_{\theta}(u)$ and $S_G(u)$ is smaller than $10^3$, in average. 
\begin{figure}[ht!]
	\centering
	\raisebox{-0.5\height}{\includegraphics[height=4.3cm]{legend.pdf}}
	\raisebox{-0.5\height}{\includegraphics[height=7cm]{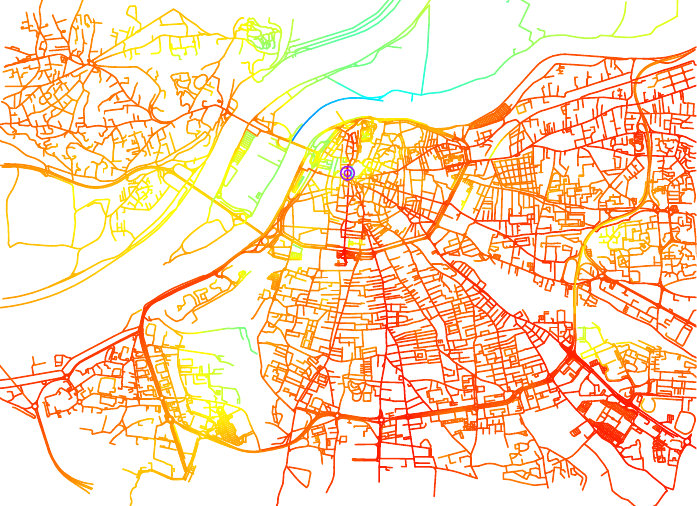}}
    \raisebox{-0.5\height}{
    	\adjustbox{trim={0\width} {0\height} {0.2\width} {0\height},clip}{
    		\includegraphics[height=7cm]{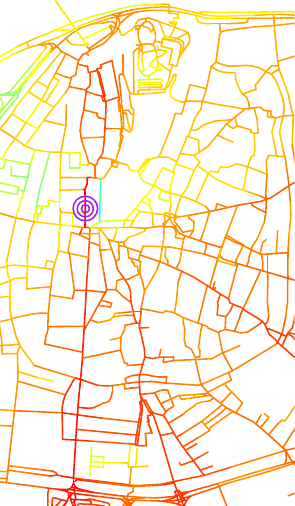}
		}
	}
    \caption{Straightness obtained for the city of Avignon, France: whole city (left) and city center (right). The color gradient corresponds to $S_{uv}(r)$, the average Straightness between a fixed vertex corresponding to the city hall (represented in purple) and each edge of the graph.}
    \label{fig:UrbanGrph}
\end{figure}

As an example, Figure~\ref{fig:UrbanGrph} displays a part of the city of Avignon, France. It represents $S_{uv}(r)$ the continuous average Straightness between a vertex of interest $r$ and each edge of the graph. Vertices are not represented because this would prevent clearly distinguishing the edges color. The right-hand plot is a zoom-in of the left-hand part, focusing on the city center. The fixed vertex corresponds to the city hall, and is represented in purple. The color distribution shows that most of the city can be accessed efficiently when starting from this point, due to its radial structure. However, there are exceptions, which can be explained. 

There are two small yellow patches in the eastern part. Both are located near the city bypass, but not directly connected to it, which explains their low accessibility (Straightness). Moreover, the bottom one neighbors a seemingly empty zone, which is actually a freight train station. The fact that there is no road crossing the railroad in this zone also explains the observed low Straightness. The larger yellow patch on the southwestern part is not very well connected due to a similar reason, except this time the empty zone corresponds to the hight-speed train station (TGV). The lower accessibility in this area is also caused by the presence of an industrial activity sector, a water treatment plant and a recycling plant. On the western side, the accessibility is also globally lower, because the Rh\^one river splits the urban area in two along a southwest to northeast axis. The blue parts correspond to a large, non-urbanized river island, which is only lightly connected to the road network. Finally, there are also traces of low Straightness in the city center, very near to the town hall. This is due to the distance effect we observed before when discussing the Straightness properties: on close destinations, it is more likely for the graph shortest path to be much longer than the distance as the crow flies, much more so than for distant destinations. This is particularly true for the center of Avignon, which is a medieval area with very convoluted streets, as illustrated by the right-hand plot of Figure~\ref{fig:UrbanGrph}. 

\subsection{Final Observations}
\label{sec:FinalObs}
Let us now summarize the main findings of this experimental section. As mentioned in the introduction, there is a conceptual interest to the computation of a \textit{continuous} average of the Straightness, by opposition to the traditional vertex-based discrete approach, because it allows taking into account any itineraries, and not only the vertex-to-vertex ones. The question is to know whether this changes anything in practice.

The first conclusion of our experiments is that, indeed, both approaches lead to different values in practice, as illustrated on random graphs (Figures~\ref{fig:ExpRandplanarStraightness} and~\ref{fig:RandomGrph}), regular graphs (Figure~\ref{fig:Squares}) and non-regular graphs (Figure~\ref{fig:RadioGrph}). Moreover, if the continuous values seem to generally be smaller than the vertex-to-vertex average, we also showed that this was not true in general (Figures~\ref{fig:ExpRandplanarStraightness}, \ref{fig:ExpRandplanarDiff}, \ref{fig:RandomGrph} and~\ref{fig:RadioGrph}). In fact, when sorting vertices in terms of average Straightness, using our continuous approach can lead to largely different rankings compared to the traditional approach, as illustrated by Figures~\ref{fig:RandomGrph} and~\ref{fig:RadioGrph}. 

So in other terms, the discrete and continuous average Straightness values behave differently. We described how these differences can be explained by the fact the continuous average Straightness is affected by the length of the considered edges. We also have shown the relevance of our continuous approach when applied to real-world road networks (Figure~\ref{fig:UrbanGrph}).

The second conclusion of our experiments is that it is not worth using the discrete approach to approximate the continuous average Straightness. Reaching a good approximation level requires performing a very fine discretization of the edges, which in turn strongly increases both the processing time and memory usage (Figures~\ref{fig:ExpRandplanarTime} and~\ref{fig:ExpRandplanarMemory}, and Table~\ref{tab:RealWorldPerf}).

\section{Conclusion}
\label{sec:Conclusion}
In this article, we derived $5$ different versions of the average Straightness for spatial graphs, based on a continuous approach, by contrast with the discrete (vertex-to-vertex) method traditionally adopted in the literature. We then validated them experimentally on a wide variety of graphs, and showed that they are qualitatively different from their discrete counterparts, in the sense they allow to describe the graphs in a different way. Moreover, we also showed that using a discrete approach to approximate the continuous average Straightness is not efficient, since it requires more resources (both time and memory) to obtain a similar precision.

We think this work could be extended mainly in two ways. First, the computational performance could be improved for some of the proposed measures, either by selecting a faster tool to perform the numerical integration, or more straightforwardly by calculating a closed form of the anti-derivative of $F$ (itself the anti-derivative of the auxiliary function $f$ from Definition~\ref{def:AuxiliaryFunction}). However $F$ is itself a convoluted function, and state-of-the-art automatic tools cannot handle it as of now.

Second, and more importantly, the method we proposed is generic and could be applied to any measure defined for spatial graphs: we focused on the Straightness here only because we needed this specific measure for a larger project. In particular, in Section~\ref{sec:StraightnessReformulation} we have expressed both the Euclidean and graph distances under a form suitable for integration, so processing the continuous average of any distance-based measure (for instance the spatial closeness centrality, or the tortuosity measure) would require performing only the extra steps from Section~\ref{sec:ContinuousAverageStraightness}. Of course, depending on the considered measure, these steps might also raise problems of their own, e.g. regarding the integrability of the measure.



\section*{Funding}
This work was supported by the Centre National de la Recherche Scientifique [CNRS PEPS MoMIS Urbi\&{}Orbi] and the University of Avignon [UAPV Projet d'excellence SpiderNet].

\renewcommand*{\bibfont}{\footnotesize}
\addcontentsline{toc}{section}{\refname}
\printbibliography

\end{document}